\documentclass[a4paper,11pt]{article}

\usepackage[english]{babel}
\usepackage[a4paper]{geometry}
\usepackage{enumerate}  
\usepackage{amsmath}
\usepackage{amsthm}
\usepackage{amssymb}
\usepackage{hyperref}   

\usepackage{amssymb}
\usepackage{color}
\usepackage[utf8]{inputenc}
\usepackage{graphicx}
\usepackage{tikz}
\usepackage{verbatim} 

\usepackage{mathabx} 
\usepackage{hyperref}
\usepackage{comment}
\usetikzlibrary{decorations.markings}
\usepackage{bbm}

\def\bR{\mathbb{R}}
\def\bN{\mathbb{N}}
\def\NN{\mathbb{N}}

\def\cU{\mathcal{U}}
\def\cW{\mathcal{W}}
\def\cV{\mathcal{V}}
\def\cF{\mathcal{F}}
\def\cG{\mathcal{G}}
\def\cQ{\mathcal{Q}}
\def\cD{\mathcal{D}}
\def\cJ{\mathcal{J}}
\def\cC{\mathcal{C}}
\def\cL{\mathcal{L}}
\def\cN{\mathcal{N}}
\def\cE{\mathcal{E}}
\def\cK{\mathcal{K}}
\def\cH{\mathcal{H}}

\newcommand{\unit}{1\!\!1}

\newcommand{\exps}{\tau}       
\newcommand{\exph}{\varepsilon}  

\def\eps{\varepsilon}
\def\ph{\varphi}

\def\wt{\widetilde}
\def\wh{\widehat}

\def \pn {\varphi_0 }
\def \whpn{\widehat{\varphi}_0}

\def \an {\mathfrak{a}_{0} }

\def \whpn { \widehat{ \varphi}_0 }

\def \cFpn {\mathcal{F}_{\bot \varphi_0}^{\leq N} }
\def \wt {\widetilde}

\def\tr{\operatorname{tr }}
\def\be{\begin{equation}}
\def\ee{\end{equation}}

\newtheorem{theorem}{Theorem}[section]  
\newtheorem{cor}[theorem]{Corollary}
\newtheorem{prop}[theorem]{Proposition}
\newtheorem{lemma}[theorem]{Lemma}

\numberwithin{equation}{section}

\begin{document}

\title{Bogoliubov Theory for Trapped Bosons \\ in the Gross-Pitaevskii Regime}

\author{Christian Brennecke$^1$, Benjamin Schlein$^2$, Severin Schraven$^3$\footnote{Corresponding author: sschraven@math.ubc.ca} \\
\\
Institute for Applied Mathematics, University of Bonn, \\
Endenicher Allee 60, 53115 Bonn, Germany$^{1}$ \\
\\
Institute of Mathematics, University of Zurich, \\
Winterthurerstrasse 190, 8057 Zurich, Switzerland$^{2,3}$}

\maketitle

\begin{abstract}
We consider systems of $N$ bosons in $\bR^3$, trapped by an external potential. The interaction is repulsive and has a scattering length of the order $N^{-1}$ (Gross-Pitaevskii regime). We determine the ground state energy and the low-energy excitation spectrum up to errors that vanish in the limit $N\to \infty$.  
\end{abstract}

\section{Introduction and Main Results}  

We consider trapped bosons described by the Hamilton operator 
\begin{equation}\label{eq:defHN} H_N = \sum_{j=1}^N \left[ -\Delta_{x_j} + V_\text{ext} (x_j) \right] + \sum_{1\leq i<j\leq N} N^2 V(N(x_i -x_j)) ,
\end{equation}
acting on a dense subspace of $L^2_s (\bR^{3N})$, the subspace of $L^2 (\bR^{3N})$ consisting of functions that are symmetric with respect to permutations of the $N$ particles. Here, $V_\text{ext}\in L^\infty_{loc}(\mathbb{R}^3)$ is a trapping potential with the property that $ V_\text{ext}(x)\to \infty $ for $|x|\to \infty$ (later, we will formulate more precise assumptions on $V_\text{ext}$). As for the interaction $V$, we assume that $V\in L^3(\bR^3)$ and that it is pointwise non-negative, spherically symmetric and compactly supported. 

The scattering length $\an$ of $V$ is defined through the zero-energy scattering equation 
\begin{equation}\label{eq:0en} \left[ - \Delta + \frac{1}{2} V (x)  \right] f (x) = 0 \end{equation}
with the boundary condition $f (x) \to 1$, as $|x| \to \infty$. For $|x|$ large enough, we have
\[ f(x) = 1- \frac{\frak{a}_0 }{|x|} \]
for a constant $\frak{a}_0$ which is determined by the potential $V$. This constant is called the scattering length of $V$. It is straightforward to show that  
		\begin{align*}  8 \pi \frak{a}_0  =  \int_{\mathbb{R}^3} V(x) f(x) dx  \end{align*}
and, by scaling, (\ref{eq:0en}) implies that 
		\[ \left[ - \Delta + \frac{1}{2} N^2 V (Nx) \right] f (Nx) = 0. \]
This means that the scattering length of the interaction $N^2 V (N.)$ appearing in (\ref{eq:defHN}) is given by $\frak{a}_0 /N$, which characterises the Gross-Pitaeavskii regime.  

From \cite{LSY, LS2,NRS} it is well-known that the ground state energy $E_N$ of the Hamilton operator \eqref{eq:defHN} satisfies		
		\begin{equation}\label{eq:LSY} \lim_{N\to \infty} \frac{E_N}N =  \inf_{\psi\in H^1(\bR^3): \|\psi\|_2=1}\cE_{GP}(\psi), 
		\end{equation}
where $\cE_{GP}$ denotes the Gross-Pitaevskii functional, defined by
		\begin{equation} \label{eq:defGPfunctional}
		\mathcal{E}_{GP}(\psi) = \int_{\bR^3} \left(  \vert \nabla \psi(x) \vert^2 + V_\text{ext}(x) \vert \psi(x)\vert^2 + 4\pi\an \vert \psi(x)\vert^4 \right) dx.
		\end{equation}
We remark that $\cE_{GP}$ admits a unique normalized, strictly positive minimizer which we denote for the rest of this paper by $\pn \in L^2(\bR^3)$. It satisfies the Euler-Lagrange equation
\begin{align*} -\Delta \pn + V_\text{ext} \pn + 8\pi \frak{a}_0 |\pn|^2 \pn = \eps_{GP} \pn \end{align*} 
with the Lagrange multiplier $\eps_{GP} = \cE_{GP} (\pn) + 4\pi \frak{a}_0 \| \pn \|_4^4$. In the Appendix, more precisely in Lemma \ref{thm:gpmin1} and Lemma \ref{thm:gpmin2}, we establish regularity and decay properties of the minimizer $\pn$, which will be useful for our analysis. 

From \cite{LS1, LS2, NRS}, it is also known that every approximate ground state of (\ref{eq:defHN}), ie. every sequence of normalized wave functions $\psi_N \in L^2_s (\bR^{3N})$ satisfying 
\[  \lim_{N \to \infty} \frac{1}{N} \langle \psi_N, H_N \psi_N \rangle = \cE_\text{GP} (\pn) \, ,  \]
exhibits complete Bose-Einstein condensation in $\pn$. In other words, denoting by $\gamma_N^{(1)} =  \operatorname{tr}_{2,\dots,N} |\psi_N\rangle\langle\psi_N|$ the one--particle reduced density associated with 
a sequence of approximate ground states $\psi_N$, we have 
	\begin{equation}\label{eq:BEC0}\lim_{N\to\infty} \langle\pn, \gamma_N^{(1)}\pn\rangle = 1. \end{equation}
		
For translation invariant systems (particles trapped in $\Lambda = [0 ;1]^3$ with periodic boundary conditions), \eqref{eq:LSY} and \eqref{eq:BEC0} have been proved in \cite{BBCS1, BBCS4, ABS, H} to hold with the optimal rate of convergence. This result was generalized in \cite{NNRT} (extending the approach of \cite{BFS,FS}) to trapped systems described by (\ref{eq:defHN}), under the assumption of a sufficiently small scattering length $\mathfrak{a}_0$. In \cite{BSS}, we recently established \eqref{eq:LSY} and \eqref{eq:BEC0} with optimal rate of convergence and with no assumption on the size of $\mathfrak{a}_0$. These bounds (obtained through a generalization of the methods of  \cite{BBCS1, BBCS4}) lead to the a-priori estimates stated below in Theorem \ref{prop:hpN} and Proposition \ref{prop:hpNcubic} and represent a crucial ingredient for our analysis, aimed at understanding the low-energy spectrum of (\ref{eq:defHN}), in the limit $N \to \infty$. 
		
To state our main theorem, we introduce the notation 
\be \label{eq:defHGPE} H_\text{GP} = -\Delta + V_\text{ext} +8\pi \mathfrak{a}_0\pn^2 -\eps_{GP}, \hspace{0.5cm} E = \big(H_\text{GP}^{1/2} (H_\text{GP}+ 16\pi \mathfrak{a}_0\pn^2)H_\text{GP}^{1/2} \big)^{1/2}. \ee 
The operators $H_\text{GP}$ and $E$ have discrete spectrum and in both cases, $\pn$ is the unique, positive and normalized ground state vector with $H_\text{GP}\pn = E \pn=0$.

Let us make our assumptions on the interaction potential $V\in L^3(\bR^3)$ and the trapping potential $V_\text{ext}:\bR^3\to \bR$ more precise. Our standing assumptions throughout the rest of this paper are as follows
		\begin{equation}\label{eq:asmptsVVext}
		\begin{split}
		(1) \;& V\in L^3(\bR^3) , V(x)\geq 0 \text{ for } a.e. \;x\in\bR^3, V(x)=V(y) \text{ if } |x|=|y|,  \\ &V  \text{ compactly supported}, \\
		(2)\; &   V_\text{ext} \in C^1\big(\bR^3; [0;\infty)\big), V_\text{ext}(x) \to \infty \text{ as } |x| \to \infty,\\ 
		& \exists\;  C > 0 \; \forall \;  x,y\in\bR^3: V_\text{ext}(x+y) \leq C (  V_\text{ext}(x)+C)(  V_\text{ext}(y)+C), \\
		& \nabla V_\text{ext} \text{ has at most exponential growth as }  |x|\to \infty, \\
		(3) \;&(H_\text{GP} + 1)^{-3/4-\eps} e^{-\alpha |x|} \text{ is a Hilbert-Schmidt operator, for any $\eps > 0$} \\ &\text{and some $\alpha > 0$}.
		\end{split}
		\end{equation}
It is worth mentioning that the assumption (2) implies that $V_\text{ext}$ grows at most exponentially fast (see \cite[Appendix A]{BSS} for a proof). Combining this with the fact that $\pn$ decays faster than any exponential (see \eqref{eq:expdecaypn}) implies that we can control any $L^p$-norm of $V_\text{ext}\pn$. 
 
Our assumptions are technical and we have not attempted to optimize them. The assumption that $V_\text{ext}\geq 0$ is for computational convenience only and can always be arranged by an energy shift (notice that we do not need to assume $V_\text{ext}\geq 0$ in \cite{BSS}). 


\begin{theorem}\label{thm:main} Let $V$ and $ V_\text{ext}$ be as in \eqref{eq:asmptsVVext}. Then, there exists $\rho>0$ such that, in the limit $N\to \infty$, the ground state energy $E_N$ of $H_N$, defined in \eqref{eq:defHN}, is given by
\be \begin{split}\label{eq:gsenergy}
E_N &= N  \cE_{GP}(\pn) - 4\pi \mathfrak{a}_0 \|\pn\|_4^4 +E_{\text{Bog}}   + \mathcal{O}\big( N^{-\rho}\big),   
\end{split}\ee
where $E_\text{Bog}$ a the finite constant (independent of $N$) given by 
\begin{equation}\label{eq:bog1} 	
\begin{split}
E_\text{Bog} = \frac{1}{2} \lim_{\delta \to 0} &\left\{ \tr_{\perp \pn} \Big[ \Big(H_\text{GP}^{1/2} (H_\text{GP} + 16 \pi \frak{a}_0 \unit_\delta \pn^2 ) H_\text{GP}^{1/2} \Big)^{1/2} - H_\text{GP} - 8\pi \frak{a}_0 \unit_\delta \pn^2 \Big] \right. \\ &\hspace{5.5cm}  \left. + \frac{(8\pi\frak{a}_0)^2}{2} \tr  \Big[ \pn^2 \unit_\delta \frac{1}{-\Delta} \unit_\delta \pn^2 \Big] \right\} \end{split}  
\end{equation}  
with $\unit_\delta$ denoting the approximation of the identity operator with integral kernel $\unit_\delta (x;y) = (2\pi \delta)^{-3/2} e^{-(x-y)^2/2\delta^2}$ (but any reasonable choice of the approximating sequence would do). Here, $\tr_{\bot \pn}$ denotes the trace over the orthogonal complement $\{\pn\}^{\bot}$. 
	
Moreover, the spectrum of $H_N-E_N$ below a threshold $\zeta>0$ (assuming $\zeta \leq CN^{-\rho/5-\eps}$ for some $\eps>0$) consists of eigenvalues of the form 
	\be \label{eq:excitHN-EN}
	\sum_{i = 1}^\infty n_i e_i + \mathcal{O}\big( N^{-\rho} ( 1+ \zeta^5 + N^{-5\rho} \zeta^{30} )\big),
	\ee
	where the $(e_j)_{j\in \mathbb{N}}$ denote the eigenvalues of the operator $E$, defined in \eqref{eq:defHGPE}, and where $n_i \in \mathbb{N}$ with $n_i \neq 0$ for finitely many $i\in \mathbb{N}$.
\end{theorem}

\vspace{0.2cm}

{\it Remark:} In (\ref{eq:bog1}), we cannot take directly the limit $\delta \to 0$, replacing $\unit_\delta$ with the identity, because the resulting operator is not of trace class and the trace would not be well-defined. Nevertheless, we will show in Section \ref{sec:gsenergy} that the limit defining $E_\text{Bog}$ exists and that it is explicitly given by
\be \label{eq:Ebog}
\begin{split}
E_{\emph{Bog}} & = \frac{\kappa^2}{4} (8\pi \mathfrak{a}_0)^2 \tr \left[    \pn^2  (-\Delta+\kappa^2)^{-1} (-\Delta)^{-1}  \pn^2 \right]\\
		&\hspace{0.5cm} + \frac14 (8\pi\mathfrak{a}_0)^2 \left\| \pn    (-\Delta+\kappa^2)^{-1}[\pn^2,-\Delta] (-\Delta+\kappa^2)^{-1/2}   \right\|_\text{HS}^2  \\
		&\hspace{0.5cm} + \frac14 (8\pi\mathfrak{a}_0)^2 \left\| \pn   (-\Delta+\kappa^2)^{-1} \nabla\pn   \right\|_\text{HS}^2 \\
		&\hspace{0.5cm} + \frac{(8\pi \mathfrak{a}_0)^2}{4} \tr \, \Big[ \pn^2 \Big( \frac{1}{-\Delta + \kappa^2} - \frac{1}{H_\text{GP} + \kappa^2} \Big) \pn^2 \Big] \\
		&\hspace{0.5cm} + \frac{(8\pi \mathfrak{a}_0)^2}{4\kappa^2} \| \pn^3 \|^2 + \frac{(8\pi \mathfrak{a}_0)^2}{4}  \Vert (H_\text{GP}+\kappa^2)^{-1/2} Q \pn^3\Vert^2 \\ 
		&\hspace{0.5cm}-\frac{(8\pi \mathfrak{a}_0)^2 \kappa^2}{4} \tr_{\perp \pn} \Big[ \pn^2 Q \frac{1}{H_\text{GP} (H_\text{GP} + \kappa^2)} Q \pn^2 \Big] \\
		&\hspace{0.5cm} - \frac{(8\pi\frak{a}_0)^2}{\pi} \int_0^\infty ds \, \sqrt{s} \, \tr_{\perp \pn}  \frac{H_\text{GP}^{1/2}}{s+ H_\text{GP}^2} \Big[ \Big[ \pn^2 , \frac{H_\text{GP}}{s+ H_\text{GP}^2} \Big] , \pn^2 \Big] \frac{H_\text{GP}^{1/2}}{s+ H_\text{GP}^2} \\
		&\hspace{0.5cm} - \frac{4(8\pi\frak{a}_0)^3}{\pi} \int_0^\infty ds \, \sqrt{s}  \tr_{\perp \pn} \Big[ \frac{1}{s+H_\text{GP}^2} H_\text{GP}^{1/2} \pn^2 H_\text{GP}^{1/2} \Big]^3 \\ &\hspace{6cm} \times \frac{1}{s+H_\text{GP}^{1/2} (H_\text{GP} + 16 \pi \frak{a}_0 \pn^2) H_\text{GP}^{1/2}} 
\end{split} \ee
for any $\kappa > 0$ large enough. In Sect. \ref{sec:gsenergy} we also prove that all contributions on the r.h.s. of (\ref{eq:Ebog}) are finite.

Theorem \ref{thm:main} extends to the Gross-Pitaevskii regime earlier results obtained for systems of bosons in the mean-field limit. The low-energy spectrum of Hamilton operators describing mean-field bosons has been established in the translation invariant case in \cite{Sei,DN} and for systems trapped by external fields in \cite{GS,LNSS}. More precise expansions in powers of $1/N$ for the ground state energy and low energy excitations of mean field Hamiltonians has been recently obtained in \cite{P3,BPS}. 

For translation invariant systems consisting of $N$ bosons moving in the box $\Lambda = [0;1]^{3}$ with periodic boundary conditions, Theorem \ref{thm:main} has been recently shown in \cite{BBCS3}, using a rigorous version of Bogoliubov theory, see \cite{bog}. While the strategy that we are going to use to show Theorem \ref{thm:main} is similar to the one developed in \cite{BBCS3}, the lack of translation invariance makes the proof more complicated. 

Let us briefly explain the main steps. First of all, in Section \ref{sec:fock}, we are going to factor out the Bose-Einstein condensate and restrict our attention on its orthogonal excitations. We do so through a unitary transformation $U_N$, first introduced in \cite{LNSS}, mapping the Hilbert space $L^2_s (\bR^{3N})$ onto the truncated Fock space $\cF^{\leq N}_{\perp\ph_0}$ constructed on the orthogonal complement of the condensate wave function $\ph_0$. This procedure, which can be seen as a rigorous version of Bogoliubov c-number substitution, leads to the excitation Hamiltonian $\cL_N = U_N  H_N U_N^*$. In second quantized form, the operator $\cL_N$ can be written as the sum of a constant and of terms that are linear, quadratic, cubic and quartic in (modified) creation and annihilation operators. In contrast with the translation invariant case, where $\cL_N$ could be easily expressed in momentum space, here we have to use operator valued distributions, localized in position space. 

It turns out that (similarly to the translation invariant case) cubic and quartic terms in the excitation Hamiltonian $\cL_N$ still contain important contributions to the ground state energy and to the energy of low-lying excited eigenvalues. To extract these contributions, we will conjugate $\cL_N$ with a (generalized) Bogoliubov transformation having the form $e^B$, with $B$ quadratic in (modified) creation and annihilation operators. 

Through the unitary operator $e^B$, we aim at removing short-scale correlations among particles; to reach this goal, we choose the integral kernel $\eta$ defining $B$ through the ground state solution of the Neumann problem associated with the interaction potential $N^2 V (N.)$, in the ball of radius $\ell$, for a fixed $0 < \ell < 1$. 

Conjugation with $e^B$ leads to a new excitation Hamiltonian $\cG_N = e^{-B} \cL_N e^B$, with renormalized quadratic off-diagonal terms. The properties of $\cG_N$ are stated in Prop. \ref{prop:GN}, which can be shown similarly as in the  translation invariant setting. For completeness, we include a sketch of the proof, but we defer it to Section \ref{sec:GN}. 

Looking at the form of $\cG_N$ in Prop. \ref{prop:GN}, we will observe that, while conjugation with $e^B$ renormalizes quadratic terms, it does not change cubic terms. Similarly as in the translation invariant case, these terms cannot be neglected; they still contribute to the low-energy spectrum of the Hamilton operator. To extract the important contributions from the cubic terms, we conjugate $\cG_N$ with another unitary operator, this time given by the exponential $e^A$ of a cubic expression in (modified) creation and annihilation operators.
This leads to a second renormalized excitation Hamiltonian $\cJ_N = e^{-A} \cG_N e^A$, whose properties are stated in Prop \ref{prop:JN}. As already observed in \cite{YY} (in a slightly different context) and then in \cite{BBCS3}, it is important that scattering events described by $A$ involve two particles with high momenta and one with small momentum (and one particle from the condensate). For technical reasons, we need to restrict to small momenta using a Gaussian, rather than a sharp cutoff (this guarantees that the cutoff decays sufficiently fast, also in position space). For this reason, the proof of Prop. \ref{prop:JN}, given in Section \ref{sec:JN}, is a bit more complicated than the proof of the corresponding result in the translation invariant setting (where sharp cutoffs could be used).  

To control error terms produced by conjugation with $e^B$ and $e^A$, we will make use of strong a-priori bounds stated in Theorem \ref{prop:hpN} and Proposition \ref{prop:hpNcubic}; these important estimates follow from the proof of complete Bose-Einstein condensation obtained in \cite{BSS}. 

It follows from Prop. \ref{prop:JN} that, up to negligible contributions, the renormalized excitation Hamiltonian $\cJ_N$ is quadratic in creation and annihiliation operators (in fact, $\cJ_N$ still contains a positive quartic term, which however can be neglected when proving lower bounds and which only produces small errors when evaluated on trial states that are appropriate to show upper bounds on 
the low-energy eigenvalues). In Section \ref{sec:diag}, we proceed therefore with the diagonalization of the quadratic part of $\cJ_N$. This is the main novelty of our analysis; the lack of translation invariant makes it much more involved than for particles on the torus. This is already clear by looking at the operators $H_\text{GP}$ and $E$, defined in \eqref{eq:defHGPE}, which determine the ground state energy and the excited eigenvalues of the Hamiltonian. In the translation invariant case, $\ph_0 \equiv 1$ and, in momentum space, $H_\text{GP} = p^2$ and $E = (|p|^4 + 16 \pi \frak{a}_0 p^2)^{1/2}$. For systems trapped by external potentials, on the other hand, $\ph_0$ is not a constant and it does not commute with $H_\text{GP}$. 

\medskip

Finally, let us remark that while writing up the last details of our manuscript, a result similar to (\ref{eq:excitHN-EN}) was posted in \cite{NT}. While the strategy used in \cite{NT} is similar to ours (and to the one previously developed in \cite{BBCS3}, in the translation invariant setting), the details of its implementation are different. Let us explain the main differences. 
\begin{itemize}
\item After removing the condensate through conjugation with the unitary map $U_N$, the authors of \cite{NT} extend the resulting excitation Hamiltonian, defined on the truncated Fock space with at most $N$ particles, to the full Fock space. This step, which can be justified with a localization argument in the number of particles operator, allows them to work with standard creation and annihilation operators $a^*, a$, rather than with the modified fields $b^*, b$ defined on $\cF_{\perp \ph_0}^{\leq N}$ (see (\ref{eq:bb-def}) below), and simplifies a bit the analysis (notice however that the action of generalized Bogoliubov transformation, defined in terms of modified creation and annihilation operators, $b^*, b$ can be controlled quite well with bounds like those collected in Lemma \ref{lm:d-bds}, which are, by now, standard tools; for this reason, working with $b^*, b$ and with generalized Bogoliubov transformations is, at least conceptually, not much more difficult than working with $a^* ,a$). 
\item Another important difference between our work and \cite{NT} is the choice of the integral kernel $\eta$ defining the (standard or generalized) Bogoliubov transformation and the cubic transformation renormalizing the excitation Hamiltonian. In our work, $\eta$ is chosen through the ground state solution of the Neumann problem, defined on a ball of radius $\ell$, for a fixed $0 < \ell  < 1$, independent of $N$. In \cite{NT}, on the other hand, $\eta$ is defined through the solution of the zero-energy scattering equation (\ref{eq:0en}), cutting off correlations at distance $1/N \ll \ell \ll 1$. While the choice of $\eta$ in terms of the solution of the Neumann problem has the advantage that the cutoff at $|x| = \ell$ does not contribute directly to the kinetic energy (in the sense that the radial derivative of $\eta$ vanishes at $|x| = \ell$), the definition in terms of the zero-energy scattering equation allows for slightly more general interaction potentials (the result of \cite{NT} only requires $V \in L^1 (\bR^3)$; in Theorem \ref{thm:main}, on the other hand, we assume $V \in L^3 (\bR^3)$, because this condition is needed to derive important properties of the solution of the Neumann problem). Also the different choice of $\ell$ produces important differences between the two approaches. We choose $\ell$ of order one, independent of $N$. This means that after quadratic and cubic renormalizations, the resulting excitation Hamiltonian only requires diagonalization on momenta of order one. In \cite{NT}, the authors choose $\ell = N^{-\alpha}$ for some $0 < \alpha < 1$. The advantage of this choice is that the integral kernel $\eta$ has small Hilbert-Schmidt norm; as a consequence, in a commutator expansion, terms of higher order in $\eta$ are negligible. The price to pay for this choice of $\ell$ is that, in the diagonalization, we need to act also on large momenta (the diagonalization, in this case, needs to renormalize also all momenta scales larger that one and smaller than $\ell^{-1}$). At the end, choosing $\ell$ small, as done in \cite{NT}, seems to make the analysis slightly simpler. 
\item While we and the authors of \cite{NT} both apply the same general strategy to diagonalize the renormalized excitation Hamiltonian, in our paper we make an additional effort to identify all non-vanishing contributions to the ground state energy $E_N$. In particular, our expression (\ref{eq:gsenergy}), with the Bogoliubov energy given by (\ref{eq:bog1}) or by (\ref{eq:Ebog}), shows  that the ground state energy, in the limit $N \to \infty$, only depends on the interaction potential through its scattering length $\frak{a}_0$ (notice that also the minimizer $\ph_0$ of the Gross-Pitaevskii energy functional only depends on $V$ through its scattering length). In \cite{NT}, on the other hand, the ground state energy is left as a complicated expression depending on $N$ and the limit is not further investigated (there is also no proof of the fact that $E_N - N \cE_\text{GP} (\ph_0)$ approaches a limit, as $N \to \infty$). In fact, it should be noted that our assumptions that $\nabla V_\text{ext}$ grows at most exponentially and that $(H_\text{GP} + 1)^{-3/4-\eps} e^{-\alpha |x|}$  is a Hilbert-Schmidt operator, for any $\eps > 0$ and some $\alpha > 0$, are only needed to study the behaviour of the ground state energy $E_N$, in the limit $N \to \infty$, and to establish the validity of (\ref{eq:gsenergy}), with $E_\text{Bog}$ as in (\ref{eq:Ebog}). 
\end{itemize}

\vspace{0.5cm}

{\it Acknowledgements.} BS acknowledges financial support from the NCCR SwissMAP and from the 
ERC through the ERC-AdG CLaQS (grant agreement No. 834782). BS and SS also acknowledge 
financial support from the Swiss National Science Foundation through the Grant “Dynamical and 
energetic properties of Bose-Einstein condensates”.

\section{Excitation Hamiltonians} \label{sec:fock}

To prove Theorem \ref{thm:main} it is convenient to factor out the Bose-Einstein condensate and to focus instead on its orthogonal excitations. Recalling that $\pn \in L^2 (\bR^3)$ denotes the unique normalized minimizer of the Gross-Pitaevskii functional (\ref{eq:defGPfunctional}), we consider the unitary map $U_N : L^2_s (\bR^{3N}) \to \cF_{\perp \pn}^{\leq N}$ mapping $L^2_s (\bR^{3N})$ into the truncated Fock space 
\[ \cF_{\perp \pn}^{\leq N} = \bigoplus_{n = 0}^N L^2_{\perp \pn} (\bR^3)^{\otimes_s n} \]
constructed on the orthogonal complement $L^2_{\perp \pn} (\bR^3)$ of $\pn$. The map $U_N$, first introduced in \cite{LNSS}, is defined by $U_N \psi_N = \{ \alpha_0, \alpha_1, \dots , \alpha_N \}$, if 
\[ \psi_N = \alpha_0 \pn^{\otimes N} + \alpha_1 \otimes_s \pn^{\otimes (N-1)} + \dots + \alpha_N \]
with $\alpha_j \in L^2_{\perp \pn} (\bR^3)^{\otimes j}$, for all $j=0,1,\dots , N$ (here $\otimes_s$ denotes the symmetric tensor product). The action of $U_N$ on creation and annihilation operators is given by 
	\begin{equation}\label{2.1.UNconjugation}
   		 \begin{split}
    		U_{N} a^* (\pn ) a (\pn) U_{N}^*  & = N-\cN, \\
    		U_{N} a^*(f)a (\pn) U_{N}^* &= a^*(f) \sqrt{N-\cN} , \\
    		U_{N} a^* (\pn) a(g) U_{N}^* &= \sqrt{N-\cN} a(g) , \\
    		U_{N} a^* (f) a(g)U_{N}^* &= a^*(f)a(g)
    		\end{split}
    		\end{equation}
for all $ f,g\in L^2_{\perp \pn} (\bR^3)$, where $ \cN$ denotes the number of particles operator in $ \cFpn$. 

With $U_N$, we define the excitation Hamiltonian $ \cL_N = U_N H_N U_N^*$, acting on a dense subspace of $\cFpn$. Using the relations \eqref{2.1.UNconjugation}, $ \cL_N$ can be written as
\begin{equation}\label{eq:cLN} \cL_N =  \cL^{(0)}_{N}+\cL^{(1)}_{N} + \cL^{(2)}_{N} + \cL^{(3)}_{N} + \cL^{(4)}_{N} , \end{equation} 
where, in the sense of forms in $ \cFpn$, we have that
	\begin{equation}\label{eq:cLNj}
	\begin{split}
	\mathcal{L}_N^{(0)} &= \big\langle \pn, \big[ -\Delta + V_\text{ext} + \frac{1}{2} \big(N^3 V(N\cdot) * \vert \pn \vert^2\big) \big] \pn \big\rangle (N-\cN) \\
	& \qquad -  \frac12\big\langle \pn,  \big(N^3 V(N\cdot) * \vert \pn \vert^2 \big)\pn \big\rangle (\cN+1)(1-\cN/N) ,  \\
	\mathcal{L}_N^{(1)} &=   \sqrt{N} b\left(  \left( N^3 V(N\cdot) * \vert \pn\vert^2   -8\pi \frak{a}_0|\pn|^2 \right)\pn\right)\\
	&\hspace{0.5cm} - \frac{\mathcal{N}+1}{\sqrt N}  b\left( \left( N^3 V(N\cdot) * \vert \pn\vert^2 \right) \pn \right) + \text{h.c.},\\
	\mathcal{L}_N^{(2)} &= \int dx\;     a_x^* (-\Delta_x) a_x +  \int dx\;  V_\text{ext}(x)a_x^{*} a_x  
		\\
		&\hspace{0.5cm} + \int dxdy\;  N^3 V(N(x-y)) \pn^2(y)  \Big(b_x^* b_x - \frac{1}{N} a_x^* a_x \Big)   \\
		&\hspace{0.5cm}+ \int   dxdy\; N^3 V(N(x-y)) \pn(x) \pn(y) \Big( b_x^* b_y - \frac{1}{N} a_x^* a_y \Big)  \\
		&\hspace{0.5cm}+ \frac{1}{2}  \int  dxdy\;  N^3 V(N(x-y)) \pn(y) \pn(x) \Big(b_x^* b_y^*  +  b_xb_y \Big), \\
	\mathcal{L}_N^{(3)} &=  \int   N^{ 5/2} V(N(x-y)) \pn(y) \big( b_x^* a^*_y a_x +   a^*_xa_yb_x\big) dx dy, \\
	\mathcal{L}_N^{(4)} &= \frac{1}{2} \int  dxdy\; N^2 V(N(x-y)) a_x^* a_y^* a_y a_x .
	\end{split}
	\end{equation}	
For $x \in \bR^3$, we introduced here operator-valued distributions  
\be \label{eq:bb-def} b_x = \sqrt{\frac{N-\cN}{N}} \, a_x , \qquad b^*_x = a^*_x \, \sqrt{\frac{N-\cN}{N}} \ee
with $a_x^*, a_x$ denoting standard creation and annihilation operators. From the canonical commutation relations $[a_x, a_y^*] = \delta(x-y)$, $[a_x, b_y] = [a_x*, a_y^*] = 0$, we obtain 
 \begin{equation} \label{eq:comm-b}
 [ b_x, b_y^* ] = \left( 1 - \frac{\cN}{N} \right) \delta (x-y) - \frac{1}{N} a_y^* a_x, \hspace{1cm}[ b_x, b_y ] = [b_x^* , b_y^*] = 0 
 \end{equation}
and 
\begin{align*}
[b_x, a_y^* a_z] &=\delta (x-y) b_z, \qquad 
[b_x^*, a_y^* a_z] = -\delta (x-z) b_y^*.
 \end{align*}
For $f \in L^2_{\perp \ph_0} (\bR^3)$, it is also useful to define the operators 
\[ b(f) = \int \overline{f (x)} \, b_x \, dx =  \sqrt{\frac{N-\cN}{N}} a (f), \qquad b(f) = \int f (x) \, b^*_x \, dx =  a^* (f) \sqrt{\frac{N-\cN}{N}}. \]
Modified creation and annihilation operators can be controlled as the standard creation and annihilation operators, ie.  
\[ \begin{split} 
\| b(f) \xi \| \leq  \| f \| \| \cN^{1/2} \xi \|, \qquad
\| b^* (f) \xi \| \leq  \| f \| \| (\cN + 1)^{1/2} \xi \|.
\end{split} \]
where we recall that $\cN$ denotes the number of particles operator on $\cF_{\perp \pn}^{\leq N}$. 

To determine the spectrum of $\cL_N$ (and thus the spectrum of (\ref{eq:defHN})), a naive application of  Bogoliubov method would suggest to ignore all contributions except the constant and quadratic terms and to diagonalize the resulting Fock space Hamiltonian. This, however, does not yield the correct spectrum. In fact, it  does not even produce the right leading order contribution to the ground state energy. The cubic and quartic terms in \eqref{eq:cLN} contain relevant contributions to the energy.  The point is that unitary conjugation with the map $U_N$ (which leads to the excitation Hamiltonian $\cL_N$) removes products of the condensate wave function, but it leaves correlations in the excitation vector $U_N \psi_N = \{ \alpha_0, \dots , \alpha_N \}$. To extract the important contributions to the energy from cubic and quartic terms in (\ref{eq:cLNj}), we need therefore to factor out some of the correlations among the particles. To this end, we follow the strategy introduced in \cite{BBCS4} and we conjugate $\cL_N$ with a generalized Bogoliubov transformation.

To define the kernel of the Bogoliubov transformation, we consider the Neumann ground state $f_\ell$ that solves  
		\begin{equation}\label{eq:scatl} \left[ -\Delta + \frac{1}{2} V \right] f_{\ell} = \lambda_{\ell} f_\ell \end{equation}
on the  ball $|x| \leq N\ell$, for some fixed $0 < \ell < 1$. As explained below, we will choose the parameter $\ell>0$ sufficiently small, but independent of $N$. To simplify notation, we do not indicate the $N$-dependence in $f_\ell$ and $\lambda_\ell$. By radial symmetry of $f_\ell$, we can normalize it so that $f_\ell (x) = 1$ if $|x| = N \ell$. By scaling, $f_\ell (N.)$ solves 
		\[ \left[ -\Delta + \frac{ N^2}{2} V (Nx) \right] f_\ell (Nx) = N^2 \lambda_\ell f_\ell (Nx) \]
on the ball where $|x| \leq \ell$. We then extend $f_\ell (N.)$ to $\bR^3$, by setting $f_{N,\ell} (x) = f_\ell (Nx)$ if $|x| \leq \ell$ and $f_{N,\ell} (x) = 1$ for $x \in \bR^3$ with $|x| > \ell$. Thus, $ f_{N,\ell}$ solves   
		\begin{equation}\label{eq:scatlN}
		 \left( -\Delta + \frac{N^2}{2} V (N.) \right) f_{N,\ell} = N^2 \lambda_\ell f_{N,\ell} \chi_\ell,
		\end{equation}
where $\chi_\ell$ denotes the characteristic function of the ball $ B_\ell(0) \subset \bR^3$. Finally, we let $w_\ell = 1-f_\ell$. Notice that $ w_\ell(N.)$ is support in $ B_\ell(0)$. Defining the Fourier transform of $w_\ell $ through
		\[ \widehat w_\ell (p) = \int_{}dx\; w_\ell(x)e^{-2\pi ipx}, \]
we obtain 
		\[ \int_{ }dx\; w_\ell(N x)e^{-2\pi ipx} = \frac1{N^3}\widehat w_\ell (p/N). \]
From \eqref{eq:scatlN}, we find  
		\begin{align*}
		\begin{split}  - 4\pi^2p^2 \widehat{w}_\ell (p/N) +  \frac{N^2}{2} ( \widehat{V} (./N) \ast \widehat{f}_{N,\ell}) (p) = N^5 \lambda_\ell ( \widehat{\chi}_\ell \ast \widehat{f}_{N,\ell}) (p). \end{split} \end{align*}
The next lemma collects important properties of $ f_\ell, w_\ell$ and the Neumann eigenvalue $ \lambda_\ell$. Its proof is based on \cite[Lemma A.1]{ESY0}, \cite[Lemma 4.1]{BBCS4} and can be found in \cite[Appendix B]{BBCS3} and \cite[Lemma 3.2]{BSS}.
\begin{lemma} \label{3.0.sceqlemma}
Let $V \in L^3 (\bR^3)$ be non-negative, compactly supported and spherically symmetric. Fix $\ell> 0$ and let $f_\ell$ denote the solution of \eqref{eq:scatl}. 
\begin{enumerate}
\item [i)] We have that
\begin{align*} 
 \lambda_\ell = \frac{3\frak{a}_0 }{(\ell N)^3} \left(1 + \frac{9}{5}\frac{\frak{a}_0}{\ell N}+\mathcal{O} \big(\frak{a}_0^2  / (\ell N)^2\big) \right).
\end{align*}
\item[ii)] We have $0\leq f_\ell, w_\ell \leq 1$ and there exists a constant $C > 0$ such that 
\begin{equation} \label{eq:Vfa0} 
\left|  \int_{\bR^3}  V(x) f_\ell (x) dx - 8\pi \frak{a}_0 \Big(1 + \frac 3 2 \frac{\frak{a}_0}{ \ell N}  \Big)  \right| \leq \frac{C \frak{a}_0^3}{(\ell N)^2} 
\end{equation}
for all $\ell \in (0;1)$, $N \in \bN$.
\item[iii)] There exists a constant $C>0$ such that 
			\begin{equation}\label{3.0.scbounds1} 
			w_\ell(x)\leq \frac{C}{|x|+1} \quad\text{ and }\quad |\nabla w_\ell(x)|\leq \frac{C }{|x|^2+1} 
			\end{equation}
for all $x\in \bR^3$, $\ell \in (0;1)$ and $N \in \bN$ large enough. Moreover,
		\begin{align*}
		\Big| \frac{1}{(N \ell)^2} \int_{\bR^3} w_{\ell}(x) dx - \frac 2 5 \pi \frak{a}_0   \Big| \leq   \frac{C \mathfrak{a}_0^2}{N \ell}
		\end{align*}
for all $\ell \in (0;1)$ and $N \in \bN$ large enough.
\item[iv)] There exists a constant $C > 0$ such that 
		\begin{equation}\label{eq:whwl} |\widehat{w}_\ell (p)| \leq \frac{C}{|p|^2} \end{equation}
for all $p \in \bR^3$, $\ell \in (0;1)$ and $N \in \bN$ large enough. 
\end{enumerate}        
\end{lemma}
We define the kernel $k\in L_s^2(\mathbb{R}^3\times\mathbb{R}^3)$ through 
		\be \label{eq:defk} k(x;y) = -N w_\ell(N(x-y)) \pn(x)\pn(y).  \ee
With $Q= 1- |\pn \rangle \langle \pn|$ denoting the orthogonal projection onto the orthogonal complement of $\pn$, we set 		
\be \label{eq:defeta} \eta(x;y) = \big[ \big(Q\otimes Q \big) k \big]  (x;y) .  \ee
The following lemma (whose proof can be found in \cite[Lemma 3.3]{BDS}, \cite[Lemma 4.3]{BS} and in \cite[Lemma 2.2]{BSS})  summarizes important properties of the kernels $k, \eta$ and $\mu = \eta - k$. 
\begin{lemma}\label{lm:bndseta}
Assume \eqref{eq:asmptsVVext} and let $\ell\in (0;1)$. Then, there exists $C>0$, uniform in $N$, $\ell\in (0;1)$ and $x\in\bR^3$, such that for $i=1,2$ it holds true that
	\be  \label{eq:bndseta}\begin{split}
	\Vert \eta \Vert &\leq C \ell^{1/2}, \hspace{0.5cm} \| k\| \leq C\ell^{1/2} , \hspace{0.5cm} \|\mu\| \leq C\ell^{1/2}, \\
	  \Vert \eta_{x} \Vert &\leq C \vert \pn(x) \vert, \hspace{0.5cm} \Vert k_{x} \Vert \leq C \vert \pn(x) \vert, \hspace{0.5cm}\Vert \mu_{x} \Vert \leq C \vert \pn(x) \vert, \\
	   \Vert \nabla_i \eta \Vert &\leq C\sqrt{N}, \hspace{0.5cm} \,\Vert \nabla_i k \Vert  \leq C \sqrt{N}, \hspace{0.5cm} \| \nabla_i\mu \|\leq C,\hspace{0.5cm} \| \Delta_i\mu\|\leq C.
	 \end{split}
	\end{equation} 
Moreover, for $x,y\in \bR^3$ and $\alpha_1, \alpha_2 \in \{ 0,1,2 \}$ we have the pointwise bounds
\begin{equation} \label{eq:bndsetaptw}
\begin{split}
&| \eta(x;y)|, | k(x;y)| \leq C \frac{\pn (x) \pn (y)}{|x-y| + N^{-1}} , \, |\mu(x;y)| \leq C \pn(x) \pn(y) , \\
&| \nabla_1 \eta(x;y)|, | \nabla_1 k(x;y)| \leq C \frac {\pn(x)\pn(y)}{|x-y|^2 + N^{-2}} +C \frac{|\nabla \pn (x)| \pn (y)}{|x-y| + N^{-1}}, \\
&|\nabla_1^{\alpha_1} \nabla_2^{\alpha_2}  \mu(x;y)| \leq C \sum_{j_1 =0}^{\alpha_1} | \nabla^{j_1}  \pn (x) |   \sum_{j_2=0}^{\alpha_2} |\nabla^{j_2} \pn (y)| 
\end{split}
\end{equation} 
(and similarly for $\nabla_2 \eta$, $\nabla_2 k$). Finally, identifying the kernel $\eta$ with the corresponding Hilbert-Schmidt operator and denoting by $\eta^{(n)}$ the kernel of its $n$-th power, we find, for every $n \geq 2$ and $x,y\in\bR^3$ that $\|\nabla_i \eta^{(n)}\|\leq C^n $, $\| \Delta_i \eta^{(n)}\|\leq C^n$ and that 	\begin{equation}\label{eq:bndsetan}
	\begin{split}
	&\vert \eta^{(n)} (x;y) \vert \leq  C \Vert \eta \Vert^{n-2} \vert \pn(x) \vert   \vert \pn(y) \vert, \\
	&\Vert \nabla_x \eta^{(n)}_x \Vert  \leq C (\vert \nabla \pn (x)\vert + \pn(x)), \\
	&\Vert \Delta_x \eta^{(n)}_x \Vert  \leq C (\vert \Delta \pn(x)\vert + \vert \nabla \pn (x)\vert + \pn(x)). 
	\end{split}
	\end{equation}	
\end{lemma} 

We introduce now the antisymmetric operator 
\begin{equation}\label{eq:Beta}
B = \frac{1}{2} \int dx dy \, \big[ \eta (x;y) b_x^* b_y^* - \text{h.c.} \big] 
\end{equation} 
and we consider the generalized Bogoliubov transformation $e^B$. Since $\eta \in (Q \otimes Q) L^2 (\bR^3 \times \bR^3)$, the unitary operator $e^B$ maps the Hilbert space $\cF^{\leq N}_{\perp \pn}$ back into itself.  Another important property of $e^B$ is the fact that it preserves, approximately, the number of particles (ie. the number of excitations of the Bose-Einstein condensate). The following lemma was proved in \cite[Lemma 3.1]{BS}.
\begin{lemma} \label{lm:Npow} Let $\eta \in L_s^2 (\bR^3 \times \bR^3)$. Let $B$ be the antisymmetric operator defined in (\ref{eq:Beta}). For every $n \in \bN$ there exists a constant $C > 0$ (depending only on $n$ and on $\| \eta \|$) such that 
\[ e^{-B} (\cN +1)^{n} e^{B} \leq C (\cN+1)^{n} \]
as an operator inequality on $\cF_{\perp \ph_0}^{\leq N}$.
\end{lemma}
We will also need more precise information on the action of the $e^B$. Keeping in mind the fact that, if we replaced in (\ref{eq:Beta}) the $b^* ,b$-operators with standard creation and annihilation operators, $e^B$ 
would be a Bogoliubov transformation with explicit action, we decompose, for an arbitrary $f \in L^2_{\perp \pn}  (\bR^3)$,   
\begin{equation}\label{eq:defd} e^{-B} b(f) e^{B} = b( \cosh_\eta (f)) + b^*( \sinh_\eta (\overline{f})) + d_\eta (f), \end{equation}
where 		
\begin{equation}\label{eq:defcoshsinheta}
\begin{split}
\cosh_\eta = \sum_{n=0}^\infty\frac{ \eta^{(2n)}}{(2n)!}, \hspace{0.5cm} \sinh_\eta = \sum_{n=0}^\infty\frac{ \eta^{(2n+1)} }{(2n+1)!}.
\end{split}
\end{equation}
We also introduce the operator-valued distributions $d_{\eta,x}, d^*_{\eta,x}$ defined through 
\begin{equation} \label{eq:defdx}
e^{-B} b_x e^{B} = b( \cosh_{\eta,x}) + b^*( \sinh_{\eta,x}) + d_{\eta,x}.
\end{equation}
We will use the short-hand notations 
\begin{equation} \label{eq:defgammasigma} 
\gamma = \cosh_\eta,  \qquad \sigma = \sinh_\eta, \qquad \text{p}_\eta = \gamma - 1, \qquad 
\text{r}_\eta = \sigma - \eta. 
\end{equation} 
It is a simple consequence of Lemma \ref{lm:bndseta} that, under the same assumptions as in Lemma \ref{lm:bndseta}, there is a constant $C > 0$ such that, for all $\alpha_1, \alpha_2 \in \{0,1,2 \}$ with $\alpha_1 + \alpha_2 \leq 2$, we have \[ \Vert \nabla_1^{\alpha_1} \nabla_2^{\alpha_2} \emph{p}_\eta \Vert, \Vert \nabla_1^{\alpha_1} \nabla_2^{\alpha_2} \emph{r}_\eta \Vert  \leq C.\] Moreover, we have the bounds  $\vert\emph{p} (x;y) \vert, \vert  \emph{r} (x;y) \vert  \leq  C  \pn(x) \pn(y)$ for all $x,y \in \bR^3$.

The idea that on states with few excitations $e^B$ acts almost as a Bogoliubov transformation is confirmed by the next lemma, which bounds the norm of the operators $d_\eta (f)$, in terms of (appropriate powers of) the number of particles operator. This lemma, whose proof is an adaptation of the translation-invariant case \cite[Lemma 3.4]{BBCS4}, requires $\| \eta \|$ to be sufficiently small. From Lemma \ref{lm:bndseta}, this can be achieved requiring that $\ell > 0$ is small enough. 
\begin{lemma}\label{lm:d-bds}
Let $n\in  \mathbb{Z}$, $f\in L^2(\bR^3)$. Let $\eta\in L^2(\bR^3\times \bR^3)$ be as defined in (\ref{eq:Beta}), with $\ell > 0$ small enough. Let $d_{\eta}(f)$ be as in (\ref{eq:defd}) (and $d_{\eta ,x}$ as defined in \eqref{eq:defdx}). Then there exists $C > 0$ such that  
\begin{align*} 
\begin{split}
\Vert (\cN+1)^{n/2}d_{\eta}(f) \xi \Vert & \leq \frac{C }{N} \|f\|  \Vert (\mathcal{N}+1)^{(n+3)/2} \xi \Vert, \\
\Vert (\cN+1)^{n/2}d_{\eta}(f)^* \xi \Vert & \leq \frac{C}{N} \|f\|  \Vert (\mathcal{N}+1)^{(n+3)/2} \xi \Vert
		\end{split}
		\end{align*}	
and such that, for all $x\in\bR^3$, we have 
		\begin{equation} \label{eq:bnddx}
		\begin{split}
		\Vert (\cN+1)^{n/2}d_{\eta,x} \xi \Vert & \leq \frac{C}{N} \Big[  \Vert a_x (\mathcal{N}+1)^{(n+2)/2} \xi \Vert +  \pn (x) \Vert (\mathcal{N}+1)^{(n+3)/2} \xi \Vert\Big].
		\end{split}
		\end{equation}
Furthermore, if we set $\overline{d}_{\eta,x} = d_{\eta,x} +  (\cN/N) b^*( \eta_{x}) $, we find 
\begin{equation} \label{eq:bndaydxbar}
		\begin{split}
		&\Vert (\cN+1)^{n/2}a_y \overline{d}_{\eta,x} \xi \Vert\\
		 &\leq \frac{C }{N}  \Big[  \pn (x) \pn (y)  \Vert (\mathcal{N}+1)^{(n+2)/2} \xi \Vert  + |\eta(x;y)|   \Vert (\mathcal{N}+1)^{(n+2)/2} \xi \Vert \\
  		  &\hspace{1.5cm}+     \pn (y)    \Vert a_x (\mathcal{N}+1)^{(n+1)/2}\xi \Vert  +  \pn (x) \Vert a_y (\mathcal{N}+1)^{(n+3)/2} \xi \Vert \\
		  &\hspace{1.5cm}+  \Vert a_x a_y (\mathcal{N}+1)^{(n+2)/2} \xi \Vert \Big]
		\end{split}
		\end{equation}
and, finally, we have that
	\begin{equation} \label{eq:bnddxdy}
	\begin{split}
	&\Vert (\cN+1)^{n/2}d_{\eta,x} d_{\eta,y} \xi \Vert \\
	&\leq \frac{C}{N^2} \Big[  \pn (x) \pn (y) \Vert (\mathcal{N}+1)^{(n+6)/2} \xi \Vert +   |\eta(x;y)|   \Vert (\mathcal{N}+1)^{(n+4)/2}\xi \Vert  \\
	&\hspace{1.2cm} +  \pn (y)  \Vert a_x (\mathcal{N}+1)^{(n+5)/2)} \xi \Vert+    \pn (x)  \Vert a_y(\mathcal{N}+1)^{ (n+5)/2} \xi \Vert  \\
	&\hspace{1.2cm} +  \Vert a_x a_y (\mathcal{N}+1)^{(n+4)/2} \xi \Vert \Big].
	\end{split}
	\end{equation}
\end{lemma}

Using the generalized Bogoliubov transformation $e^B$, we define the renormalized excitation Hamiltonian 
\be  \label{eq:defGN}\cG_{N} = e^{-B} \cL_N e^{B}. \ee
In the next proposition, we collect some important properties of $\cG_N$. Here and in the rest of the paper, 
we are going to use the notation 
\[ \begin{split}
		\cK & = \int dx\; a^*_x( -\Delta_x) a_x, \hspace{3.5cm} \cV_\text{ext} = \int dx\, V_\text{ext}(x)a^*_x a_x, \\
		\cV_N &= \frac12 \int dxdy\; N^2V(N(x-y))a^*_x a^*_y a_y a_x,\hspace{0.5cm}  \cH_N = \cK + \cV_\text{ext} + \cV_N.  \end{split} \]
\begin{prop}\label{prop:GN}
Assume \eqref{eq:asmptsVVext}, let $\eta$ be defined as in \eqref{eq:defeta} and let $E_N$ denote the ground state energy of $\cG_{N}$. Then the following holds true:
\begin{itemize}
\item[a)] We have that $ |E_N - N\cE_{GP}(\pn)| \leq C$ and 
		\begin{equation}\label{eq:GDelta} \cG_N - E_N = \cH_N + \Delta_N, \end{equation}
where the error term $\Delta_N$ is such that for every $\delta >0$, there exists $C > 0$ with
		\begin{equation}\label{eq:Delta-bd} \pm \Delta_N \leq \delta \cH_N + C (\cN + 1). \end{equation}
Furthermore, for every $k \in \bN$ there exists a $C > 0$ such that
		\begin{equation}\label{eq:adkN}
		\pm \text{ad}^{\, (k)}_{\, i\cN} (\cG_N) = \pm \text{ad}^{\, (k)}_{\, i\cN} (\Delta_N) = \pm \big[ i \cN , \dots \big[ i \cN , \Delta_N \big] \dots \big] \leq C (\cH_N + 1). \end{equation}

\item[b)] Let $\sigma$ and $ \gamma$ be defined as in \eqref{eq:defcoshsinheta} and let $\kappa_{\cG_N}$ denote the constant
		\begin{equation}\label{eq:defcN}
		\begin{split}
		\kappa_{\cG_N} & = N\big\langle \pn, (-\Delta +V_\text{ext} + \wh V(0)\pn^2/2)\pn\big\rangle- 4\pi \mathfrak{a}_0\|\pn\|_4^4  \\
		&\hspace{0.5cm}  + \tr\big(\sigma (-\Delta+V_\text{ext} -\eps_{GP})\sigma\big) \\
		&\hspace{0.5cm}+ \tr\big( \gamma   \big[N^3 V(N(x-y)) \pn(x)\pn(y)\big]  \sigma ) \\
		&\hspace{0.5cm}  +\tr\big( \sigma \big[   N^3 V(N.) \ast \pn^2  + N^3 V(N(x-y)) \pn(x)\pn(y)\big] \sigma\big)\\
		&\hspace{0.5cm} + \frac12  \int dxdy\, N^2V(N(x-y))| \langle\sigma_x, \gamma_y\rangle|^2. 
		\end{split}
		\end{equation}
Here,  $(N^3V(N.)\ast \pn^2)$ acts as multiplication operator and we identify kernels like $ N^3V(N(x-y))\pn(x)\pn(y)$ with their associated Hilbert Schmidt operators. Let
		\begin{equation}\label{eq:Phi-def} \begin{split}
		\Phi &=  \gamma \big(-\Delta+V_\text{ext}-\eps_{GP}\big)\gamma + \sigma \big(-\Delta+V_\text{ext}-\eps_{GP}\big)\sigma\\
		&\hspace{0.4cm}+ \gamma \big( N^3V(N.)\ast\pn^2 + N^3V(N(x-y))\pn(x)\pn(y)\big)\gamma \\
		&\hspace{0.4cm}+ \sigma\big(N^3V(N.)\ast\pn^2 + N^3V(N(x-y))\pn(x)\pn(y)\big)\sigma\\
		&\hspace{0.4cm}+\Big( \gamma \big(  N^3(Vf_\ell)(N(x-y))\pn(x)\pn(y)\big)\sigma+\emph{h.c.}\Big),
		\end{split}\end{equation} 
and 
		\begin{equation} \label{eq:Gamma-def} \begin{split}
		\Gamma &=  \gamma \big(  N^3(Vf_\ell)(N(x-y))\pn(x)\pn(y)\big)\gamma\\
		&\hspace{0.5cm}+\sigma \!\big( N^3(Vf_\ell)(N(x-y))\pn(x)\pn(y)\big)\!\sigma\\
		&\hspace{0.5cm} +\big[ \sigma \big(-\Delta+V_\text{ext}-\eps_{GP}\big)\gamma +\emph{h.c.}\big]\\
		&\hspace{0.5cm}+ \big[ \sigma \big(N^3V(N.)\ast\pn^2 + N^3V(N(x-y))\pn(x)\pn(y)\big)\gamma + \emph{h.c.}\big].\\
		\end{split}\end{equation} 
With $\Phi$ and $\Gamma$, we then define the quadratic Fock space Hamiltonian
		\begin{equation}\label{eq:cQcGN}
		\begin{split}
		\cQ_{\cG_N} = \int dxdy\, \Phi(x;y) b^*_x b_y + \frac12 \int dxdy\,  \Gamma(x;y) \big( b^*_x b^*_y + b_xb_y\big),
		\end{split}
		\end{equation}
and we denote by $\cC_{\cG_N}$ the cubic operator
		\begin{equation}\label{eq:cCcGN}
		\cC_{\cG_N}=  \int dxdy \, N^{5/2}V(N(x-y)) \pn(y)  b^*_x b^*_y \big(b (\gamma_x) + b^*(\sigma_x)\big) +\emph{h.c.} 
		\end{equation}
Then, we have that 
		\begin{equation}\label{eq:propGN}
		\begin{split}
		\mathcal{G}_{N } 
		&=  \kappa_{\cG_N} + \cQ_{\cG_N} + \cC_{\cG_N} +\cV_N + \cE_{\cG_N}
		\end{split}
		\end{equation}
for a self-adjoint operator $ \cE_{\cG_N}$ that satisfies in $\cFpn$ the operator bound
		 \begin{equation} \label{eq:cEcGNbnd}
		 \pm \cE_{\cG_N} \leq  CN^{-1/2} (\cH_N +\cN^2+ 1)(\cN+1).
		 \end{equation}
\end{itemize}		
\end{prop}

The proof of Prop. \ref{prop:GN} is similar to the proof of \cite[Prop. 3.2]{BBCS3}. For completeness, we sketch it in Sect. \ref{sec:GN}. 

Combining Prop. \ref{prop:GN} with the lower bound 
\begin{equation}\label{eq:BSS} \cG_N - E_N \geq c \cN - C \end{equation} 
from \cite[Eq. (3.1)]{BSS} (and from Lemma \ref{lm:Npow} in the present paper), we obtain strong a-priori estimates on the number and on the energy of excitations of the Bose-Einstein condensate in low-energy states of (\ref{eq:defHN}). These bounds are crucial for the rest of out analysis. 
\begin{theorem}\label{prop:hpN}
Assume \eqref{eq:asmptsVVext}, let $\eta$ be defined as in \eqref{eq:defeta} and let $E_N$ denote the ground state energy of $H_N$, defined in \eqref{eq:defHN}. Let $\psi_N \in L^2_s (\mathbb{R}^{3N})$ with $\| \psi_N \| = 1$ belong to the spectral subspace of $H_N$ with energies below $E_N + \zeta$, for some $\zeta > 0$, i.e. 
		\[ \psi_N = {\bf 1}_{(-\infty ; E_N + \zeta]} (H_N) \psi_N.\]

Let $\xi_N = e^{-B} U_N \psi_N\in \cFpn$ be the renormalized excitation vector associated with $\psi_N$. Then, for every $j\in\mathbb N$ there exists a constant $C > 0$ such that 
\begin{align*}
\langle \xi_N, (\cH_N+1) (\cN +1)^j  \xi_N \rangle \leq C (1 + \zeta^{j+1}). \end{align*}
\end{theorem}
\begin{proof}
The proof carries over from \cite[Prop. 4.1]{BBCS3}, once we show that
		\begin{equation}\label{eq:cGNgeqcHN}
		\cG_N -E_N \geq   c\, \cH_N - C
		\end{equation}
for two positive constants $c, C>0$, independent of $N$. Indeed, this enables us to use the induction argument from \cite[Prop. 4.1]{BBCS3} combined with the bounds \eqref{eq:adkN}. The bound (\ref{eq:cGNgeqcHN}) can be shown interpolating (\ref{eq:BSS}) with the estimate 
\[ \cG_N - E_N \geq \frac{1}{2} \cH_N - C (\cN+ 1) \]
which follows from (\ref{eq:GDelta}) and (\ref{eq:Delta-bd}). 
\end{proof}

Although Prop. \ref{prop:GN} and Theorem \ref{prop:hpN} provide strong control on low-energy states of $\cG_N$, this is not enough to deduce Theorem \ref{thm:main}. The reason is that the cubic operator $\cC_{\cG_N}$ and the quartic operator $\cV_N$ appearing on the r.h.s. of \eqref{eq:cCcGN} still contain energy contributions of order one. To extract them, we follow again the strategy introduced in \cite{BBCS4} and we conjugate $\cG_N$ with an additional unitary transformation, this time generated by an operator cubic in the modified creation and annihilation fields. 

In order to define our cubic transformation, we need to introduce some notation. For an exponent $\exph > 0$ to be specified later on, we use $\chi_H (p) = \chi (|p| > N^\exph)$ to restrict to high momenta $|p| > N^\exph$. We consider the kernel 
\begin{equation}\label{eq:def-wtk} \wt{k}_H (x;y) = \left[ - N w_\ell (N.) * \check{\chi}_H \right] (x-y) \ph_0 (y). \end{equation} 
On the other hand, to restrict to small momenta, we use a Gaussian cutoff. For $0< \exps < \exph$, we set $g_L (p) = e^{-p^2/ N^{2\tau}}$. The advantage of $g_L$, with respect to a sharp cutoff, is that it decays fast also in position space. In particular, from the assumptions \eqref{eq:asmptsVVext} on $V_\text{ext}$, we find 
\begin{equation} \label{eq:Vextconvg} 
\vert V_\text{ext} * \check{g}_L \vert \leq C (V_\text{ext} + 1).
\end{equation} 
Recalling the definition of $\sigma, \gamma$ in (\ref{eq:defgammasigma}), we introduce the notation  
\begin{equation}\label{eq:def-gsL}  \sigma_{L} = \sigma *_2 \check{g}_L , \qquad \gamma_{L} = \gamma *_2 \check{g}_L \end{equation} 
where $*_2$ is the convolution in the second variable. In particular, for every $x \in \bR^3$, we have $\sigma_{L,x} = \sigma_x * \check{g}_L$ and $\gamma_{L,x} = \gamma_x * \check{g}_L$. 

In the next lemma, we collect useful bounds for various kernels.
\begin{lemma}\label{lm:bds-cubic}
We assume that $0<3\exps\leq \exph$. Let $\wt{k}_H, \sigma_L$ and $\gamma_L$ as in \eqref{eq:def-wtk} and \eqref{eq:def-gsL}. Then we have 
\begin{equation} \label{eq:bds-cubic1}
\begin{split}
&\| \wt{k}_H \|  \leq C N^{-\exph / 2} ,  \quad \sup_x \| \wt{k}_{H} (x,.)\|  \leq C N^{-\exph/2}, \quad \Vert \wt{k}_{H} (.,y) \Vert \leq C N^{-\exph/2} \pn(y) 
\end{split} \end{equation} 
for all $y \in \bR^3$. Setting $\wt{k} (x;y) = - N w_\ell (N(x-y)) \ph_0 (y)$ (no cutoff), we have 
\begin{equation}\label{eq:diffk} \| \nabla_x (\wt{k} - \wt{k}_H) (x,.) \| \leq C N^{\eps/2} , \qquad | (\wt{k} - \wt{k}_H) (x;y)| \leq C N^\eps \ph_0 (y). \end{equation} 
Moreover, we find 
\begin{equation}\label{eq:bds-cubic2} 
\begin{split} 
\| \sigma_{L,x} \| , \| \big[ (\gamma  - 1) *_2 \check{g}_L \big]_x \| \leq C \pn (x),   \end{split} \end{equation} 
for all $x \in \bR^3$, which implies $\| \sigma_L \| , \| (\gamma  - 1) *_2 \check{g}_L  \| , \| \gamma_L \|_\text{op} \leq C$, and also 
\begin{align*} 
\begin{split} 
\Vert V_\text{ext} \big[(\gamma - 1)*_2 \check{g}_L \big] \Vert, \Vert V_\text{ext} \sigma_{L} \Vert \leq C.
\end{split} \end{align*} 
Finally, we have 
\begin{equation} \label{eq:bds-cubic4} 
 \Vert \nabla_x \sigma_{L,x} \Vert \leq C N^{\exps/2} \big[\pn(x)  + |\nabla \pn (x)| \big] , \quad \| \nabla_x \big[ (\gamma - 1)_x *_2 \check{g}_L \big] \| \leq C  \big[\pn(x)  + |\nabla \pn (x)| \big].
\end{equation} 
\end{lemma} 
\begin{proof}
Using \eqref{eq:whwl} we compute
	\begin{align*}
	\Vert N w_\ell (N.) * \check{\chi}_H \Vert^2 \leq C \int_{\vert p \vert \geq N^\exph} \vert p \vert^{-4} dp = C N^{-\exph}.
	\end{align*}
	This implies all the bounds in (\ref{eq:bds-cubic1}). To show (\ref{eq:diffk}), we observe that $(\wt{k} - \wt{k}_H) (x;y) = \big[-N w_\ell (N.) * \check{\chi}_{H^c} \big] (x-y) \pn (y)$, with $\chi_{H_c}$ the indicator function of the set $\{ p \in \bR^3 : |p| \leq N^\eps \}$ and that, from (\ref{eq:whwl}), 
	\[ \big\| \nabla \big[ -N w_\ell (N.) * \check{\chi}_{H^c} \big] \big\|^2_2 = \int_{|p| \leq N^\eps}  dp \left[ \frac{|p|}{N^2} |\hat{w}_\ell (p/N)| \right]^2 \leq C N^{\eps} \]
	and 
	\[ \|  -N w_\ell (N.) * \check{\chi}_{H^c} \|_\infty \leq C \int_{|p| \leq N^\eps} \frac{1}{N^2} \hat{w}_\ell (p/N) dp \leq C N^\eps. \]
To prove (\ref{eq:bds-cubic2}), we estimate, with \eqref{eq:bndseta}, (\ref{eq:bndsetan}), Parseval's identity and $\vert g_L \vert \leq 1$, 
\[  \| \sigma_{L,x} \|= \Vert \widehat{\sigma_x} g_L \Vert \leq \Vert \sigma_x \Vert \leq \sum_{n\geq 0} \frac{1}{(2n+1)!} \Vert \eta^{(2n+1)}_x \Vert \leq C \pn(x).  \] 
The bound for $(\gamma -1) *_2 \check{g}_L$ can be proven similarly. Finally, we show (\ref{eq:bds-cubic4}). With \eqref{eq:bndsetan}, it is easy to check that $\| \nabla_x (\sigma_{L} -   k *_2 \check{g}_L)_x \| \leq C (\pn (x) + |\nabla \pn (x)|)$, with $k$ as defined in (\ref{eq:defk}). We have 
\begin{equation}\label{eq:nablasigma} \begin{split}  \nabla_x (k *_2 \check{g}_L) (x;y) = &- \int N^2 \nabla w_\ell (N(x-z-y)) \pn (x) \pn (y+z) \check{g}_L (z) dz \\ &- \int N w_\ell (N(x-z-y)) \nabla \pn (x) \pn (y+z) \check{g}_L (z) dz. \end{split} \end{equation} 
The $L^2$-norm squared (in the variable $y$, for fixed $x$) of the first term can be bounded, with \eqref{3.0.scbounds1}, by 
\[ \begin{split} \int dy &dz_1 dz_2 \frac{1}{|x-z_1-y|^2 |x-z_2-y|^2}  \pn (x)^2 \pn (y+z_1) \pn(y+z_2)  \check{g}_L (z_1) \check{g}_L (z_2) \\ \leq \; & C \int dz_1 dz_2 \, \frac{1}{|z_1 - z_2|} \check{g}_L (z_1) \check{g}_L (z_2) \leq C N^\tau \int dw_1 dw_2 \frac{1}{|w_1 - w_2|} e^{-w_1^2} e^{-w_2^2} \leq C N^\tau \end{split}  \]
where we switched to new variables $w_i = z_i  N^\tau$, for $i=1,2$. The second term on the r.h.s. of (\ref{eq:nablasigma}) can be controlled analogously. 
\end{proof}

We introduce the antisymmetric operator 
\begin{equation}\label{eq:A-def}  A =  \frac{1}{\sqrt{N}} \int dx dy \, \wt{k}_H (x;y) \wt{b}_x^* \wt{b}_y^* \left[ \wt{b} (\gamma_{L,x})  + \wt{b}^* (\sigma_{L,x})  \right] - \text{h.c.} \end{equation} 
In order to preserve the space $\cF_{\perp \ph_0}^{\leq N}$ and stay orthogonal to $\ph_0$, we used here 
the operators  
\[ \wt{b}_x = b (Q_x) = \int dz \, Q (x;z) b_z \]
with $Q = 1 - |\ph_0 \rangle \langle \ph_0|$. Using $A$, we define the cubically renormalized excitation Hamiltonian $\cJ_N = e^{-A} \cG_N e^{A}$. 

The next proposition summarises important properties of the operator $\cJ_N$. 
\begin{prop} \label{prop:JN}
	Assume \eqref{eq:asmptsVVext} and let $0<6 \exps\leq \exph\leq \frac{1}{2}$. Then we have
	\begin{equation}\label{eq:JNdec}
	\mathcal{J}_N = \kappa_{\mathcal{J}_N} + \mathcal{Q}_{\mathcal{J}_N} + \mathcal{V}_N + \mathcal{E}_{\mathcal{J}_N},
	\end{equation}
	where 
\be \label{eq:kJN}\begin{split}
	\kappa_{\mathcal{J}_N} 
	&= \kappa_{\cG_N}  -\emph{tr}(\sigma N^3 (Vw_\ell)(N.) * \pn^2 \sigma) - \emph{tr}(\sigma N^3 (Vw_\ell)(N(x-y)) \pn(x)\pn(y) \sigma)
	\end{split}\ee
	with $\kappa_{\cG_N}$ defined in \eqref{eq:defcN} and where the quadratic operator $ \cQ_{\mathcal{J}_N}$ is given by
	\begin{equation}\label{eq:cQcJN}
	\begin{split}
	\cQ_{\mathcal{J}_N} = \int dxdy\, \widetilde{\Phi}(x;y) b^*_x b_y + \frac12 \int dxdy\, \widetilde{ \Gamma}(x;y) \big( b^*_x b^*_y + b_xb_y\big)
	\end{split}
	\end{equation}
	for
	\begin{equation}\label{eq:deftildePhi}
	\begin{split}
	\widetilde{\Phi} &=  \gamma \big(-\Delta+V_\text{ext}-\eps_{GP}\big)\gamma + \sigma \big(-\Delta+V_\text{ext}-\eps_{GP}\big)\sigma\\
	&\hspace{0.4cm}+ \gamma \big( 8\pi\frak{a}_0\pn^2 + N^3(Vf_\ell)(N(x-y))\pn(x)\pn(y)\big)\gamma \\
	&\hspace{0.4cm}+ \sigma\big(8\pi\frak{a}_0\pn^2 + N^3(Vf_\ell)(N(x-y))\pn(x)\pn(y)\big)\sigma\\
	&\hspace{0.4cm}+\Big( \gamma \big(  N^3(Vf_\ell)(N(x-y))\pn(x)\pn(y)\big)\sigma+\emph{h.c.}\Big),
	\end{split}
	\end{equation}
	and 
	\begin{equation}\label{eq:deftildeGamma}
	\begin{split}
	\widetilde{\Gamma} &=  \gamma \big(  N^3(Vf_\ell)(N(x-y))\pn(x)\pn(y)\big)\gamma\\
	&\hspace{0.5cm}+\sigma \!\big( N^3(Vf_\ell)(N(x-y))\pn(x)\pn(y)\big)\!\sigma\\
	&\hspace{0.5cm} +\big[ \sigma \big(-\Delta+V_\text{ext}-\eps_{GP}\big)\gamma +\emph{h.c.}\big]\\
	&\hspace{0.5cm}+ \big[ \sigma \big(8\pi\frak{a}_0\pn^2 + N^3(Vf_\ell)(N(x-y))\pn(x)\pn(y)\big)\gamma + \emph{h.c.}\big].\\
	\end{split}
	\end{equation}
	Moreover, the self-adjoint operator $\mathcal{E}_{\mathcal{J}_N}$ is bounded by 
		\begin{equation} \label{eq:errorJN}
	\pm \mathcal{E}_{\mathcal{J}_N} \leq \frac{C}{N^{\min \{\frac{\exps}{2}, \frac{1}{2}-\exph\}}}  (\mathcal{H}_N+1)(\mathcal{N}+1)^3.
	\end{equation}
\end{prop}
		
The proof of Prop. \ref{prop:JN}, which is similar to the proof of \cite[Prop. 3.3]{BBCS3}, is sketched it in Sect. \ref{sec:GN}. 

In order to compute the spectrum of $\cJ_N$, we need to have good control on the number of excitations of wave functions in low-energy spectral subspaces of $\mathcal{J}_N$. This is summarized in the following proposition whose proof is deferred to Section \ref{sec:JN}.
	\begin{prop}\label{prop:hpNcubic}
	Assume \eqref{eq:asmptsVVext}, let $\eta$ be defined as in \eqref{eq:defeta}, $A$ as defined in \eqref{eq:A-def} and let $E_N$ denote the ground state energy of $H_N$. Let $\psi_N \in L^2_s 
(\mathbb{R}^{3N})$ with $\| \psi_N \| = 1$ belong to the spectral subspace of $H_N$ with energies below $E_N + \zeta$, for some $\zeta > 0$, i.e. 
	 \[ \psi_N = {\bf 1}_{(-\infty ; E_N + \zeta]} (H_N) \psi_N. \]
	Let $\xi_N =e^{-A} e^{-B} U_N \psi_N\in \cFpn$ be the renormalized excitation vector associated with $
\psi_N$. Then, for any $j \in\mathbb N$ there exists a constant $C > 0$ such that 
	\[ \langle \xi_N, (\cN +1)^j (\cH_N+1) \xi_N \rangle \leq C (1 + 
\zeta^{j+3}). \]
\end{prop}

\section{Bounds on $\sigma(H_N)$ and Proof of Theorem \ref{thm:main}}  \label{sec:diag}

In this section, we conclude the proof of Theorem \ref{thm:main} and determine the ground state energy $E_N$ of $H_N$ as well as the spectrum of $H_N-E_N$ below a threshold $\zeta>0$, up to errors that vanish in the limit $N\to\infty$. Our proof is based on Proposition \ref{prop:JN}, recalling that $\cJ_N = e^{-A}e^{-B}U_NH_NU_N^*e^Be^A$ is unitarily equivalent to $H_N$. For the rest of this section, we assume that the parameter $\ell \in (0;1)$ introduced in (\ref{eq:scatl}) is sufficiently small (by Lemma \ref{lm:bndseta}, this guarantees that $\| \eta \|$ is small enough). We also assume that the parameters introduced around 
(\ref{eq:def-wtk}) are fixed as $\eps = 6/13$ and $\tau =\eps/6$, so that we can apply the results of Proposition \ref{prop:GN} and Proposition \ref{prop:JN}. From \eqref{eq:errorJN}, we have \begin{equation}\label{eq:errorrho}\pm \cE_{\cJ_N} \leq C N^{-\rho} (\cH_N + 1) (\cN+1)^3,\end{equation}  with $\rho = 1/26$. 
		
To compute the spectrum of $H_N$, we proceed in two main steps, proving first lower and then upper bounds.
To show the lower bound, we start from (\ref{eq:JNdec}), we drop the non-negative potential energy $\cV_N$ in \eqref{eq:JNdec} and then we diagonalize the remaining quadratic Fock space Hamiltonian. This is discussed in the Section \ref{sec:lwrbnd}. For the upper bound, we need additionally to control the expectation of the potential energy on low energy states; this is discussed in Section \ref{sec:uprbnd}. Combining the two results, we obtain asymptotically matching lower and upper bounds on the min-max values of $\cJ_N$. To conclude Theorem \ref{thm:main}, it then only remains to verify Eq. \eqref{eq:gsenergy} which determines the ground state energy $E_N$ up to errors that vanish as $N\to\infty$. This is the content of Section \ref{sec:gsenergy}. 

Before we start with the lower bound, it is convenient to switch to the full excitation Fock space $\cF_{\bot\pn}$, replacing the modified creation and annihilation operators in \eqref{eq:cQcJN} by the standard ones. This will enable us to diagonalize the resulting quadratic Fock space Hamiltonian exactly. To this end, let us denote by 
		\begin{equation} \label{eq:deftildeQ}
		\widetilde{\mathcal{Q}}_{\mathcal{J}_N} = \int dx dy \, \widetilde{\Phi}(x,y) a_x^* a_y + \frac{1}{2}\int dx dy \ \widetilde{ \Gamma}(x,y) (a_x^* a_y^* + \text{h.c.})
		\end{equation}
the quadratic Fock space Hamiltonian that is obtained from $\cQ_{\cJ_N}$ replacing $b, b^*$ operators by the $a, a^*$ fields. We claim that
		\be \label{eq:QJNdec} \cQ_{\cJ_N} = \wt {\mathcal{Q}}_{\mathcal{J}_N} +  \wt{\mathcal{E}}_N    \ee
for an error $\wt{\cE}_N $ that is bounded, in the sense of forms in $\cFpn$, by
		\be\label{eq:QJNdecerror} \pm \wt{\cE}_N \leq  CN^{-1/2} ( \cK + \cV_\text{ext}  +\cN+1)(\cN+1). \ee
To prove the bound \eqref{eq:QJNdecerror}, consider first the non-diagonal contribution to $\cQ_{\cJ_N}$. From the definition (\ref{eq:bb-def}) and using that $|\sqrt{1-x/N}\sqrt{1-x/N-1/N}-1| \leq C x/N$ for $0\leq x\leq N$, we bound \begin{equation}\label{eq:Gammaa} \pm \bigg(\int dx dy \ \widetilde{ \Gamma}(x,y) ( b_x^* b_y^*-a_x^* a_y^* +\text{h.c.})\bigg) \leq CN^{-1}\| \wt \Gamma\|_\text{HS} (\cN+1)^2. \end{equation} 
In particular, this error is small if we show that 
		\be \label{eq:HStGamma} \| \wt \Gamma\|_\text{HS} \leq C \ee
for some constant $C>0$, independent of $N\in\mathbb{N}$. Here, we identify the kernel $\wt \Gamma$ with the corresponding operator on $L^2_{\bot \pn}(\bR^3)$ and $\|\cdot\|_\text{HS}$ denotes the Hilbert-Schmidt norm. To control $ \| \wt \Gamma\|_\text{HS}$, we consider the different contributions to $\wt \Gamma$ in \eqref{eq:deftildeGamma}. We first note that the operator with kernel $N^{3}(Vf_\ell)(N(x-y))\pn(x)\pn(y)$ is bounded, uniformly in $N$. Analogously, $\|  8\pi \mathfrak{a}_0\pn^2 \|_{op} \leq C$. From Lemma  \ref{lm:bndseta}, we have $\| \sigma \|_\text{HS} < \infty$, and also $\| \sigma V_\text{ext} \|_\text{HS} , \| \sigma \Delta \gamma - k \Delta \|_\text{HS} < \infty$.  Thus, to show (\ref{eq:HStGamma}), we only need to control the norm of the operator on $L^2_{\perp \ph_0} (\bR^3)$ with kernel 
\begin{equation}\label{eq:kern-3} \begin{split} 
 N^3(Vf_\ell)(N(x-y)) &\pn(x)\pn(y) + (\Delta_x + \Delta_y) \big( N w_\ell(N(x-y))\pn(x)\pn(y)\big)  \\
= \; &N^3 \left[ 2\Delta w_\ell (N(x-y)) + (Vf_\ell) (N(x-y)) \right] \pn (x) \pn (y) \\ &+ N w_\ell (N(x-y)) \big[ \Delta \pn (x) \pn (y) + \pn (x) \Delta\pn (y) \big] \\ &+ 2N^2 \nabla w_\ell (N(x-y)) \cdot \big[ \nabla \pn (x) \pn (y) - \pn (x) \nabla \pn (y) \big]. \end{split}\end{equation} 
With the scattering equation \eqref{eq:scatl} (and using the bounds in Lemma \ref{3.0.sceqlemma}), we can show that the Hilbert-Schmidt norm of the operator associated with the kernel on the first line is bounded, uniformly in $N$. The operator associated with the kernel on the second line is also Hilbert-Schmidt, with norm bounded uniformly in $N$. As for the kernel on the last line, we can write 
\[ \begin{split} 
\nabla_i \ph_0 (x) &\ph_0 (y) - \ph_0 (x) \nabla_i \ph_0 (y)  \\ &= \int_0^1 ds \, \frac{d}{ds} \, \nabla_i \ph_0 (sx + (1-s)y) \ph_0 ((1-s)x + s y)  \\ &=  \sum_{j=1}^3  (x_j - y_j) \int_0^1 ds \, \Big[ \big(\nabla_i \nabla_j \ph_0 \big) (sx + (1-s)y) \ph_0 ((1-s)x + s y) \\ &\hspace{5cm} + \big( \nabla_i \ph_0 \big) (sx + (1-s)y) \big(\nabla_j \ph_0\big) ((1-s)x + s y) \Big]. \end{split} \]
With the factor $(x-y)$ multiplying $N^2 (\nabla w_\ell) (N(x-y))$, and with the estimates \eqref{3.0.scbounds1}, \eqref{exponential decay}, we conclude that also the last term on the right hand side of (\ref{eq:kern-3}) has a bounded Hilbert-Schmidt norm. This concludes the proof of (\ref{eq:HStGamma}) and shows that the error (\ref{eq:Gammaa}) is small. Similarly, we can also bound the diagonal part of $\cQ_{\cJ_N}$. We find  
		\[ \pm \bigg(\int dx dy \ \widetilde{ \Phi}(x,y) ( b_x^* b_y   -   a_x^* a_y)  \bigg)  \leq CN^{-1}(\cK + \cV_\text{ext} + \cN+1)(\cN+1). \]
which completes the proof of \eqref{eq:QJNdecerror}.


\subsection{Lower Bound on $\sigma(H_N)$}\label{sec:lwrbnd}
In this section, we compare the min-max values of $\cJ_N$ with those of the quadratic operator $ \wt{\cQ}_{\cJ_N}$, defined through \eqref{eq:deftildeQ} as a quadratic form in the full excitation Fock space $\cF_{\bot \pn}$. To obtain a lower bound on the spectrum of $\cJ_N$, suppose that $ \lambda_n(\cJ_N) - E_N \leq  \zeta$. Then the min-max principle implies with \eqref{eq:JNdec}, \eqref{eq:errorrho}, \eqref{eq:QJNdec}, \eqref{eq:QJNdecerror} and Proposition \ref{prop:hpNcubic} 
		\be \label{eq:lowbndHN}\begin{split}
		\lambda_n(\cJ_N) 
		& = \kappa_{\cJ_N}+ \,\min_{\substack{Y \subset P_\zeta( \cFpn), \\ \dim(Y)=n}} \max_{\substack{ \xi \in Y: \\ \Vert \xi \Vert =1}} \langle \xi, (\wt{\cQ}_{\cJ_N} + \cV_N+ \cE_{\cJ_N} + \wt{\cE}_N) \xi\rangle \\
		& \geq \kappa_{\cJ_N} + \min_{\substack{Y \subset P_\zeta( \cFpn), \\ \dim(Y)=n}} \max_{\substack{ \xi \in Y: \\ \Vert \xi \Vert =1}} \langle \xi, \wt{\cQ}_{\cJ_N}   \xi\rangle - CN^{-\rho}(1+\zeta^{6}) \\
		&\geq \kappa_{\cJ_N}+ \min_{\substack{X \subset   \cF_{\bot \pn} , \\ \dim(Y)=n}} \max_{\substack{ \xi \in Y: \\ \Vert \xi \Vert =1}} \langle \xi,  \wt{\cQ}_{\cJ_N} \xi\rangle - CN^{-\rho}(1+\zeta^{6}) \\
		&= \kappa_{\cJ_N} + \lambda_n(\wt{\cQ}_{\cJ_N})  - CN^{-\rho}(1+\zeta^{6}). 
		\end{split}\ee
Here, we have abbreviated by $P_\zeta$ the spectral projection of $\cJ_N - E_N$ associated with the interval $(-\infty; \zeta]$. Up to the constant $\kappa_{\cJ_N}$ (and an error that vanishes as $N\to\infty$), we thus obtain a lower bound on the min-max values of $\cJ_N$ through those of $\wt{\cQ}_{\cJ_N}$. 

Let us now diagonalize the quadratic operator $\wt{\cQ}_{\cJ_N}$. We adapt here the arguments used in \cite{GS} for mean-field bosons. We recall the definitions of $H_\text{GP}$ and $E$ in \eqref{eq:defHGPE} and of $\wt{\Phi}$, $ \wt{\Gamma}$ in \eqref{eq:deftildePhi}, \eqref{eq:deftildeGamma}. We identify these operators with operators mapping $L^2_{\bot \pn}(\bR^3)$ back into itself and, by slight abuse of notation, we still write $H_\text{GP}, E,\wt \Phi, \wt \Gamma $ instead of $H_\text{GP}Q$, $EQ, Q\wt \Phi Q, Q\wt \Gamma Q$ (recall that $Q= 1-|\pn\rangle\langle\pn|$). Notice in particular that both $H_\text{GP}$ and $E$ are invertible in $L^2_{\bot \pn}(\bR^3)$. In addition to these operators, we also define, on $L_{\perp \pn}^2(\mathbb{R}^3)$,  
		\be \label{eq:deftildeDE} \begin{split}
		\wt{D} &= \widetilde{\Phi} - \widetilde{ \Gamma}  = e^{-\eta} H_\text{GP}\, e^{-\eta},\,\, \,\,\wt{E}  = \big( \wt{D}^{1/2} ( \wt{D} + 2\wt\Gamma) \wt{D}^{1/2}\big)^{1/2} \,\,\,\,\text{ with}\\
		    \wt D+2\wt \Gamma &= \widetilde{\Phi} + \widetilde{ \Gamma}  =  e^\eta \big( H_\text{GP} + 2 N^3(Vf_\ell)(N(x-y))\pn(x)\pn(y)\big) e^\eta 
		\end{split} \ee
as well as
		\be\label{eq:ABalpha} \begin{split}
		A = {\wt{D}}^{1/2} \wt{E}^{-1/2}, \hspace{0.5cm}  B= \big(A^{-1}\big)^* = \wt{D}^{-1/2} \wt{E}^{1/2} , \hspace{0.5cm} \alpha = \log\big( \vert A^* \vert\big).
		\end{split} \ee
Using the polar decomposition, we also denote by $W$ the partial isometry on $L_{\perp \pn}^2(\mathbb{R}^3)$ that is defined through
 		\be \label{eq:defW}
		A = W \vert A \vert = \vert A^* \vert W, \hspace{0.5cm} B= \vert B^* \vert W.
		\ee
The next lemma collects basic properties of the operators introduced above.  
\begin{lemma} \label{lm:alpha} Let $ \wt{E}, \wt{D}$ be defined in $L_{\perp \pn}^2(\mathbb{R}^3)$ by \eqref{eq:deftildeDE}. Let $\ell \in (0;1)$ be small enough and let $N$ be large enough. Then:
	\begin{enumerate}[a)]
		\item There exists $c>0$ such that $ \wt{E},  \wt{D}\geq c>0$. In particular, the operators $A, B, \alpha$ and $W$, defined in \eqref{eq:ABalpha}, \eqref{eq:defW} are well-defined in $L^2_{\bot \pn}(\bR^3)$. Moreover, $W$ is unitary. 
		\item There exist $c, C>0$ such that $ c \wt D^2 \leq \wt E^2 \leq C\wt D^2 $. In particular, $c^{1/2} \wt{D} \leq \wt{E} \leq C^{1/2} \wt{D}$. 
		\item There exists $C > 0$ such that $H_\text{GP} - C \leq \wt{D} \leq H_\text{GP} + C$.   
		\item We have $\|A-1\|_\text{HS}, \|B-1\|_\text{HS} \leq C$ for some $C>0$. 
		\item For every $\beta < 1$, there exists $C>0$ such that $\| \wt D^{\beta/2} \alpha \wt{D}^{\beta/2} \|_\text{HS} \leq C$. Moreover, there exists $C >0$ such that $\| \wt{D}^{1/2} \alpha \|_\text{HS} \leq C$ (from parts a), b),c), these estimates remain true if we replace factors of $\wt{D}$ by the identity, by factors of $\wt{E}$ or by factors of $H_\text{GP}$).  
	\end{enumerate}
\end{lemma}
\begin{proof}
We start with the proof of a). On $L^2_{\perp \pn} (\bR^3)$, we have $H_\text{GP} \geq c > 0$.Thus
\[ \langle \psi, \wt D\psi \rangle = \langle e^{-\eta }\psi, H_\text{GP} e^{-\eta}\psi\rangle \geq c  \| e^{-\eta}\psi\|^2 \geq c (1-C\|\eta\|) \|\psi\|^2 \]
for all $\psi \in L^2_{\bot \pn}(\bR^3)$. By Lemma \ref{lm:bndseta}, $\| \eta \| \leq C \ell^{1/2}$. Thus $D \geq c > 0$ follows choosing $\ell > 0$ small enough. To prove the claim for $\wt{E}$, we compare first the operator with kernel $N^3 V(N(x-y)) \pn (x) \pn (y)$ with $8\pi \frak{a}_0 \pn^2$. To this end, we observe that, for an arbitrary $\phi \in L^2_{\perp \pn} (\bR^3)$,  
\begin{equation} \label{eq:aux1}\begin{split}
&\left\vert \int dx dy \ \overline{\phi}(x) \phi (y) \pn(x) \pn(y) N^3 (Vf_\ell)(N(x-y)) - 8\pi\frak{a}_0 \int dx \ |\phi (x)|^2 \pn(x)^2 \right\vert \\
&\leq CN^{-1} \| \phi \|^2 + \left\vert  \int dx dy \ (Vf_\ell)(y) \overline{\phi (x)} \pn(x) \big (\phi (x+y/N)\pn(x+y/N) - \phi(x) \pn(x)\big)\right\vert \\
		&\leq C N^{-1} \| \phi \|^2 + C N^{-1} \| (-\Delta + 1)^{1/2} \phi \|^2  \leq C N^{-1} \| H_\text{GP}^{1/2} \phi \|^2 	\end{split}
\end{equation}
and thus
		\[\begin{split} \langle\psi, (\wt D+ 2\wt \Gamma)\psi\rangle&\geq  \langle e^{\eta} \psi,  \big( H_\text{GP}(1-C/N) + 16\pi \mathfrak{a}_0\pn^2\big) e^{\eta} \psi\rangle \\
		&\geq   \langle e^{\eta} \psi,  H_\text{GP}(1-C/N) e^{\eta} \psi\rangle \geq c(1-C\|\eta\|)\|\psi\|^2 \end{split} \]
for $N$ large enough. This shows that $ \wt E \geq c \wt D \geq c >0$. An immediate consequence is that $A, B, \alpha$ and $W$ as in \eqref{eq:ABalpha}, \eqref{eq:defW} are well-defined. Moreover, $W$ is a unitary map, because $\text{ker} (W) = \text{ker}(A) = \{0\}$ (recall that $\wt E, \wt D \geq c >0$ in $L^2_{\bot\pn}(\bR^3)$). 

Let us switch to part b). We have 
\[ \wt{E}^2 = \wt{D}^{1/2} (\wt{D} + 2 \wt{\Gamma}) \wt{D}^{1/2} = \wt{D}^2 + 2 \wt{D}^{1/2} \wt{\Gamma} \wt{D}^{1/2}. \]
With (\ref{eq:HStGamma}) (and with part a)), we conclude that 
\[ \wt{E}^2 \leq \wt{D}^2 + C \wt{D} \leq C \wt{D}^2 \]
and also that $\wt{E}^2 \geq \wt{D}^2 - C \wt{D}$, which implies (from the proof of  a), we have $\wt{E} \geq c \wt{D}$) 
\[ \wt{D}^2 \leq \wt{E}^2 + C \wt{D} \leq \wt{E}^2 + C \wt{E} \leq C \wt{E}^2. \]

As for part c), we observe that 
\[ \wt{D} = e^\eta H_\text{GP} e^\eta = H_\text{GP} + (e^\eta -1) H_\text{GP} + H_\text{GP} (e^\eta- 1) + (e^\eta -1) H_\text{GP} (e^\eta -1) \]
and that, from Lemma \ref{lm:bndseta}, $(e^\eta - 1) H_\text{GP} + H_\text{GP} (e^\eta- 1)$ and $(e^\eta -1) H_\text{GP} (e^\eta -1)$ are bounded operators (the kernel of $k (-\Delta) + (-\Delta) k$ has the form $(\Delta_x + \Delta_y) (Nw_\ell (N(x-y)) \pn (x) \pn (y)$ and can be handled as in \eqref{eq:kern-3}, since the term $N^3 (Vf_\ell) (N(x-y)) \pn (x) \pn (y)$ corresponds to a bounded operator).

Let us now prove part d). We focus on the operator $A-1$. Using functional calculus (as in \cite[Lemma 3]{GS}) we write
$$
A- 1 = (\wt D^{1/2}-\wt E^{1/2}) \wt E^{-1/2} = 2\sqrt 2\pi\int_0^\infty dt\, t^{1/4}  \frac{\wt D^{1/2}}{t+\wt D^2}   \wt \Gamma   \wt D^{1/2} \frac{ \wt E^{-1/2}}{t+\wt E^2}.$$
With (\ref{eq:HStGamma}) and with the bounds $\| \wt{D}^{1/2} / (t + \wt{D}^2) \|_\text{op} \leq C (t + 1)^{-3/4}$, $\| (t + \wt{E}^2)^{-1}  \|_\text{op}  \leq C (t+1)^{-1}$, $\| \wt{D}^{1/2} \wt{E}^{-1/2} \|_\text{op} \leq C$ (by part b), we obtain 
\[ \| A-1 \|_\text{HS} \leq C  \int_0^\infty dt\, t^{1/4} \Big\| \frac{\wt D^{1/2}}{t+\wt D^2} \Big\|_\text{op} \big\| \wt \Gamma  \big\|_\text{HS} \| \wt D^{1/2} \wt{E}^{-1/2} \|_\text{op} \Big\| \frac{1}{t+\wt E^2} \Big\| \leq C \,. \]

Finally, let us prove part e). From part b), we have $0 < c \leq A A^* \leq C$. For $x \in [c;C]$, we can  bound \[ x-1 - C (x-1)^2 \leq \log x \leq x-1 + C (1-x)^2.\]  Thus 
\begin{equation} \label{eq:DaD} \| \wt{D}^{\beta/2} \alpha \wt{D}^{\beta/2} \|_\text{HS} \leq C \| \wt{D}^{\beta/2} (1- AA^*) \wt{D}^{\beta/2} \|_\text{HS} + C  \| \wt{D}^{\beta/2} (1- AA^*)^2 \wt{D}^{\beta/2} \|_\text{HS}. \end{equation} 
With functional calculus, we find
\[ \begin{split}  \wt{D}^{\beta/2} (1 - A A^*) \wt{D}^{\beta/2} &= \wt{D}^{(\beta+1)/2} (\wt{D}^{-1} - \wt{E}^{-1} ) \wt{D}^{(\beta +1)/2} \\ &= C \int_0^\infty \frac{dt}{t^{1/2}} \frac{\wt{D}^{1+\beta/2}}{t + \wt{D}^2} \wt{\Gamma} D^{1/2} \frac{1}{t + \wt{E}^2} \wt{D}^{(\beta+1)/2}. \end{split} \]
Thus, with (\ref{eq:HStGamma}), $\| \wt{D}^{1+\beta/2} / (t + \wt{D}^2) \|_\text{op}, \| \wt{E}^{1+\beta/2} / (t + \wt{E}^2) \|_\text{op} \leq C (t+1)^{-1/2+\beta/4}$ and $\| \wt{D}^{1/2} \wt{E}^{-1/2} \|_\text{op} ,  \| \wt{E}^{-(\beta+1)/2} \wt{D}^{(\beta+1)/2} \|_\text{op} \leq C$ (by part b)), we find 
\[ \|  \wt{D}^{\beta/2} (1 - A A^*) \wt{D}^{\beta/2} \|_\text{HS} \leq C \int_0^\infty dt \, t^{-1/2} (1+t)^{-1+\beta/2} \leq C \]
for any $\beta < 1$. Since $\| \wt{D}^{\beta/2} (1- AA^*)^2 \wt{D}^{\beta/2} \|_\text{HS} \leq C \| \wt{D}^{\beta/2} (1- AA^*) \wt{D}^{\beta/2} \|^2_\text{HS} \leq C$, it follows from (\ref{eq:DaD}) that $\| \wt{D}^{\beta/2} \alpha \wt{D}^{\beta/2} \|_\text{HS} \leq C$, for all $\beta <1$. The bound for $\| \wt{D}^{1/2} \alpha \|_\text{HS}$ can be proven similarly.  
\end{proof}
Using the operators defined above, we can now construct a suitable Bogoliubov transformation that diagonalizes $\wt{\cQ}_{\cJ_N}$. To this end, we denote by $(\varphi_j)_{j\in \mathbb{N}}$ the eigenbasis of $\wt E$, we set $a^\sharp_j  = a^\sharp (\varphi_j)$ for $\sharp\in \{\cdot, *\}$, $\alpha_{ij} = \langle \varphi_i, \alpha\, \varphi_j\rangle$ and we define
		\be\label{eq:defWXU} \cW = \Gamma(W), \hspace{0.5cm} X  = \frac12 \sum_{i, j =1}^\infty \alpha_{ij} a^*_i a^*_j -\text{h.c.}, \hspace{0.5cm} \cU = e^{X} \cW.   \ee
Note in particular that $\cU : \cF_{\bot \pn}\to \cF_{\bot \pn}$ and recall that $\cU$ acts as
		\[\begin{split} 
		&\cU^* a^*(f) \cU\\ 
		&= \frac12 a^*( W^*|A^*| f) + \frac12 a^*(W^*|A^*|^{-1}f) +\frac12 a( W^*|A^*| \bar f) - \frac12 a^*(W^*|A^*|^{-1}\bar f)\\
		&= \frac12 a^* ( \wt E^{-1/2}\wt D^{1/2} f) + \frac12 a^* ( \wt E^{1/2}\wt D^{-1/2}f )  + \frac12 a ( \wt E^{-1/2}\wt D^{1/2}\bar f ) - \frac12 a ( \wt E^{1/2}\wt D^{-1/2} \bar f) 
		\end{split}\]
for $f\in L^2_{\bot \pn}(\bR^3)$, by the definitions in \eqref{eq:ABalpha}. From \cite{GS}, we recall that $\cU$ is constructed so that $\cU^* \wt{\cQ}_{\cJ_N} \cU $ is diagonal. Indeed, representing $\wt{\cQ}_{\cJ_N}$ in $\cF_{\bot \pn }$ as
		\[\wt{\cQ}_{\cJ_N} = \sum_{i,j=1}^\infty \langle \varphi_i, \wt \Phi  \, \varphi_j\rangle \,a^*_i a_j +\frac12 \sum_{i,j=1}^\infty \langle \varphi_i, \wt \Gamma\,  \varphi_j\rangle \big(a^*_i a^*_j+\text{h.c.}\big),  \] 
a lengthy, but straightforward calculation verifies that 
		\be\label{eq:diag}\begin{split}
		\cU^* \wt{\cQ}_{\cJ_N} \cU & =  \frac12 \tr_{\bot \pn} \left( \frac{1}{2}(\wt{D}^{1/2} \wt{E} \wt{D}^{-1/2} + \wt{D}^{-1/2} \wt{E} \wt{D}^{1/2}) - \wt D - \wt \Gamma\right) \\&\hspace{0.5cm}+\sum_{i,j=1}^\infty \frac12 \big(\wt \Phi_{ij}+ \wt \Gamma_{ij}\big) a^*(\wt E^{-1/2}\wt D^{1/2}\varphi_i) a(\wt E^{-1/2}\wt D^{1/2}\varphi_j)\\
		&\hspace{0.5cm} + \sum_{i,j=1}^\infty \frac12 \big(\wt \Phi_{ij}- \wt \Gamma_{ij}\big) a^*(\wt E^{1/2}\wt D^{-1/2}\varphi_i) a(\wt E^{1/2}\wt D^{-1/2}\varphi_j)  \\
		&  =   \frac12 \tr_{\bot \pn} \big( \frac{1}{2}(\wt{D}^{1/2} \wt{E} \wt{D}^{-1/2} + \wt{D}^{-1/2} \wt{E} \wt{D}^{1/2}) - \wt D - \wt \Gamma\big) + d\Gamma(\wt E),
		\end{split}\ee
where $ \tr_{\bot \pn} $ denotes the trace in $L^2_{\bot \pn}(\bR^3)$. As a consequence of \eqref{eq:lowbndHN}, we obtain that if $\lambda_n(\cJ_N)-E_N\leq \zeta$, then 
		\be\label{eq:lwrbndeq2}\begin{split}
		\lambda_n(\cJ_N) \geq  \kappa_{\cJ_N} &+ \frac12 \tr_{\bot \pn} \left( \frac{1}{2}\left(\wt{D}^{1/2} \wt{E} \wt{D}^{-1/2} + \wt{D}^{-1/2} \wt{E} \wt{D}^{1/2}\right) - \wt D - \wt \Gamma\right) \\&+ \lambda_n\big( d\Gamma(\wt E)\big) - CN^{-\rho}(1+\zeta^6) . 
		\end{split}\ee
To obtain a lower bound for the eigenvalues of $\cJ_N$ (and hence of $H_N$) matching the statement of Theorem 
\ref{thm:main}, we still have to compare the operator $\wt E$, defined in \eqref{eq:deftildeDE}, with the operator $E$, defined in \eqref{eq:defHGPE}. This is done in the next lemma.
\begin{lemma}\label{lm:tildeEE} 
On $L^2_{\bot \pn}(\bR^3)$ we have that 
\[\begin{split} &(1-C/N) \big( \wt D^{1/2} e^{\eta}( H_\text{GP}+ 16\pi \mathfrak{a}_0 \pn^2 )e^\eta \wt D^{1/2}\big)^{1/2} \leq  \wt E \\ &\hspace{6cm} \leq (1+C/N) \big( \wt D^{1/2} e^{\eta}( H_\text{GP}+ 16\pi \mathfrak{a}_0 \pn^2 )e^\eta\wt D^{1/2}\big)^{1/2}   \end{split}\]
for a constant $C>0$, independent of $N$. Moreover, the operator $E$ defined in \eqref{eq:defHGPE} is unitarily equivalent to $\big( \wt D^{1/2} e^{\eta}( H_\text{GP}+ 16\pi \mathfrak{a}_0 \pn^2 )e^\eta\wt D^{1/2}\big)^{1/2}  $. We conclude that   
		\[ (1-C/N) \lambda_n ( E )\leq \lambda_n (\wt E) \leq (1+C/N) \lambda_n (E). \]
\end{lemma}
\begin{proof}
Using the bound \eqref{eq:aux1}, we see that 
		\[ (1-C/N) e^\eta(H_\text{GP} + 16\pi \mathfrak{a}_0\pn^2)e^\eta\leq ( \wt D + 2\wt \Gamma) \leq  (1+C/N) e^\eta(H_\text{GP} + 16\pi \mathfrak{a}_0\pn^2)e^\eta,\]
which immediately implies the first claim. The statement about the unitary equivalence of $\big( \wt D^{1/2} e^{\eta}( H_\text{GP}+ 16\pi \mathfrak{a}_0 \pn^2 )e^\eta\wt D^{1/2}\big)^{1/2} $ and $E$ follows if we write 
		\[ \wt D^{1/2}  = | H_\text{GP}^{1/2}e^{-\eta}| = U_0^* H_\text{GP}^{1/2}e^{-\eta} = e^{-\eta}H_\text{GP}^{1/2} U_0 \] 
for some partial isometry $U_0$, by polar decomposition. Since $e^{-\eta}$ is bijective and $H_\text{GP}^{1/2}$ is injective, we see that $\text{ker} \big(H_\text{GP}^{1/2}e^{-\eta}\big)= \{0\}$ and hence $U_0$ is a unitary map. This implies that
		$\big( \wt D^{1/2} e^{\eta}( H_\text{GP}+ 16\pi \mathfrak{a}_0 \pn^2 )e^\eta\wt D^{1/2}\big)^{1/2} =U_0^* E U_0$.
\end{proof}
With Lemma \ref{lm:tildeEE} and Eq. \eqref{eq:lwrbndeq2}, we conclude for $\lambda_n(\cJ_N)-E_N \leq \zeta$ that 
		\be\label{eq:lwrbndsHN}\begin{split}
		\lambda_n(\cJ_N) &\geq \kappa_{\cJ_N} + \frac12 \tr_{\bot \pn} \left( \frac{1}{2}\left(\wt{D}^{1/2} \wt{E} \wt{D}^{-1/2} + \wt{D}^{-1/2} \wt{E} \wt{D}^{1/2} \right) - \wt D - \wt \Gamma\right)\\&\hspace{0.5cm} + \lambda_n\big( d\Gamma(E)\big) - CN^{-\rho}(1+\zeta^6)\\
		&\hspace{0.5cm} - CN^{-1} \Big( - \frac12 \tr_{\bot \pn} \left[ \frac{1}{2}\left(\wt{D}^{1/2} \wt{E} \wt{D}^{-1/2} + \wt{D}^{-1/2} \wt{E} \wt{D}^{1/2}\right) - \wt D - \wt \Gamma \right]  \\
		&\hspace{2.5cm}+\lambda_n(\cJ_N) - \kappa_{\cJ_N} \Big).
		\end{split}\ee		
In the next two sections, we will prove that the ground state energy of $\cJ_N$ is equal to $E_N =  \kappa_{\cJ_N} + \frac12 \tr_{\bot \pn} \left( \frac{1}{2}(\wt{D}^{1/2} \wt{E} \wt{D}^{-1/2} + \wt{D}^{-1/2} \wt{E} \wt{D}^{1/2}) - \wt D - \wt \Gamma\right)+\mathcal{ O}\big(N^{-\rho}(1+\zeta^6)\big)$. This will imply that the last term on the right hand side in \eqref{eq:lwrbndsHN} is a negligible error of order $\mathcal{O}(  N^{-1}\zeta)$. The bound \eqref{eq:lwrbndsHN} then also shows that the spectrum of $\cJ_N$ above $E_N$ is bounded from below by that of $d\Gamma(E)$, up to errors that vanish in the limit $N\to \infty$. Note that the eigenvalues of $d\Gamma(E)$ are explicitly given by finite sums of the form
		\[ \sum_{j=1}^\infty n_j e_j, \]
where the $(e_j)_{j\in \mathbb{N}} $ denote the eigenvalues of the operator $E$ in $L^2_{\bot \pn}(\bR^3)$, defined in \eqref{eq:defHGPE}, and $n_j\in \mathbb{N}$ are non-zero for finitely many $j\in \mathbb{N}$.

\subsection{Upper Bound on $\sigma(H_N)$}  \label{sec:uprbnd}
In this section, we show that the spectrum of $d\Gamma(E)$ provides an upper bound on the spectrum of $\cJ_N$ above its ground state energy. To show this, we construct suitable trial states that approximately diagonalize $\wt \cQ_{\cJ_N}$, defined in \eqref{eq:deftildeQ}, and we apply Lemma \ref{lm:tildeEE}. Since $\wt \cQ_{\cJ_N}$ is defined in the full excitation space $\cF_{\bot \pn}$, we can not directly choose its eigenvectors as trial states, but have to truncate them first to $\cFpn$. Furthermore, in contrast to the lower bound, we also have to evaluate the potential energy $\cV_N$ of the trial states and show that it is negligible in the limit $N\to \infty$. 

We start the analysis with some preliminary results on the conjugation of $\cN$, $\cK+\cV_\text{ext} $ and $\cV_N$ by the unitary map $\cU$, defined in \eqref{eq:defWXU}. 
\begin{lemma}\label{lm:U*KNU}
	For every $j\in \mathbb{N}$, there exists a constant $C>0$ such that in $\mathcal{F}_{\perp \pn}$ we have
		\be \label{eq:U*NU} \mathcal{U}^* (\mathcal{N}+1)^j \mathcal{U} \leq C (\mathcal{N}+1)^j. \ee
Moreover, with $\wt E$ as defined in \eqref{eq:deftildeDE}, we have that  
		\be \label{eq:U*dGEU}  \mathcal{U}^* d\Gamma(\wt E) (\mathcal{N}+1)^j \mathcal{U} \leq C (d\Gamma(\wt E)+\cN +1)(\mathcal{N}+1)^j \ee
and consequently that 
		\be \label{eq:U*KU}\mathcal{U}^* (\cK+\cV_\text{ext}) (\mathcal{N}+1)^j \mathcal{U} \leq C (d\Gamma(\wt E)+\cN +1)(\mathcal{N}+1)^j . \ee
\end{lemma}
\begin{proof}
Proceeding as in \cite[Lemma 3.1]{BS} and using that $\|\alpha\|_\text{HS}\leq C$, we readily find (recalling the notation introduced in (\ref{eq:defWXU})), that
		\[ \sup_{s\in [-1;1] } \,e^{-sX} (\cN+1)^j e^{sX} \leq C (\cN+1)^j.  \]
The bound \eqref{eq:U*NU} then follows from $\cN = d\Gamma(1)$ so that $ \cW^*( \cN+1)^j \cW  = (\cN+1)^j$.

Next, we note that, by \eqref{eq:U*dGEU} and Lemma \ref{lm:alpha} (part c)), we have
		\[ \mathcal{U}^* (\cK+\cV_\text{ext}) (\mathcal{N}+1)^j \mathcal{U} \leq C \, (d\Gamma(\wt E)+1)(\mathcal{N}+1)^j.\]
Hence, \eqref{eq:U*KU} is a consequence of \eqref{eq:U*dGEU}. To prove the latter, we expand for $t\in [-1;1]$
		\[
		e^{-tX} d\Gamma(\wt E) (\cN+1)^j e^{tX}  =  d\Gamma(\wt E)(\cN +1)^j +  \int_0^t ds \, e^{-sX}[ d\Gamma(\wt E) (\cN+1)^j, X ] e^{sX}.
		\]
Denoting (as we did around (\ref{eq:defWXU})) by $(\varphi_k)_{k\in\mathbb{N}}$ the eigenbasis of $\wt E$ and setting $a^{\sharp}_k = a^\sharp(\varphi_k)$, we find 
		\[ [d\Gamma(\wt E)(\cN+1)^j, X ]  = \bigg( \sum_{k,l=1}^\infty \lambda_k(\wt E) \alpha_{kl} a^*_k a^*_l  + \text{h.c.}\bigg)(\cN+1)^j + d\Gamma(\wt E) [ (\cN+1)^j, X]. \]
An application of Lemma \ref{lm:alpha} (part e)) and Cauchy-Schwarz shows that 
		\[\begin{split}
		   \sum_{k,l=1}^\infty \lambda_k(\wt E) \big| \alpha_{kl} \langle \xi, a^*_k a^*_l  (\cN+1)^j \big)\xi \rangle \big|
		& \leq \langle\xi, \big( d\Gamma(\wt E)    + \| \wt E^{1/2}\alpha\|_\text{HS}^2  \, (\cN+1)\big)(\cN+1)^j\xi\rangle \\
		& \leq \langle\xi, d\Gamma(\wt E) (\cN+1)^j\xi\rangle  + C \langle\xi, (\cN+1)^{j+1}\xi\rangle
		\end{split}\]		
for all $\xi\in \cF_{\bot \pn}$. Moreover, observing that
		\[ [ (\cN+1)^j, X] =  \sum_{k,l=1}^\infty   \alpha_{kl} a^*_k a^*_l  \,j\, (\cN+\Theta(\cN)+1)^{j-1} +\text{h.c.}\]	
for some $\Theta: \mathbb{N}\to (0;1)$ by the mean value theorem, we conclude similarly that
		\[ | \langle\xi,  d\Gamma(\wt E) [ (\cN+1)^j, X]\xi \rangle| \leq  C\langle\xi, d\Gamma(\wt E) (\cN+1)^j\xi\rangle  + C \langle\xi, (\cN+1)^{j+1}\xi\rangle.  \]	
Applying Gronwall to the map $t\mapsto e^{-tX}d\Gamma(\wt E) (\cN+1)^je^{tX} $ and using \eqref{eq:U*NU}, we find
\begin{equation}\label{eq:gro-EN} \begin{split}
		\sup_{s\in [-1,1] } e^{-sX}  d\Gamma(\wt E)   (\cN+1)^j e^{sX} &\leq C (d\Gamma(\wt E)+ \cN +1)  (\cN+1)^j. \end{split} \end{equation}
To conclude the bound \eqref{eq:U*dGEU}, we then notice that 
		\begin{equation}\label{eq:WENW} \cW^*d\Gamma(\wt E) (\cN+1)^j \cW = d\Gamma( W^* \wt E \,W) (\cN+1)^j \end{equation} 
and that $W^* \wt{E} W \leq C \wt{E}$. From (\ref{eq:defW}) and using Lemma \ref{lm:alpha}, part b), this last bound follows if we can show that 
\begin{equation}\label{eq:WEW-bd} \Big\| \wt{E}^{-1/2} \frac{1}{\sqrt{\wt{E}^{-1/2} \wt{D} \wt{E}^{-1/2}}} \wt{E}^{1/2} \Big\|_\text{op} \leq C \end{equation}
uniformly in $N$. To prove (\ref{eq:WEW-bd}), we observe first that 
\[  \wt{E}^{-1/2} (\wt{E}-\wt{D}) = C \int_0^\infty dt \, \sqrt{t}  \,\wt{E}^{-1/2} \wt{D}^{1/2} \frac{1}{t + \wt{D}^2}  \wt{\Gamma}\wt{D}^{1/2} \frac{1 }{t+\wt{E}^2}  \]
which implies that $\| \wt{E}^{-1/2} (\wt{E}-\wt{D}) \|_\text{op} \leq C$. Then we write 
\[ \wt{E}^{-1/2} \frac{1}{\sqrt{\wt{E}^{-1/2} \wt{D} \wt{E}^{-1/2}}} \wt{E}^{1/2} -1 = C \int_0^\infty \frac{dt}{t^{1/2}} \frac{1}{t+1} \, \wt{E}^{-1/2} \frac{1}{t+\wt{E}^{-1/2} \wt{D} \wt{E}^{-1/2}} \wt{E}^{-1/2} ( \wt{E} -\wt{D} ) \]
to conclude that 
\[ \Big\| \wt{E}^{-1/2} \frac{1}{\sqrt{\wt{E}^{-1/2} \wt{D} \wt{E}^{-1/2}}} \wt{E}^{1/2} -1 \Big\|_\text{op} \leq C \]
and therefore to obtain (\ref{eq:WEW-bd}). 
\end{proof}	

As mentioned at the beginning of this section, below we have to evaluate the potential energy of suitable trial states. To this end, we make use of the following lemma that controls the conjugation of $\cV_N$ by the unitary map $\cU$, defined in \eqref{eq:defWXU}. 
\begin{lemma}\label{lm:U*VNU}
There exists $C>0$ such that in $\cF_{\bot \pn}$ we have 
		$$ 
		\cU^* \,\cV_N\,  \cU \leq CN^{-1/2}  (d\Gamma (\wt E)+1)^2. $$
\end{lemma}
\begin{proof} Recall from (\ref{eq:defWXU}) that $\cU = e^X \cW$. For $t \in [-1;1]$, we have 
\begin{equation}\label{eq:VNX0} \begin{split}
		e^{-tX} \cV_N e^{tX} =&\, \cV_N + \int_0^t ds\, \int dxdy\, N^2V(N(x-y)) \alpha(x;y) \big( e^{-sX} a_xa_y e^{sX}+ \text{h.c.}\big) \\
		& +  \int_0^t ds\, \int dxdy\, N^2V(N(x-y))   \big( e^{-sX}a^*_y a (\alpha_x) a_xa_ye^{sX} + \text{h.c.}\big). 
		\end{split}\end{equation} 
With the operator inequality  
\begin{equation}\label{eq:ESY} v (x-y) \leq C \|v\|_1 (-\Delta_x+1)^\beta (-\Delta_y+1)^\beta \end{equation} 
for any $\beta > 3/4$, we find 
\[  \begin{split} 
\Big| \int dx dy \, N^2 &V (N(x-y)) \alpha (x;y) \langle \psi, a_x a_y \psi \rangle \Big| \\ \leq \; &C  \| \psi \| \| \cV^{1/2}_N \psi \|  \left[ \int dx dy N^2 V (N(x-y)) |\alpha (x;y)|^2 \right]^{1/2} \\  \leq \; &C N^{-1/2} \| (1-\Delta)^{\beta /2} \alpha (1-\Delta)^{\beta /2} \|_\text{HS}   \| \psi \| \| \cV^{1/2}_N \psi \|.  \end{split} \]
 Using Lemma \ref{lm:alpha}, part c) and then part e), to estimate $\| (1-\Delta)^{\beta/2} \alpha (1-\Delta)^{\beta/2} \|_\text{HS} \leq \| \wt{D}^{\beta/2} \alpha \wt{D}^{\beta/2} \|_\text{HS} \leq C$, we obtain 
 \begin{equation}\label{eq:VNX1} \Big| \int dx dy \, N^2 V (N(x-y)) \alpha (x;y) \langle \psi, a_x a_y \psi \rangle \Big| \leq C \langle \psi, \cV_N \psi \rangle + C N^{-1} \| \psi \|^2. \end{equation}
 Similarly, we can estimate 
 \begin{equation}\label{eq:VNX2} \begin{split} 
 \Big| \int dx dy \, N^2 &V (N(x-y)) \langle \psi, a^*_y a (\alpha_x) a_xa_y \psi \rangle \Big| \\ \leq \; &C  \| \cV^{1/2}_N \psi \|  \left[ \int dx dy N^2 V (N(x-y)) \| a^* (\alpha_x) a_y  \psi \|^2 \right]^{1/2} \\  \leq \; &C N^{-1/2}   \| (1-\Delta)^{1/2} \alpha \|_\text{HS} \| \cV^{1/2}_N \psi \| \| (\cK + \cN)^{1/2} (\cN+1)^{1/2} \psi \|  \\ \leq \; &C \langle \psi, \cV_N \psi \rangle + C N^{-1} \langle \psi, (\cK+ \cN) (\cN+1) \psi \rangle. \end{split} \end{equation} 
Inserting (\ref{eq:VNX1}), (\ref{eq:VNX2}) in (\ref{eq:VNX0}), using Lemma \ref{lm:alpha} to replace $\cK$ with $d\Gamma (\wt{E})$ and then applying (\ref{eq:gro-EN}), we arrive at
\[ \begin{split} \langle \psi, e^{-tX} \cV_N e^{tX} \psi \rangle \leq \; &C \int_0^t ds  \langle e^{-sX} \psi, \cV_N e^{sX} \psi \rangle +  \langle \psi, \cV_N \psi \rangle \\ &+ C N^{-1/2}  \langle \psi, (d\Gamma (\wt{E}) +1) (\cN+1) \psi \rangle. \end{split} \]
With Gronwall, we obtain 
\[ \langle \psi, e^{-X} \cV_N e^{X} \psi \rangle \leq C \langle \psi, \cV_N \psi \rangle + C N^{-1/2}  \langle \psi, (d\Gamma (\wt{E}) +1) (\cN+1) \psi \rangle. \]
Since $\cU = e^X \cW$, we find, using (\ref{eq:WENW}), 
\begin{equation}\label{eq:UVU-1} \begin{split}  \langle \psi, \cU^* \cV_N \cU \psi \rangle &\leq C \langle \psi, \cW^* \cV_N \cW \psi \rangle + C N^{-1/2}  \langle \psi, \cW^*  (d\Gamma (\wt{E}) +1) (\cN+1) \cW \psi \rangle \\ &\leq C \langle \psi, \cW^* \cV_N \cW \psi \rangle + C N^{-1/2}  \langle \psi,  (d\Gamma (\wt{E}) +1) (\cN+1) \psi \rangle. \end{split}  \end{equation} 
 To estimate the expectation of $\cW^* \cV_N \cW$, we observe that, on the $n$-particle sector, using (\ref{eq:ESY}) and Lemma \ref{lm:alpha}, $\cW = \Gamma (W)$ and the bound $W^* \wt{E} W \leq C \wt{E}$ established after (\ref{eq:WENW}), 
 \[ \begin{split} \langle \psi^{(n)} , \cW^* \cV_N \cW \psi^{(n)} \rangle &= \sum_{i<j}^n \langle \cW \psi^{(n)} , N^2 V (N(x_i -x_j)) \cW \psi^{(n)} \rangle \\  &\leq C N^{-1} \sum_{i<j}^n \langle \cW \psi^{(n)} , (1-\Delta_{x_i}) (1-\Delta_{x_j}) \cW \psi^{(n)} \rangle 
 \\ &\leq C N^{-1} \sum_{i<j}^n \langle \psi^{(n)} , (W^* \wt{E} W)_i  (W^* \wt{E} W)_j  \psi^{(n)} \rangle 
 \\ &\leq C N^{-1} \sum_{i<j}^n \langle \psi^{(n)} , \wt{E}_i  \wt{E}_j  \psi^{(n)} \rangle 
 \leq C N^{-1} \langle \psi^{(n)}  d\Gamma (\wt{E})^2 \psi^{(n)} \rangle. \end{split} \]
 From (\ref{eq:UVU-1}), we obtain 
\[ \langle \psi, \cU^* \cV_N \cU \psi \rangle \leq C N^{-1/2} \langle \psi, (d\Gamma (\wt{E}) +1)^2 \psi \rangle. \]
\end{proof} 

Equipped with Lemmas \ref{lm:U*KNU} and \ref{lm:U*VNU}, let us now bound the spectrum $\sigma(H_N)$ from above. Like for the lower bound, we make use of the min-max principle and we derive upper bounds on the eigenvalues of $\cJ_N$. We first control the ground state energy $E_N$ and afterwards the spectrum of $\cJ_N$ above $E_N$. 

Recalling the definition of the unitary map $\cU: \cF_{\bot \pn} \to \cF_{\bot \pn}$ in \eqref{eq:defWXU}, the ground state energy of $\cJ_N$ is bounded from above by
		\[ E_N= \lambda_0(\cJ_N) \leq \frac{ \langle \chi_{ \leq N} \cU \Omega, \cJ_N \chi_{ \leq N} \cU \Omega \rangle }{\langle \chi_{ \leq N}\cU \Omega, \chi_{ \leq N}\cU \Omega\rangle }.\]
Here, $\Omega = (1,0,0,\dots) \in \cF_{\bot\pn}$ denotes the vacuum and we abbreviate $\chi_{\leq N} = 1 (\cN\leq N)$. Similarly, we write in the following $\chi_{> N} = 1(\cN>N)$ so that $ 1 = \chi_{\leq N} + \chi_{> N}$ in $\cF_{\bot \pn}$. By Markov's inequality and \eqref{eq:U*NU}, we note that
		\[  \| \chi_{ \leq N}\cU \Omega \|^2= 1 -  \| \chi_{ > N}\cU \Omega \|^2 \geq 1- N^{-2} \| \cN  \cU \Omega\|^2\geq 1- CN^{-2}.  \]
Using \eqref{eq:JNdec}, \eqref{eq:errorrho}, \eqref{eq:QJNdec}, \eqref{eq:QJNdecerror}, the fact that $\cJ_N\geq 0$ and that $\cK, \cV_\text{ext},\cV_N\geq 0$ commute with $\cN$, we obtain 
 		\[\begin{split}
		E_N &\leq   (1-C/N^2)^{-1} \big ( \kappa_{\cJ_N} + \langle \chi_{ \leq N} \cU \Omega,\wt \cQ_{ \cJ_N} \chi_{ \leq N} \cU \Omega \rangle\big) + C\, \langle  \cU \Omega, \cV_N \cU \Omega \rangle\\
		&\hspace{0.5cm}+ CN^{-\rho}  \langle \cU \Omega, (\cH_N +1)(\cN+1)^3\cU \Omega\rangle. 
		\end{split}\]
Applying Lemmas \ref{lm:U*KNU} and \ref{lm:U*VNU}, this implies
		\[\begin{split}
		E_N &\leq (1-C/N^2)^{-1} \big ( \kappa_{\cJ_N} + \langle \chi_{ \leq N} \cU \Omega, \wt{\cQ}_{\cJ_N} \chi_{ \leq N} \cU \Omega \rangle\big)  + C N^{-\rho}. 
		\end{split}\]
To evaluate this further, we use that in the sense of forms in $\cF_{\bot \pn}$, we have 
		\be \label{eq:chiQN-QN} \pm\big( \chi_{ \leq N} \wt \cQ_{\cJ_N} \chi_{ \leq N} - \wt \cQ_{\cJ_N} \big) \leq CN^{-1} (\cK+\cV_\text{ext}+\cN+1)(\cN+1) \ee
which is a straightforward consequence of the bound \eqref{eq:HStGamma} and Markov's inequality. Combining this with the identity \eqref{eq:diag}, we find that 
		\[\begin{split} E_N &\leq (1-C/N^2)^{-1}\left( \kappa_{\cJ_N} + \frac12  \tr_{\bot \pn} \left( \frac{1}{2}(\wt{D}^{1/2} \wt{E} \wt{D}^{-1/2} + \wt{D}^{-1/2} \wt{E} \wt{D}^{1/2}) - \wt D - \wt \Gamma\right) \right) + C N^{-\rho}\\
		& \leq \kappa_{\cJ_N} + \frac12  \tr_{\bot \pn} \left( \frac{1}{2}(\wt{D}^{1/2} \wt{E} \wt{D}^{-1/2} + \wt{D}^{-1/2} \wt{E} \wt{D}^{1/2}) - \wt D - \wt \Gamma\right) + C N^{-\rho}.
		\end{split}\]
Note that in the last step we used  $\kappa_{\cJ_N} +\frac12 \tr_{\bot \pn} \left( \frac{1}{2}\wt{D}^{1/2} \wt{E} \wt{D}^{-1/2} + \frac{1}{2}\wt{D}^{-1/2} \wt{E} \wt{D}^{1/2} - \wt D - \wt \Gamma\right) = \mathcal{O}(N)$, which is a consequence of the first inequality in the previous bound, the lower bound \eqref{eq:lwrbndeq2} and the fact that $E_N = \mathcal{O}(N)$. Thus, in summary, we obtain together with \eqref{eq:lwrbndeq2} that 
		\be \label{eq:uprbndEN2} E_N = \kappa_{\cJ_N} + \frac12  \tr_{\bot \pn} \left( \frac{1}{2}\left(\wt{D}^{1/2} \wt{E} \wt{D}^{-1/2} + \wt{D}^{-1/2} \wt{E} \wt{D}^{1/2}\right) - \wt D - \wt \Gamma\right) + \mathcal{O}(N^{-\rho}).	\ee
		
To continue with the higher min-max values of $\cJ_N$, we assume that $\lambda_n(\cJ_N)$ is such that $\lambda_n(\cJ_N)- E_N \leq \zeta$. By \eqref{eq:lwrbndeq2} and the identity \eqref{eq:uprbndEN2}, we also have 
		\be\label{eq:apbnddGE} \lambda_n\big( d\Gamma (\wt E)\big) \leq \zeta + C N^{-\rho}(1+\zeta^6). \ee
Now, let us recall that 
		\[ \lambda_n\big(d\Gamma(\wt E)\big) = \sum_{j=1}^\infty l_j^{(n)} \wt e_j,  \]
where the $(\wt e_j)_{j\in \mathbb{N}}$ denote the eigenvalues of $\wt E$ and the $l^{(n)}_j \in \mathbb{N}$ are non-zero for finitely many $j\in\mathbb{N}$. As we did in \eqref{eq:defWXU}, let us also denote by $(\varphi_j)_{j\in\mathbb{N}}$ an orthonormal eigenbasis of $\wt E$ such that $\wt E \varphi_j = \wt e_j \varphi_j$, for $j\in\mathbb{N}$. Then, we choose $n$ orthonormal  eigenvectors $ (\xi_{k})_{k=1}^n $ of the form
		\[ \xi_k = C_k  \prod_{j=1}^\infty \big( a^*(\varphi_j)\big)^{l^{(k)}_j } \Omega \in \cF_{\bot \pn}, \]
with $C_k $ ensuring that $\|\xi_k\|=1$ for each $k=1,\dots,n$ and such that 
		\[ d\Gamma ( \wt E) \xi_k = \lambda_k \big( d\Gamma(\wt E)\big) \xi_k \hspace{0.5cm} \text{for } k=1,\dots,n. \]
Given the vectors $ (\xi_{k})_{k=1}^n $, we define the linear space $Y^{(n)}\subset \cFpn$ by 
		\[Y^{(n)}  =  \chi_{\leq N}\, \cU\, \text{span} ( \xi_1,\dots, \xi_n), \] 
which we use below in the min-max principle to get an upper bound on $\lambda_n(\cJ_N)$. Before evaluating the upper bound, note that $ \chi_{\leq N}\cU\,:  \,\text{span} \big( \xi_1,\dots, \xi_n)\to Y^{(n)} $ is invertible, because we have for a normalized $\xi \in  \cU \, \text{span} ( \xi_1,\dots, \xi_n)$ and $m\in\mathbb{N}$ large enough (but fixed) that
		\be \label{eq:Ynproj} \| \chi_{\leq N}\cU\, \xi\| ^2 \geq \| \xi \|^2 -   \| \chi_{>N}\cU\,\xi\|^2 \geq 1-\big( \zeta + C N^{-\rho}(1+\zeta^6)\big)^m  N^{-m}> \frac12    \ee	
for $N$ large enough. Indeed, the second bound follows from Markov's inequality, the assumption $\zeta \leq CN^{-\rho/5-\eps}$ for some $\eps>0$, the bound \eqref{eq:U*NU}, the fact that $\cN$ commutes with $d\Gamma(\wt E)$, the operator bound $\cN \leq C  d\Gamma(\wt E) $ in $\cF_{\bot \pn}$ and because $\lambda_k \big(d\Gamma(\wt E)\big) \leq \zeta + C N^{-\rho}(1+\zeta^6)$ for each $k=1,\dots, n$ by \eqref{eq:apbnddGE}. Thus, by \eqref{eq:Ynproj}, $Y^{(n)}$ is in particular $n$-dimensional. 
		
Now, let us control $\lambda_n(\cJ_N)$ through the min-max principle. We obtain with \eqref{eq:JNdec}, \eqref{eq:errorrho}, \eqref{eq:QJNdec}, \eqref{eq:QJNdecerror}, \eqref{eq:chiQN-QN}, \eqref{eq:Ynproj} as well as Lemmas \ref{lm:U*KNU} and \ref{lm:U*VNU} 
that for $m\in \mathbb{N}$ large (but fixed) and $N$ large enough 
		\be \label{eq:uprbndln1} \begin{split}
		\lambda_n(\cJ_N) &\leq \sup_{  \xi\in Y^{(n)}, \|\xi\| =1 } \langle \xi, \cJ_N \xi \rangle = \sup_{  \xi\in \, \text{span}(\xi_1,\dots,\xi_n),\|\xi\|=1 }  \frac{ \langle \chi_{\leq N}\cU \xi, \cJ_N \chi_{\leq N}\cU \xi \rangle }{\langle \chi_{\leq N} \cU\xi,\chi_{\leq N} \cU\xi\rangle}  \\
		& \leq \big(1- ( N^{-1}\zeta  + C N^{-1-\rho}(1+\zeta^6))^m  \big)^{-1}  \sup_{ \xi\in \, \text{span}(\xi_1,\dots,\xi_n),\|\xi\|=1}\big( \kappa_{\cJ_N} + \langle \cU \xi, \wt \cQ_{\cJ_N} \cU \xi \rangle \big)\\
		& \hspace{0.5cm} + C N^{-\rho}  \sup_{ \xi\in \, \text{span}(\xi_1,\dots,\xi_n),\|\xi\|=1} \langle \xi, ( d\Gamma(\wt E) +\cN+ 1)(\cN+1)^4 \xi\rangle \\
		&\hspace{0.5cm} + C N^{-1/2}  \sup_{ \xi\in \, \text{span}(\xi_1,\dots,\xi_n),\|\xi\|=1} \langle \xi, ( d\Gamma(\wt E) + 1)^2 \xi\rangle.
		\end{split}\ee
Note that we used here the bounds $ d\Gamma(-\Delta + V_\text{ext}) \leq C d\Gamma(\wt E)$ and $\cV_N\leq \cK \cN$. To bound the right hand side in \eqref{eq:uprbndln1} further, we use that $\cN$ commutes with $d\Gamma(\wt E)$, $ \cN \leq C d\Gamma(\wt E)$ in $\cF_{\bot \pn}$
, the a priori estimate \eqref{eq:apbnddGE} as well as the identities \eqref{eq:diag} and \eqref{eq:uprbndEN2} so that
		\[ \begin{split}
		\lambda_n(\cJ_N) 
		& \leq   E_N + \lambda_n\big(d\Gamma(\wt E)\big)+ CN^{-\rho} \big( 1+ \zeta^5 + N^{-5\rho} \zeta^{30}\big).  \\
		\end{split}\]
Finally, applying Lemma \ref{lm:tildeEE} and once again \eqref{eq:apbnddGE}, we arrive at 
		\be \label{eq:uprbndln2} \begin{split}
		\lambda_n(\cJ_N) & \leq   E_N + \lambda_n\big(d\Gamma(E)\big)+ CN^{-\rho} \big( 1+ \zeta^5 + N^{-5\rho} \zeta^{30}\big).  
		\end{split}\ee

\subsection{Proof of Theorem \ref{thm:main}}  \label{sec:gsenergy}

In this section, we conclude the proof of Theorem \ref{thm:main}, based on the upper and lower bounds on the spectrum $\sigma(H_N)$ from Sections \ref{sec:lwrbnd} and \ref{sec:uprbnd}. 

\begin{proof}[Proof of Theorem \ref{thm:main}.]
Recalling \eqref{eq:uprbndEN2}, we have shown that 
		\be \label{eq:EN0}  E_N = \kappa_{\cJ_N} + \frac12 \tr_{\bot \pn} \left( \frac{1}{2}\left(\wt{D}^{1/2} \wt{E} \wt{D}^{-1/2} + \wt{D}^{-1/2} \wt{E} \wt{D}^{1/2}\right) - \wt D - \wt \Gamma\right) + \mathcal{O} (N^{-\rho})  \ee
with $\kappa_{\cJ_N}, \,\wt \Gamma$ from Proposition \ref{prop:JN} and with $\wt D, \wt E$ defined in \eqref{eq:deftildeDE}. Moreover, the bounds \eqref{eq:lwrbndsHN}, \eqref{eq:uprbndln2} show that if $\lambda_n(H_N) -E_N\leq \zeta $, then we have for sufficiently large $N$ that  
		\[ \lambda_n(H_N)- E_N =  \lambda_n\big(d\Gamma(E)\big)+ \mathcal{O}\big( N^{-\rho} ( 1+ \zeta^5 + N^{-5\rho} \zeta^{30} ) \big), \]
with $E$ defined in \eqref{eq:defHGPE}. The last identity proves the validity of \eqref{eq:excitHN-EN} in Theorem \ref{thm:main}. 

To finish the proof of the theorem, it now remains to verify the identity \eqref{eq:gsenergy}. This follows from evaluating the constant $\kappa_{\cJ_N} +  \tr_{\bot \pn} ( \frac{1}{2}\wt{D}^{1/2} \wt{E} \wt{D}^{-1/2} + \frac{1}{2}\wt{D}^{-1/2} \wt{E} \wt{D}^{1/2} - \wt D - \wt \Gamma)/2$ on the r.h.s. of \eqref{eq:EN0}, up to errors that vanish in the limit $N\to\infty$. To this end, it will be convenient to abbreviate by $\wt{K}_N:L^2 (\bR^3)\to L^2 (\bR^3)$ the operator with kernel
		\[ \wt{K}_N(x;y)= N^3(Vf_\ell)(N(x-y))\pn(x)\pn(y)\]
and to denote by $K_N:L^2_{\bot \pn}(\bR^3)\to L^2_{\bot \pn}(\bR^3)$ the operator
		\[ K_N=   Q \,\wt K_N\, Q, \]
where we recall that $Q = 1- |\pn\rangle\langle\pn|$. Now, inserting the definitions of $\kappa_{\cJ_N}$, $\wt D$ and $\wt \Gamma$ into \eqref{eq:EN0}, we use cyclicity of the trace and $1= f_\ell +w_\ell$, $\cosh^2 = 1 + \sinh^2$ to write 
		\[ 
		\begin{split}
		E_N &= N \int dx\, \Big( |\nabla\pn(x)|^2 + V_\text{ext}(x)\pn^2(x) + \frac12 \big( N^3(Vf_\ell)(N.)\ast \pn^2\big)(x)\pn^2(x)\Big)\\
		&\hspace{0.5cm} +\frac N2 \int dxdy\,   N^3(Vw_\ell)(N(x-y))  \pn^2(y)\pn^2(x)  - 4\pi \mathfrak{a}_0 \int dx\, \pn^4(x)\\
		&\hspace{0.5cm} + \frac12  \int dxdy\, N^2V(N(x-y))| \langle\sigma_x, \gamma_y\rangle|^2\\
		&\hspace{0.5cm} + \tr \big(\gamma \big[ N^3(Vw_\ell)(N(x-y))\pn(x)\pn(y) \big] \sigma \big)  \\
		&\hspace{0.5cm}+ \frac12 \tr_{\bot \pn} \left( \frac{1}{2}\left(\wt{D}^{1/2} \wt{E} \wt{D}^{-1/2} + \wt{D}^{-1/2} \wt{E} \wt{D}^{1/2}\right) - \wt D - \wt \Gamma\right)   + \mathcal{O} (N^{-\rho}).	  
		\end{split} \]
Applying Lemma \ref{lm:bndseta}, we also find that
		\[\begin{split}
		& \frac12  \int dxdy\, N^2V(N(x-y))| \langle\sigma_x, \gamma_y\rangle|^2 \\
		 &\hspace{1.5cm} = -\frac N2 \int dxdy\,   N^3(Vw^2_\ell)(N(x-y))  \pn^2(y)\pn^2(x)\\
		 &\hspace{2cm}  -\tr \big(\gamma \big[ N^3(Vw_\ell)(N(x-y))\pn(x)\pn(y) \big] \sigma \big) +\mathcal{O} (N^{-1}),
		\end{split}\]
which implies (recalling the definition (\ref{eq:defk}) of the kernel $k$) that 
		\be \label{eq:EN2} 
		\begin{split}
		E_N &= N \int dx\, \Big( |\nabla\pn(x)|^2 + V_\text{ext}(x)\pn^2(x) + \frac12 \big( N^3(Vf_\ell)(N.)\ast \pn^2\big)(x)\pn^2(x) \Big)\\ 
		&\hspace{0.5cm}- 4\pi \mathfrak{a}_0 \|\pn\|_4^4+ \frac12 \tr_{\bot \pn} \left(  \frac{1}{2}\left(\wt{D}^{1/2} \wt{E} \wt{D}^{-1/2} + \wt{D}^{-1/2} \wt{E} \wt{D}^{1/2}\right) - H_\text{GP}- K_N   \right) \\&\hspace{0.5cm}-\frac12 \tr \big(\wt  K_N k\big)+\mathcal{O} (N^{-\rho}).
		\end{split} \ee

To evaluate the right hand side in \eqref{eq:EN2} we use Lemma \ref{3.0.sceqlemma} \emph{ii)}, switch to Fourier space and apply Plancherel's theorem to find
		\[\begin{split}
		&\int dx\, \big( N^3(Vf_\ell)(N.)\ast \pn^2\big)(x)\pn^2(x)= \int dp\, \wh{(Vf_\ell)}(p/N) \big| \wh{(\pn^2)}(p)\big|^2 \\
		& = \wh{(Vf_\ell)}(0) \int dx\, \pn^4(x) -  N^{-1}  \int dx\,  V(x)f_\ell(x)\, x   \cdot  \int dy\, \pn^2(y) \nabla (\pn^2)(y)  + \mathcal{O}(N^{-2})\\
		& = \big(8\pi\mathfrak{a}_0  + 12\pi  \mathfrak{a}^2_0 \,\ell^{-1} N^{-1 }\big) \|\pn\|_4^4 + \mathcal{O}(N^{-2}).
		\end{split}\]
Note here that $\int dx\, V(x)f_\ell(x) x =0$ by radial symmetry of $V$ and $f_\ell$. 

Next, using the scattering equation \eqref{eq:scatl}, we write for $x\in \bR^3$
		\[ - Nw_\ell(Nx) = \int \frac{dy}{4\pi }\, \frac{1}{|x-y| } \Big( -\frac12 N^3 (Vf_\ell)(Ny) + N^3 \lambda_\ell f_\ell(Ny)\chi_\ell (y)\Big)\]
and obtain with similar arguments as above that 
\[ \begin{split}
-\frac12 \tr \big(\wt  K_N k\big) =&\, \frac1{16\pi^2} \int dpdq\, \big| \wh{(\pn^2)}(p)\big|^2   \wh{(V f_\ell)}((p-q)/N) \wh {(Vf_\ell)}(q/N) | q|^{-2}\\
& - \frac{ N^3\lambda_\ell}{8\pi^2} \int dpdq\, \big| \wh{(\pn^2)}(p)\big|^2   \wh{(V f_\ell)}((p-q)/N) \wh \chi_\ell(q ) | q|^{-2}\\
&+ \frac{ \lambda_\ell}{8\pi^2} \int dpdq\, \big| \wh{(\pn^2)}(p)\big|^2   \wh{(V f_\ell)}((p-q)/N) \wh{w_\ell}(q/N) | q|^{-2}\\
=&\, \frac1{16\pi^2} \int dpdq\, \big| \wh{(\pn^2)}(p)\big|^2   \wh{(V f_\ell)}((p-q)/N) \wh {(Vf_\ell)}(q/N) | q|^{-2}\\
& - \frac{ N^3\lambda_\ell}{8\pi^2} \int dpdq\, \big| \wh{(\pn^2)}(p)\big|^2   \wh{(V f_\ell)}((p-q)/N) \wh \chi_\ell(q ) | q|^{-2} + \mathcal{O}(N^{-1}),
\end{split} \]
where, in the last step, we used Lemma \ref{3.0.sceqlemma} \emph{i)}, the bound \eqref{eq:whwl} and the inequality $\Vert \wh{w_\ell}\Vert_\infty \leq \Vert w_\ell \Vert_1 \leq  CN^2 \ell^2$. Since $|\wh \chi_\ell(q)|\leq C|q|^{-2}$, Lemma \ref{3.0.sceqlemma} \emph{i)}, \emph{ii)} and a first order Taylor expansion also imply
 		\[\begin{split}
		&\frac{ N^3\lambda_\ell}{8\pi^2} \int dpdq\, \big| \wh{(\pn^2)}(p)\big|^2   \wh{(V f_\ell)}((p-q)/N) \wh \chi_\ell(q ) |q|^{-2}\\
		& = \frac{ N^3\lambda_\ell}{8\pi^2} \int_{|q|\leq N} dpdq\, \big| \wh{(\pn^2)}(p)\big|^2   \wh{(V f_\ell)}((p-q)/N) \wh \chi_\ell(q ) | q|^{-2}\\
		& \hspace{0.5cm} +\frac{ N^3\lambda_\ell}{8\pi^2} \int_{|q|>  N} dpdq\, \big| \wh{(\pn^2)}(p)\big|^2   \wh{(V f_\ell)}((p-q)/N) \wh \chi_\ell(q ) |  q|^{-2}\\
		& = 3\mathfrak{a}_0^2\,\pi^{-2} \ell^{-3}\!\!  \int dpdq\, \big| \wh{(\pn^2)}(p)\big|^2 \wh \chi_\ell(q ) | q|^{-2} + \mathcal{O}( \log N /N)= 6 \pi   \mathfrak{a}_0^2\, \ell^{-1} \|\pn\|_4^4+ \mathcal{O}( \log N /N )
		\end{split} \]
so that 
		\be \label{eq:EN3} 
		\begin{split}
		E_N &= N \cE_{GP}(\pn)- 4\pi \mathfrak{a}_0 \|\pn\|_4^4  \\ 
		&\hspace{0.5cm} +\frac14 \int dxdy\, N^3(Vf_\ell)(N(x-y)) \bigg[ \frac{1}{4\pi |.|}\ast  N^3(Vf_\ell)(N.)\bigg](x-y) \pn^2(x)\pn^2(y)\\
		&\hspace{0.5cm}+ \frac12 \tr_{\bot \pn} \left(  \frac{1}{2}\left(\wt{D}^{1/2} \wt{E} \wt{D}^{-1/2} + \wt{D}^{-1/2} \wt{E} \wt{D}^{1/2}\right) - H_\text{GP}- K_N   \right)+  \mathcal{O} (N^{-\rho}).
		\end{split} \ee
It remains to combine the last two lines in \eqref{eq:EN3}. 
We can rewrite
		\begin{align*}
		\begin{split}
		& \frac14 \int dxdy\, N^3(Vf_\ell)(N(x-y)) \bigg[ \frac{1}{4\pi |.|}\ast  N^3(Vf_\ell)(N.)\bigg](x-y) \pn^2(x)\pn^2(y) \\
		&= \frac14 \tr \left[ \wt K_N  \pn \big[ (-\Delta)^{-1}\wh {(Vf_\ell)}(./N)\big] \pn  \right],
		\end{split}
		\end{align*}
where $\pn$ acts as multiplication operator in position space and where $ \wh {(Vf_\ell)}(./N) $ acts as Fourier multiplier. Proceeding similarly as in \cite[Eq. (64)]{NNRT}, we first rewrite
\[\begin{split}
&\tr \left[ \wt K_N  \pn \big[ (-\Delta)^{-1}\wh {(Vf_\ell)}(./N)\big] \pn  \right]\\
&=\tr \left[ \frac12 \wt K_N  \pn \big[ (-\Delta)^{-1}\wh {(Vf_\ell)}(./N)\big] \pn  +  \frac12 \pn \big[\wh {(Vf_\ell)}(./N) (-\Delta)^{-1}\big] \pn  \wt K_N \right]\\
		& =  \tr \left[ \frac12 \wt K_N  \pn \big[ (-\Delta+\kappa^2)^{-1}\wh {(Vf_\ell)}(./N)\big] \pn  +  \frac12 \pn \big[\wh {(Vf_\ell)}(./N) (-\Delta+\kappa^2)^{-1}\big] \pn  \wt K_N \right]\\
		&\hspace{0.4cm} +\frac{\kappa^2}{2} \tr \left[ \wt K_N  \pn \big[ (-\Delta+\kappa^2)^{-1}(-\Delta)^{-1}\wh {(Vf_\ell)}(./N)\big] \pn +\text{h.c.}   \right]\\
		\end{split}\]
for a parameter $\kappa > 0$ that later will be chosen large enough. Plugging in $\wt K_N = \pn \wh {(Vf_\ell)}(./N)\pn $, we have
		\[\begin{split}
		&\frac12 \pn \wh {(Vf_\ell)}(./N)   \left( \pn^2 \frac1{-\Delta+\kappa^2}+ \frac1{-\Delta+\kappa^2} \pn^2 \right)  \wh {(Vf_\ell)}(./N) \pn  \\ 
		&= \pn \wh {(Vf_\ell)}(./N) \pn   \left( \frac1{-\Delta+\kappa^2} \right)\pn  \wh {(Vf_\ell)}(./N) \pn\\ 
		&\hspace{0.5cm}+ \pn \wh {(Vf_\ell)}(./N)   \left(  \frac1{-\Delta+\kappa^2}[\pn,-\Delta]\frac1{-\Delta+\kappa^2}[\pn,-\Delta]\frac1{-\Delta+\kappa^2} \right)  \wh {(Vf_\ell)}(./N) \pn\\ 
		&\hspace{0.5cm}+ \pn \wh {(Vf_\ell)}(./N)   \left(  \frac1{-\Delta+\kappa^2}|\nabla\pn|^2\frac1{-\Delta+\kappa^2}\right)  \wh {(Vf_\ell)}(./N) \pn .
		\end{split}\]
Inserting an additional projection $Q=1-|\pn\rangle\langle\pn|$, we arrive at
		\be \label{eq:maintr2}\begin{split}
		&\frac14 \tr \left[  \wt K_N  \pn \big[ (-\Delta)^{-1}\wh {(Vf_\ell)}(./N)\big] \pn   \right] = \frac14 \tr_{\bot \pn} \left[ K_N \frac1{H_\text{GP}} K_N\right] +\sum_{i=1}^6 X_{i}, 
		\end{split}  \ee
where we defined 
		\be \label{eq:defXj}\begin{split}
		\text{X}_1 &= \frac{\kappa^2}{8} \tr \left[  \pn \wh{(Vf_\ell)}(./N)  \pn^2  (-\Delta+\kappa^2)^{-1}(-\Delta)^{-1}\wh {(Vf_\ell)}(./N) \pn +   \text{h.c.}  \right] ,\\
		\text{X}_2 &= \frac14 \tr \left[\pn \wh {(Vf_\ell)}(./N)   \left(  \frac1{-\Delta+\kappa^2}[\pn,-\Delta]\right)^2\frac1{-\Delta+\kappa^2}   \wh {(Vf_\ell)}(./N) \pn\right]  ,\\
		\text{X}_3 &= \frac14 \tr \left[\pn \wh {(Vf_\ell)}(./N)   \left(  \frac1{-\Delta+\kappa^2}|\nabla\pn|^2\frac1{-\Delta+\kappa^2}\right)  \wh {(Vf_\ell)}(./N) \pn\right]  ,\\
		\text{X}_4 &= \frac14  \tr \left[ \pn \wh {(Vf_\ell)}(./N) \pn   \left( \frac1{-\Delta+\kappa^2} - \frac{1}{H_\text{GP} + \kappa^2} \right)\pn  \wh {(Vf_\ell)}(./N) \pn \right], \\
		\text{X}_5 &= \frac{1}{4\kappa^2}  \| \pn \wh {(Vf_\ell)}(./N) \pn^2 \|^2  + \frac14 \| (H_\text{GP}+\kappa^2)^{-1/2} Q \pn \wh{(Vf_\ell)}(./N) \pn^2\|^2,\\ 
\text{X}_6 &= -\frac{\kappa^2}{4} \tr \left[ K_N \frac{1}{ H_\text{GP} (H_\text{GP} + \kappa^2)} K_N \right].
		\end{split}\ee
Observe here, under the assumption $V \in L^2 (\bR^3)$, we have $0\leq  \tr_{\bot \pn}( K_N H_\text{GP}^{-1} K_N) =\| K_N H_\text{GP}^{-1/2}Q\|_\text{HS}^2 \leq CN<\infty$. Thus, the operator $K_N H_\text{GP}^{-1} K_N$ is of trace class and the first term on the r.h.s. of (\ref{eq:maintr2}) is well-defined.

Next, we show that the terms $\text{X}_1$ to $\text{X}_6$, defined in \eqref{eq:defXj}, are $\mathcal{O}(1)$ and that, in the limit $N \to \infty$, we can replace in each of them the Fourier multiplier $\wh{(Vf_\ell)}(./N)$ by $8\pi \mathfrak{a}_0$. We start with $\text{X}_1$. Since $\wh{(Vf_\ell)}(./N)$ and $\pn$ have bounded operator norm, we easily find (with Lemma \ref{thm:gpmin1})  
\[  |\text{X}_1| \leq C \big\| \pn (-\Delta + \kappa^2)^{-1/2} (-\Delta)^{-1/2} \big\|_\text{HS}^2 \leq C. \]
Similarly, we can bound, for any $\alpha < 1$, 
		\begin{equation} \label{eq:X1-com} \begin{split}
		\Big|\text{X}_1 - \frac{\kappa^2}{4}& (8\pi \mathfrak{a}_0)^2 \tr \left[    \pn^2  (-\Delta+\kappa^2)^{-1}(-\Delta)^{-1}  \pn^2    \right] \Big| \\ 
		\leq \; &C \| [ \wh{(Vf_\ell)}(./N) - 8\pi\mathfrak{a}_0]  \pn^2 (-\Delta)^{-1/2} (-\Delta + \kappa^2)^{-3/4+\eps} \|_\text{HS} \\ &+ C \| (-\Delta)^{-1/2} (-\Delta + \kappa^2)^{-3/4+\eps}  [ \wh{(Vf_\ell)}(./N) - 8\pi\mathfrak{a}_0]  \pn^2\|_\text{HS} \leq C N^{-\alpha} 
		\end{split} \end{equation} 
using that $| \wh{(Vf_\ell)}(p/N) - 8\pi\mathfrak{a}_0| \leq C (|p|/N)^\alpha$ (from the compact support of $V$), $\Vert \Delta \pn^2 \Vert \leq C$ (from Lemma \ref{thm:gpmin2}) and where we have chosen $\eps>0$ small enough depending on $\alpha$. 

For the contributions $\text{X}_2, \text{X}_3$, we use that $\wh{(Vf_\ell)}(./N)$ is real (by radial symmetry). We find  
		\[\begin{split} 
		\text{X}_2 & = \frac14  \left\| \pn \wh {(Vf_\ell)}(./N)   (-\Delta+\kappa^2)^{-1}[\pn,-\Delta] (-\Delta+\kappa^2)^{-1/2}  \right\|_\text{HS}^2, \\
		\text{X}_3 &= \frac14 \tr \left\|\pn \wh {(Vf_\ell)}(./N)  (-\Delta+\kappa^2)^{-1} \nabla\pn  \right\|_\text{HS}^2 . 
		 \end{split}\]
Proceeding as above, it is then straightforward to see that $ | \text{X}_2|, |\text{X}_3|\leq C$ and, using the bounds $\| \nabla\pn \|_{op}\leq C$, $\| [\pn,-\Delta] (-\Delta+\kappa^2)^{-1/2} \|_{op}\leq C $, moreover, that for any $\alpha<1/2$,
		\[\begin{split}
		&\Big| \text{X}_2- \frac14 (8\pi\mathfrak{a}_0)^2 \left\| \pn    (-\Delta+\kappa^2)^{-1}[\pn,-\Delta] (-\Delta+\kappa^2)^{-1/2}   \right\|_\text{HS}^2 \Big| \leq   CN^{-\alpha}, \\
		&\Big| \text{X}_3- \frac14 (8\pi\mathfrak{a}_0)^2 \left\| \pn   (-\Delta+\kappa^2)^{-1} \nabla\pn   \right\|_\text{HS}^2 \Big| \leq  CN^{-\alpha}. 
		\end{split}\]
To bound $\text{X}_4$, we write 
\[ \begin{split}
		\text{X}_4 
		& = \frac14 \tr \left[ \pn \wh {(Vf_\ell)}(./N) \pn  \frac1{-\Delta+\kappa^2} \big( V_\text{ext} +  8\pi \frak{a}_0 \pn^2 -\eps_\text{GP} \big) \frac1{H_\text{GP} + \kappa^2} \pn \wh {(Vf_\ell)}(./N) \pn \right].\\
		\end{split} \]
Thus 
\[ |\text{X}_4| \leq C \big\| \pn (-\Delta+\kappa^2)^{-1} (V_\text{ext} + 8\pi \frak{a} \pn^2 -\eps_\text{GP}) \big\|_\text{HS} \big\| (H_\text{GP} + \kappa^2)^{-1} \pn \big\|_\text{HS}.  \]
Since $V_\text{ext}$ grows at most exponentially by \eqref{eq:asmptsVVext} and, by (\ref{eq:expdecaypn}), $\pn$ decays faster than any exponential, choosing $\kappa > 0$ large enough we obtain that    
\[ \begin{split} \big\|  \pn (-\Delta+\kappa^2)^{-1} &(V_\text{ext} + 8\pi \frak{a} \pn^2 -\eps_\text{GP}) \big\|_\text{HS}^2 \\ &\leq C \int \pn^2 (x) \frac{e^{- 2 \kappa |x-y|}}{|x-y|^2} \big[ V_\text{ext} (y) + 8\pi \frak{a}_0 \pn^2 (y) - \eps_\text{GP} \big]^2  dx dy \leq C < \infty.  \end{split} \] 
From \eqref{eq:asmptsVVext}, we also have $\| (H_\text{GP} + \kappa^2)^{-1} \pn \|_\text{HS} \leq C$ and thus $|\text{X}_4| \leq C$. Proceeding similarly as in (\ref{eq:X1-com}) (and writing $\pn (-\Delta+\kappa^2)^{-1} = (-\Delta+\kappa^2)^{-1} \pn + (-\Delta + \kappa^2)^{-1} [-\Delta , \pn ] (-\Delta+\kappa^2)^{-1}$), we conclude that 
\[ \Big| \text{X}_4 -  \frac{(8\pi \frak{a}_0)^2}{4} \tr \left[ \pn^2  \Big( \frac1{-\Delta+\kappa^2} - \frac1{H_\text{GP} + \kappa^2} \Big) \pn^2 \right] \Big| \leq C N^{-\alpha} \]
for any $\alpha < 1$. With \eqref{eq:asmptsVVext}, we can also show that $\text{X}_5, \text{X}_6$ are bounded and that 
\[ \begin{split} 
\Big| \text{X}_5 - \frac{(8\pi\frak{a}_0)^2}{4\kappa^2} \| \pn^3  \| -  \frac{(8\pi\frak{a}_0)^2}{4} \| (H_\text{GP}+\kappa^2)^{-1/2} Q  \pn^3\|^2 \Big|  &\leq CN^{-1}.  \\
 \Big| \text{X}_6 + \frac{\kappa^2 (8\pi \frak{a}_0)^2}{4} \tr_{\perp \pn} \big[ \pn^2 Q \frac{1}{H_\text{GP} (H_\text{GP} + \kappa^2)} Q \pn^2 \Big] \Big| &\leq C N^{-\alpha}.
 \end{split} \] 
for every $0< \alpha < 1$. In fact, proceeding similarly as in (\ref{eq:X1-com}) and using that $\| \wh{V f_\ell} (p/N) - \wh{Vf_\ell} (0)| \leq C (|p|/N)^\alpha$ and that $\| |\nabla|^{\alpha} H^{-\alpha/2}_\text{GP} \|_\text{op} \leq C$, we can estimate
\[ \begin{split}  \Big| \text{X}_6 + &\frac{\kappa^2 (8\pi \frak{a}_0)^2}{4} \tr_{\perp \pn} \big[ \pn^2 Q \frac{1}{H_\text{GP} (H_\text{GP} + \kappa^2)} Q \pn^2 \Big] \Big| \\ &\leq  \frac{C}{N^\alpha} \big\| H_\text{GP}^{\alpha/2} \pn H_\text{GP}^{-1}(H_\text{GP}+\kappa^2)^{-3/4+\eps} \big\|_\text{HS} \big\|  H_\text{GP}^{-1/2} (H_\text{GP}+\kappa^2)^{-1/4-\eps} K_N \big\|_\text{HS}.
\end{split} \] 
The second norm on the r.h.s. is bounded, uniformly in $N$, by \eqref{eq:asmptsVVext} (and by the exponential decay of $\pn$, shown in  \eqref{eq:expdecaypn}).To bound the first Hilbert-Schmidt norm, we write 
\[ \begin{split} 
H_\text{GP}^{\alpha/2} \pn & \frac{1}{H_\text{GP}^{1/2} (H_\text{GP}+\kappa^2)^{3/4-\eps}} \\
&= \pn  \frac{H_\text{GP}^{\alpha/2}}{H_\text{GP}^{1/2} (H_\text{GP}+\kappa^2)^{3/4-\eps}} 
\\ &+ C \int_0^\infty ds \ s^{\alpha/2} \frac{1}{s+H_\text{GP}} [H_\text{GP}, \pn] \frac{1}{s+ H_\text{GP}} \frac{1}{H_\text{GP}^{1/2} (H_\text{GP}+\kappa^2)^{3/4-\eps}} \end{split} \]
for an appropriate constant $C > 0$ and with $[H_\text{GP}, \pn]=- 2 \nabla \cdot \nabla \pn - \Delta \pn$. 
Estimating $\Vert (s+H_\text{GP})^{-1} \nabla \Vert_\text{op} \leq C (1+s)^{1/2}$, $\Vert (s+H_\text{GP})^{-1} \Vert_\text{op} \leq C (1+s)^{-1}$ and again $\| e^{-\alpha |x|} H_\text{GP}^{-1/2} (H_\text{GP}+\kappa^2)^{-1/4-\eps} \|_\text{HS} \leq C$, we obtain the desired bound. 

Let us now consider the trace on the last line of (\ref{eq:EN3}). From the proof of Lemma~\ref{lm:tildeEE}, recall that $\wt D^{1/2} = U_0^*H_\text{GP}^{1/2}e^{-\eta}$ for some unitary $U_0:L^2_{\bot \pn}(\bR^3) \rightarrow L^2_{\bot \pn}(\bR^3)$. Thus (recall that by the proof of Lemma \ref{lm:alpha} \emph{a)} we have $H_\text{GP}^{1/2} (H_\text{GP}+2K_N)H_\text{GP}^{1/2} \geq 0$)  
		\begin{align*}
		&\frac12\tr_{\bot \pn} \left(  \frac{1}{2}\left(\wt{D}^{1/2} \wt{E} \wt{D}^{-1/2} + \wt{D}^{-1/2} \wt{E} \wt{D}^{1/2}\right) - H_\text{GP}- K_N   \right) \\&=  \frac12 \tr_{\bot \pn} \Big[ \frac{1}{2}\left( e^{-\eta} H_\text{GP}^{1/2}\big(  H_\text{GP}^{1/2}( H_\text{GP} + 2K_N) H_\text{GP}^{1/2}\big)^{1/2} H_\text{GP}^{-1/2} e^\eta + \text{h.c.} \right)  - H_\text{GP}- K_N    \Big].
		\end{align*}

Next we want to remove the factors of $e^{\pm \eta}$. If all operators were trace-class, we could use cyclicity of the trace. To overcome this issue, we rewrite it as
\begin{equation} \label{eq:cyclicexp}
	\begin{split}
	&\tr_{\bot \pn} \left( \left[ e^{-\eta} H_\text{GP}^{1/2}  \big(H_\text{GP}^{1/2}( H_\text{GP} + 2K_N) H_\text{GP}^{1/2}\big)^{1/2}   H_\text{GP}^{-1/2} e^\eta - H_\text{GP} - K_N \right] + \text{h.c.} \right) \\
	&= \tr_{\bot \pn} \Big( \Big[  H_\text{GP}^{1/2}\left(  \big(H_\text{GP}^{1/2}( H_\text{GP} + 2K_N) H_\text{GP}^{1/2}\big)^{1/2} - H_\text{GP} \right) H_\text{GP}^{-1/2} - K_N  \\
	&\hspace{1.7cm}+ \left(e^{-\eta}-1\right)H_\text{GP}^{1/2}\left(  \big(H_\text{GP}^{1/2}( H_\text{GP} + 2K_N) H_\text{GP}^{1/2}\big)^{1/2} - H_\text{GP} \right) H_\text{GP}^{-1/2} \left(e^\eta -1\right)  \\
	&\hspace{1.7cm}+ \left(e^{-\eta}-1\right)H_\text{GP}^{1/2}\left[  \big(H_\text{GP}^{1/2}( H_\text{GP} + 2K_N) H_\text{GP}^{1/2}\big)^{1/2} - H_\text{GP} \right] H_\text{GP}^{-1/2}  \\
	&\hspace{1.7cm}+ H_\text{GP}^{1/2}\left[  \big(H_\text{GP}^{1/2}( H_\text{GP} + 2K_N) H_\text{GP}^{1/2}\big)^{1/2} - H_\text{GP} \right] H_\text{GP}^{-1/2} \left(e^{-\eta}-1\right)  \\
	&\hspace{1.7cm}+ \left(e^{-\eta}-1+\eta\right)H_\text{GP} \left(e^\eta -1-\eta\right) + \left(1-\eta\right)H_\text{GP} \left(e^{-\eta}-1-\eta\right)\\
	&\hspace{1.7cm} + \left( e^{-\eta}-1+\eta \right) H_\text{GP} (1+\eta) - \eta H_\text{GP} \eta \Big] + \text{h.c.} \Big). 
	\end{split}
\end{equation}	
Now we want to move all factors involving $\eta$ to the right. For this we will use that for all $\varepsilon>0$ we have 
\begin{equation} \label{eq:HEHHS}
\begin{split}
&\Vert Q H_\text{GP}^{1/2-\varepsilon} \left[ \big(  H_\text{GP}^{1/2}( H_\text{GP} + 2K_N) H_\text{GP}^{1/2}\big)^{1/2} - H_\text{GP} \right] H_\text{GP}^{-1/2} Q \Vert_\text{HS} <\infty, \\ &\Vert  Q H_\text{GP}^{1/2} \left[ \big(  H_\text{GP}^{1/2}( H_\text{GP} + 2K_N) H_\text{GP}^{1/2}\big)^{1/2} - H_\text{GP} \right] H_\text{GP}^{-1/2-\varepsilon} Q \Vert_\text{HS} < \infty.
\end{split}
\end{equation}
Before proving \eqref{eq:HEHHS}, we show how we apply it to \eqref{eq:cyclicexp}. For the term on the second line of the r.h.s of \eqref{eq:cyclicexp} we can use \eqref{eq:HEHHS}, the boundedness of $e^{-\eta}-1$ and $\Vert Q H_\text{GP}^{\varepsilon} (e^{\eta}-1) \Vert_\text{HS}\leq e^{\Vert \eta \Vert_\text{op}} \Vert H_\text{GP}^{\varepsilon-1} Q \Vert_\text{op} \Vert H_\text{GP} \eta \Vert_\text{HS} < \infty$ for $0\leq \varepsilon \leq 1$ to get
\begin{align*}
&\tr_{\bot \pn}\left( \left(e^{-\eta}-1\right)H_\text{GP}^{1/2}\left(  \big(H_\text{GP}^{1/2}( H_\text{GP} + 2K_N) H_\text{GP}^{1/2}\big)^{1/2} - H_\text{GP} \right) H_\text{GP}^{-1/2-\varepsilon} H_\text{GP}^\varepsilon\left(e^\eta -1\right) \right) \\
&= \tr_{\bot \pn}\left( H_\text{GP}^{1/2}\left(  \big(H_\text{GP}^{1/2}( H_\text{GP} + 2K_N) H_\text{GP}^{1/2}\big)^{1/2} - H_\text{GP} \right) H_\text{GP}^{-1/2} \left(e^\eta -1\right) \left(e^{-\eta}-1\right) \right).
\end{align*} 
Using additionally $\Vert H_\text{GP}^\varepsilon (e^{-\eta}-1) H_\text{GP}^\varepsilon \Vert_\text{HS}\leq \Vert H_\text{GP}^\varepsilon \eta H_\text{GP}^\varepsilon \Vert_\text{HS} + e^{\Vert \eta \Vert_\text{op}} \Vert H_\text{GP} \eta \Vert_\text{HS}^2 <\infty$ for $0\leq \varepsilon <\frac{1}{2}$ we obtain
\begin{align*}
&\tr_{\bot \pn} \left( \left(e^{-\eta}-1\right)H_\text{GP}^{1/2}\left[  \big(H_\text{GP}^{1/2}( H_\text{GP} + 2K_N) H_\text{GP}^{1/2}\big)^{1/2} - H_\text{GP} \right] H_\text{GP}^{-1/2} \right) \\
&= \tr_{\bot \pn} \left( H_\text{GP}^{-\varepsilon} \left( H_\text{GP}^{1/2}\left[  \big(H_\text{GP}^{1/2}( H_\text{GP} + 2K_N) H_\text{GP}^{1/2}\big)^{1/2} - H_\text{GP} \right] H_\text{GP}^{-1/2-\varepsilon} \right) H_\text{GP}^\varepsilon (e^{-\eta} -1) H_\text{GP}^\varepsilon \right) \\
&= \tr_{\bot \pn} \left( H_\text{GP}^{1/2}\left[  \big(H_\text{GP}^{1/2}( H_\text{GP} + 2K_N) H_\text{GP}^{1/2}\big)^{1/2} - H_\text{GP} \right] H_\text{GP}^{-1/2} \left(e^{-\eta}-1\right) \right).
\end{align*}
Now we write $e^{-\eta}-1+\eta=g_\eta \eta^2$, where $g_\eta$ is a bounded operator that commutes with $\eta$. Using that $\eta$ and $\eta H_\text{GP}$ are both Hilbert-Schmidt, we obtain
\begin{align*}
	\tr_{\bot \pn}\left( \left( e^{-\eta}-1+\eta\right) H_\text{GP} (1+\eta)  \right)
	&=\tr_{\bot \pn}\left( \eta H_\text{GP} \eta g_\eta (1+\eta)  \right) \\
	&= \tr_{\bot \pn}\left(  H_\text{GP} (1+\eta) \left( e^{-\eta}-1+\eta\right) \right).
\end{align*}
The remaining terms can be dealt with in a similar fashion. Thus, once we have shown \eqref{eq:HEHHS} we obtain
\begin{align*}
\tr_{\bot \pn}& \left( \left[ e^{-\eta} H_\text{GP}^{1/2}\left(  \big(H_\text{GP}^{1/2}( H_\text{GP} + 2K_N) H_\text{GP}^{1/2}\big)^{1/2} - H_\text{GP} \right) H_\text{GP}^{-1/2} e^\eta - K_N \right] + \text{h.c.} \right) \\
&= \tr_{\bot \pn} \left(  \left[ H_\text{GP}^{1/2}\left(  \big(H_\text{GP}^{1/2}( H_\text{GP} + 2K_N) H_\text{GP}^{1/2}\big)^{1/2} - H_\text{GP} \right) H_\text{GP}^{-1/2} - K_N \right] + \text{h.c.} \right).
\end{align*}
Let us now prove \eqref{eq:HEHHS}.
With the integral representation 
\[ x^{1/2} = x \cdot x^{-1/2} = \frac{1}{\pi} \int_0^\infty \frac{ds}{\sqrt{s}} \frac{x}{s+x} =\lim_{\Lambda \to \infty} \frac{1}{\pi} \int_0^\Lambda \frac{ds}{\sqrt{s}} - \frac{1}{\pi} \int_0^\Lambda ds \, \sqrt{s} \frac{1}{s+x} \]
we can write 
\begin{equation} \label{eq:intE-H}
	\begin{split} \big( H_\text{GP}^{1/2}( H_\text{GP} &+ 2K_N) H_\text{GP}^{1/2}\big)^{1/2} - H_\text{GP} \\ &= -\frac{1}{\pi} \int_0^\infty ds \, \sqrt{s} \, \left[ \frac{1}{s+H^2_\text{GP} + 2 H^{1/2}_\text{GP} K_N H^{1/2}_\text{GP}} - \frac{1}{s+H^2_\text{GP}} \right] \\ &=  \frac{2}{\pi} \int_0^\infty ds \, \sqrt{s} \, \frac{1}{s+H^2_\text{GP}} H_\text{GP}^{1/2} K_N H_\text{GP}^{1/2} \frac{1}{s+H^2_\text{GP} + 2 H^{1/2}_\text{GP} K_N H^{1/2}_\text{GP}}. \end{split}
\end{equation}
Using the resolvent identiy, we obtain
\begin{align*}
	&QH_\text{GP}^{1/2} \left[ \big(  H_\text{GP}^{1/2}( H_\text{GP} + 2K_N) H_\text{GP}^{1/2}\big)^{1/2} - H_\text{GP} \right] H_\text{GP}^{-1/2-\varepsilon}Q \\
	&= \frac{2}{\pi} \int_0^\infty ds \, \sqrt{s} \, Q\frac{1}{s+H^2_\text{GP}} H_\text{GP} K_N  \frac{1}{s+H^2_\text{GP}} H_\text{GP}^{-\varepsilon}Q \\
	&\hspace{0.5cm} + \frac{4}{\pi} \int_0^\infty ds \, \sqrt{s} \, Q\frac{1}{s+H^2_\text{GP}} H_\text{GP} K_N  H_\text{GP}^{1/2} \frac{1}{s+H_\text{GP}^2+ 2H_\text{GP}^{1/2} K_N H_\text{GP}^{1/2}} H_\text{GP}^{1/2} K_N \frac{H_\text{GP}^{-\varepsilon}}{s+H_\text{GP}^2}Q.
\end{align*}
Using that $\Vert Q[K_N, H_N] H_\text{GP}^{-1/2} Q\Vert_\text{HS}<\infty$ we obtain the second inequality in \eqref{eq:HEHHS}. Similarly one obtains the other inequality in \eqref{eq:HEHHS}.
Therefore, we can use cyclicity and get
\begin{align*}
&\frac12\tr_{\bot \pn} \left(  \frac{1}{2}\left(\wt{D}^{1/2} \wt{E} \wt{D}^{-1/2} + \wt{D}^{-1/2} \wt{E} \wt{D}^{1/2}\right) - H_\text{GP}- K_N   \right) \\&=  \frac12 \tr_{\bot \pn} \Big[ \frac{1}{2}\left( H_\text{GP}^{1/2}\big(  H_\text{GP}^{1/2}( H_\text{GP} + 2K_N) H_\text{GP}^{1/2}\big)^{1/2} H_\text{GP}^{-1/2} + \text{h.c.} \right)  - H_\text{GP}- K_N    \Big]\\
&=  \frac14 \tr_{\bot \pn} \Big[ \left( H_\text{GP}^{1/2}\left[\big(  H_\text{GP}^{1/2}( H_\text{GP} + 2K_N) H_\text{GP}^{1/2}\big)^{1/2} - H_\text{GP} \right] H_\text{GP}^{-1/2} -K_N \right) + \text{h.c.}  \Big].
\end{align*}

Applying the resolvent identity twice to \eqref{eq:HEHHS}, we obtain
\begin{equation} \label{eq:resol} \begin{split} H_\text{GP}^{1/2} &[\big( H_\text{GP}^{1/2}( H_\text{GP} + 2K_N) H_\text{GP}^{1/2}\big)^{1/2} - H_\text{GP}] H_\text{GP}^{-1/2} - K_N \\  = \; &\frac{2}{\pi} \int_0^\infty ds \, \sqrt{s} \, H_\text{GP}\frac{1}{s+H^2_\text{GP}}  K_N  \frac{1}{s+H^2_\text{GP}}   - K_N \\ &-\frac{4}{\pi} \int_0^\infty ds \, \sqrt{s} \, \frac{1}{s+H^2_\text{GP}} H_\text{GP} K_N H_\text{GP}^{1/2} \frac{1}{s+H^2_\text{GP}} H_\text{GP}^{1/2} K_N  \frac{1}{s+H^2_\text{GP}}  \\ &+\frac{8}{\pi} \int_0^\infty ds \, \sqrt{s} \, H_\text{GP}^{1/2}\Big[ \frac{1}{s+H^2_\text{GP}} H_\text{GP}^{1/2} K_N H_\text{GP}^{1/2} \Big]^3 \frac{1}{s+H^2_\text{GP}+ 2 H^{1/2}_\text{GP} K_N H^{1/2}_\text{GP} } H_\text{GP}^{-1/2}.   \end{split}  \end{equation} 
Since 
\[ \frac{2}{\pi} \int_0^\infty ds \, \sqrt{s} \, \frac{H_\text{GP}}{(s + H_\text{GP}^2)^2} = 1 \]
the trace of the first line on the r.h.s. of (\ref{eq:resol}) vanishes, by cyclicity. To apply cyclicity of the trace, we need here to be a bit careful, since $K_N$ is only guaranteed to be trace class, if $\wh {(Vf_\ell)} (./N)$ is integrable, a property which does not follow from our assumptions \eqref{eq:asmptsVVext} on the interaction $V$. 
In general, we can justify cyclicity remarking that, for self-adjoint operators $A,B$, we have the identity
\[ 2ABA - A^2 B - B A^2 = - \big[ A, \big[ A, B \big] \big]. \] 
Thus 
\[ \begin{split} &\left( \frac{2}{\pi} \int_0^\infty ds \, \sqrt{s}  \, H_\text{GP}\frac{1}{s+H_\text{GP}^2}  K_N  \frac{1}{s + H_\text{GP}^2} - K_N \right) + \text{h.c.} \\ &\hspace{0.5cm}= \left( \frac{1}{2} \left(H_\text{GP} K_N H_\text{GP}^{-1} - K_N \right) - \frac{1}{\pi} \int_0^\infty ds \, \sqrt{s} \, H_\text{GP}\Big[ \frac{1}{s+ H_\text{GP}^2} , \Big[   \frac{1}{s+ H_\text{GP}^2} , K_N \Big] \Big]\right) + \text{h.c.} \end{split} \]
A tedious but straightforward analysis shows that the r.h.s. of the last equation is trace class and that its trace norm depends continuously on $Vf_\ell$ in the $L^2(\bR^3)$ topology. For this argument we use that $(\partial_i \partial_j \pn)$ has exponential decay which can be proved similarly like \eqref{eq:expdecaypn} and \eqref{exponential decay}. With a simple approximation argument, and cyclicity of the trace for the second and third term in \eqref{eq:resol}, we can therefore conclude that 
\begin{equation}\label{eq:resol2} \begin{split} 
\frac14 \tr_{\bot \pn} &\Big[ \left( H_\text{GP}^{1/2}\left[\big(  H_\text{GP}^{1/2}( H_\text{GP} + 2K_N) H_\text{GP}^{1/2}\big)^{1/2} - H_\text{GP} \right] H_\text{GP}^{-1/2} -K_N \right) + \text{h.c.}  \Big] \\  =  &-\frac{2}{\pi} \int_0^\infty ds \, \sqrt{s} \, \tr_{\perp \pn} \frac{1}{s+H^2_\text{GP}} H_\text{GP}^{1/2} K_N H_\text{GP}^{1/2} \frac{1}{s+H^2_\text{GP}} H_\text{GP}^{1/2} K_N H_\text{GP}^{1/2} \frac{1}{s+H^2_\text{GP}}  \\ &+\frac{4}{\pi} \int_0^\infty ds \, \sqrt{s} \,  \tr_{\perp \pn}\Big[ \frac{1}{s+H^2_\text{GP}} H_\text{GP}^{1/2} K_N H_\text{GP}^{1/2} \Big]^3 \frac{1}{s+H^2_\text{GP}+ 2 H^{1/2}_\text{GP} K_N H^{1/2}_\text{GP} }.   \end{split}  \end{equation} 
To handle the first term on the r.h.s. we use cyclicity of the trace (from \eqref{eq:asmptsVVext}, it follows that $\| K_N H_\text{GP}^{-3/4-\eps} Q \|_\text{HS} \leq C$ for any $\eps > 0$, and therefore the operator under consideration is of trace class, justifying cyclicity) and the fact that 
\[ \frac{2}{\pi} \int_0^\infty ds \sqrt{s} \, \frac{1}{(s+H_\text{GP}^2)^3}  = \frac{1}{4 H_\text{GP}^3} \]
to write
\begin{equation}\label{eq:resol3} \begin{split} -\frac{2}{\pi} &\int_0^\infty ds \sqrt{s}  \tr_{\perp \pn} \frac{1}{s+H^2_\text{GP}} H_\text{GP}^{1/2} K_N H_\text{GP}^{1/2} \frac{1}{s+H^2_\text{GP}} H_\text{GP}^{1/2} K_N H_\text{GP}^{1/2} \frac{1}{s+H^2_\text{GP}} \\ = \; &-\frac{1}{4} \tr_{\perp \pn} K_N \frac{1}{H_\text{GP}} K_N - \frac{1}{\pi} \int_0^\infty ds \sqrt{s} \tr_{\perp \pn} \frac{H_\text{GP}^{1/2}}{s+H_\text{GP}^2} \Big[ \Big[ K_N , \frac{H_\text{GP}}{s+H_\text{GP}^2} \Big] , K_N \Big] \frac{H_\text{GP}^{1/2}}{s+H_\text{GP}^2}. \end{split} \end{equation} 
The first term on the r.h.s. cancels precisely with the main contribution on the r.h.s. of (\ref{eq:maintr2}). Expanding the commutators in the second term, we generate several contributions. All these contributions are bounded, uniformly in $N$. As an example, we estimate the term
\begin{equation} \label{eq:example} \begin{split} \Big|  \frac{1}{\pi} \int_0^\infty &ds \sqrt{s} \tr_{\perp \pn}  \frac{H_\text{GP}^{1/2}}{s+H_\text{GP}^2} \Big[ K_N, H_\text{GP} \Big] \frac{1}{s+H_\text{GP}^2} \Big[ K_N , H_\text{GP} \Big] \frac{H_\text{GP}^{3/2}}{(s+H_\text{GP}^2)^2}  \Big| \\ & \leq C \int_0^\infty ds \, \sqrt{s} \, \Big\| \frac{H_\text{GP}^{1/2}}{s + H_\text{GP}^2} e^{-\alpha |x|} \Big\|_\text{HS} \Big\|  e^{\alpha |x|} [ K_N , H_\text{GP} ] \frac{1}{(H_\text{GP}+1)^{1/2}}  \Big\|_\text{op} \\ & \hspace{1cm} \times  \Big\| \frac{H_\text{GP}+1}{s+ H_\text{GP}^2} \Big\|_\text{op} \Big\| \frac{1}{(H_\text{GP} + 1)^{1/2}} [ K_N , H_\text{GP} ] e^{\alpha |x|}   \Big\|_\text{op}  \Big\| e^{-\alpha |x|} \frac{H^{3/2}_\text{GP}}{(s+ H_\text{GP}^2)^2} \Big\|_\text{HS}  \\ &\leq C \int_0^{\infty} ds \, \sqrt{s} \, \frac{1}{(1+s)^{7/4-2\eps}} 
\leq C 
\end{split} \end{equation} 
choosing $\eps > 0$ small enough. Here we used that, for any $\alpha > 0$, $\| e^{\alpha |x|} [K_N , H_\text{GP}] (H_\text{GP}+1)^{-1/2} \|_\text{op} \leq C$ (notice that the Laplacian commutes with the multiplier $\wh{Vf_\ell} (./N)$ in $K_N$; when commuted through the multiplication operator $\pn$, it produces terms that can be controlled by $(H_\text{GP} + 1)^{-1/2}$ and decay faster than $e^{\alpha |x|}$ in space, by Lemma \ref{thm:gpmin2}) and that, from \eqref{eq:asmptsVVext}, $\| H_\text{GP}^{1/2} (s + H_\text{GP}^2)^{-1} e^{-\alpha |x|} \|_\text{HS} \leq C (1+s)^{-3/8+\eps}$, $\| (1+H_\text{GP}) / (s + H_\text{GP}^2) \|_\text{op} \leq C (1+s)^{-1/2}$, $\| e^{-\alpha |x|} H_\text{GP}^{3/2} / (s+H_\text{GP}^2)^2 \|_\text{HS} \leq C (1+s)^{-7/8+\eps}$. Proceeding similarly, we can control all other contributions arising from the second term on the r.h.s. of (\ref{eq:resol3}). Arguing analogously to (\ref{eq:X1-com}), we can also show that, as $N \to \infty$, these contributions approach the expressions obtained replacing $\wh{Vf_\ell}(./N)$ by $8\pi \frak{a}_0$. Here the convergence has a rate proportional to $N^{-\alpha}$, for any $\alpha < 1$. This is related to the fact that the integral on the last line of (\ref{eq:example}) remains finite if we insert a factor $s^{\alpha/4}$, for any $\alpha < 1$, and on the observation that, since $| \wh{(Vf_\ell)}(p/N) - 8\pi\mathfrak{a}_0| \leq C (|p|/N)^\alpha$, a rate $N^{-\alpha}$ corresponds to $|p|^\alpha \simeq H_\text{GP}^{\alpha/2}$ and thus exactly to a factor $s^{\alpha/4}$. 

Finally, let us consider the second term on the r.h.s. of (\ref{eq:resol2}). We can bound it by 
\[ \begin{split} 
\Big| \int_0^\infty ds \, \sqrt{s} \, \tr_{\perp \pn} &\Big[  \frac{1}{s+H^2_\text{GP}} H_\text{GP}^{1/2} K_N H_\text{GP}^{1/2} \Big]^3 \frac{1}{s+H^2_\text{GP}+ 2 H^{1/2}_\text{GP} K_N H^{1/2}_\text{GP}} \Big| \\ &\hspace{4cm} \leq C \int_0^\infty ds \, \sqrt{s} \, \frac{1}{(1+s)^{7/4-2\eps}} \leq C \end{split}  \]
because $\| e^{\alpha |x|} K_N  \|_\text{op} \leq C$, $\| H_\text{GP}^{1/2} K_N H_\text{GP}^{-1/2} \|_\text{op} \leq C$ (controlling the commutator between $H_\text{GP}$ and $K_N$ as explained in the previous paragraph) and $\| H^{1/2}_\text{GP} (s+H_\text{GP}^2)^{-1} e^{-\alpha |x|} \|_\text{HS} \leq C (1+s)^{-3/8+\eps}$, $\| H_\text{GP} (s+H_\text{GP}^2)^{-1} e^{-\alpha |x|} \|_\text{HS} \leq C (1+s)^{-1/8+\eps}$, $\| H_\text{GP}^{3/2} (s+H_\text{GP}^2)^{-1} \|_\text{op} \leq C (1+s)^{-1/4}$. Also in this case, we can show that as $N \to \infty$, this term approaches 
the expression obtained replacing $\wh{(Vf_\ell)}(./N)$ by $8\pi \frak{a}_0$, with a rate of order $N^{-\alpha}$, for $\alpha < 1$. 

Collecting all bounds for the terms $\text{X}_1, \dots , \text{X}_6$ and for the contributions on the r.h.s. of (\ref{eq:resol2}) (after removing the first term on the r.h.s. of (\ref{eq:resol3})), we obtain the formula (\ref{eq:gsenergy}) for the ground state energy, with $E_\text{Bog}$ as on the r.h.s. of (\ref{eq:Ebog}). It is easy to check that the arguments of this section remain true if we replace the potential $\wh{Vf_\ell} (./N)$ by the approximate identity $\unit_\delta$ (corresponding, in Fourier space, to the multiplier with the Gaussian $e^{-\delta p^2 /2}$) and we let $\delta \to 0$. This shows the validity of (\ref{eq:bog1}) and concludes the proof of Theorem \ref{thm:main}.   
\end{proof}

\section{Analysis of $\cG_{N}$}  \label{sec:GN}
In this section we analyze $\cG_{N}$, defined in \eqref{eq:defGN}. Motivated by \eqref{eq:cLNj}, we write
		\[\cG_{N} =  \cG_{N}^{(0)} +\cG_{N}^{(1)}+\cG_{N}^{(2)}+\cG_{N}^{(3)}+\cG_{N}^{(4)}, \]
where, for $j\in \{0,1,2,3,4\}$, we set
		\begin{align*}
		\cG_{N}^{(j)} = e^{-B}\cL_{N}^{(j)}e^{B}.
		\end{align*}
In the following, we will analyze the operators $\cG_{N}^{(j)} $ separately; the results will be combined in Sec. \ref{sec:proofpropGN} to conclude the proof of Proposition \ref{prop:GN}. The strategy used to show Prop.  \ref{prop:GN} is very similar to the one used in \cite[Sec. 6]{BS}, \cite[Sec. 7]{BBCS4} and \cite[Sec. 6]{BSS}, so we outline the key steps only. We focus in particular on the operator bounds that lead to \eqref{eq:Delta-bd} and \eqref{eq:cEcGNbnd}. The commutator bounds \eqref{eq:adkN} can be proved in the same way, following the remarks at the beginning of \cite[Sec. 7]{BBCS4}. We omit the details, because this would only lead to additional notation.

Before we start with the analysis of the operators $\cG_{N}^{(j)} $, let us record the following lemma that gives a rough bound on the conjugation of the kinetic and potential energies.
\begin{lemma}\label{lm:roughVNK}
	Assume \eqref{eq:asmptsVVext}. For every $j\in \NN$, there exists a constant $C>0$ such that 
	\begin{equation} \label{bound conjugation K N}
	e^{-B} \mathcal{K}(\mathcal{N}+1)^j e^{B} \leq C \mathcal{K}(\mathcal{N}+1)^j + CN (\mathcal{N}+1)^{j+1}
	\end{equation}
and 
	\begin{equation} \label{bound conjugation V_N}
		e^{-B} \mathcal{V}_N(\mathcal{N}+1)^j e^{B} \leq C \mathcal{V}_N(\mathcal{N}+1)^j + CN ( \cN+1)^{j}.
	\end{equation}
Similarly, for every $j\in \NN$, there exists $C>0$ such that
	\begin{equation}\label{eq:Vextgrow}
	e^{-B} \cV_\text{ext}(\cN+1)^j e^{B} \leq  C (\cV_\text{ext} + \cN+1)(\cN+1)^j.
	\end{equation}
\end{lemma}
\begin{proof}
For simplicity, let us focus on the case $j=0$. The general case is then a straightforward generalization of the argument (see also \cite[Lemma 7.1]{BBCS4} for the analogous statements in the translation invariant setting). 

Let $\xi\in \cFpn $ be normalized and define $s\mapsto\xi_s$ by $\xi_s =\langle\xi, e^{-sB} \cK e^{sB} \xi\rangle$. The key estimate to obtain the bound \eqref{bound conjugation K N} is the operator bound 
		\[\begin{split}
		\pm [\cK, B] = \pm \int dx\, \big( \nabla_x b^*_x b^*(\nabla_x\eta_x) + \text{h.c.}\big) \leq \cK + CN(\cN+1),
		\end{split}\]
where we used that $\|\nabla_1\eta\|^2 \leq CN$, by \eqref{eq:bndseta}. This implies with Lemma \ref{lm:Npow} that
		\[ |\partial_s \xi_s | \leq \xi_s + CN \langle\xi, (\cN+1)\xi\rangle  \]
and \eqref{bound conjugation K N} follows from Gronwall's lemma. Similarly, \eqref{bound conjugation V_N} is based on $\|\eta\|_\infty\leq CN$ and 
		\[\begin{split}
		\pm [\mathcal{V}_N, B] 
	=&\, \pm \frac12 \int dx dy  N^2 V(N(x-y))  \big( \eta(x;y)b_x^* b_y^* + 2 b_x^* b_y^*   a^*(\eta_x) a_y     +\text{h.c.}\big)   \\
	\leq&\,   C ( \cV_N  +  N)
		\end{split}\] 
		and the last bound \eqref{eq:Vextgrow} follows from
		\[\pm [\cV_\text{ext},B] =\pm \bigg( \int dx \ V_\text{ext}(x) b^*_x b^*(\eta_x)+\text{h.c.}\bigg)\leq C (\cN+1). \]
Here, we used that $\Vert V_\text{ext} \pn \Vert \leq C$, which follows form the assumption in \eqref{eq:asmptsVVext} that $V_\text{ext}$ has at most exponential growth, the estimate $\|\eta_x\|\leq C\pn(x)$ and the fact that $\pn$ has exponential decay with arbitrary rate, by Eq. \eqref{eq:expdecaypn}.  
\end{proof}

Through the rest of this section, we assume that $V$ and $V_\text{ext}$ satisfy the assumptions  \eqref{eq:asmptsVVext}.

\subsection{Analysis of $\cG_{N}^{(0)}$ and $\cG_{N}^{(1)}$}\label{sec:GN0}
\begin{lemma}\label{lm:emBcNeB}
We have that 
		\[\begin{split} 
		e^{-B} \cN e^{B} &= \int dx\, \big[ b^*(\gamma_x)b(\gamma_x) + b^*(\sigma_x)b(\sigma_x) + b^*(\gamma_x)b^*(\sigma_x) + b(\gamma_x)b(\sigma_x)\big] +  \|\sigma\|^2 + \cE_{\cN},
		 \end{split}\]
where the self-adjoint operator $\cE_{\cN}$ satisfies $ \pm \cE_{\cN}\leq C N^{-1} (\cN+1)^2$.
\end{lemma}		
\begin{proof}
Observing that 
		\[\cN = \int dx\, a^*_x a_x  =\int dx\, b^*_x b_x +  \cN(\cN-1)/N, \]
the claim follows from the decomposition \eqref{eq:defdx}, the bound \eqref{eq:bnddx} and Lemma \ref{lm:Npow}.
\end{proof}	
\begin{cor}\label{cor:GN0} We have that
		\[\begin{split}
		 \cG_{N}^{(0)} &=  N  \big\langle \pn, \big( -\Delta + V_\text{ext}+(N^3V(N.)\ast \pn) \pn\big)\pn\big\rangle -\frac12 \wh V(0) \| \pn\|_4^4   \\
	& \hspace{0.5cm} - \Big( \big\langle \pn, \big( -\Delta + V_\text{ext}\big)\pn\rangle +  \wh V(0) \| \pn\|_4^4 \Big)\\
	&\hspace{1cm} \times \bigg(  \int dx\, \big[ b^*(\gamma_x)b(\gamma_x) + b^*(\sigma_x)b(\sigma_x)  + b^*(\gamma_x)b^*(\sigma_x) + b(\gamma_x)b(\sigma_x)\big] +  \|\sigma\|^2\bigg) \\ &\hspace{0.5cm}+ \cE_{N}^{(0)},
		 \end{split}\]
where the self-adjoint operator $\cE_{N}^{(0)}$ satisfies $ \pm \cE_{N}^{(0)} \leq C N^{-1} (\cN+1)^2$.
\end{cor}
\begin{proof} 
We notice that 
		\[ \Big| \int dxdy\, N^3V(N(x-y)) \pn^2(x)\pn^2(y) - \wh V(0) \int dx\, \pn^4(x)\Big| \leq CN^{-1} \]
for some $C=C(\|\pn\|_{H^1}, \||\cdot|V\|_1)>0$ and the rest follows from \eqref{lm:emBcNeB} and Lemma \ref{lm:Npow}.
\end{proof}
To analyze $\cG_{N}^{(1)}$, let us define $h_N\in L^2(\bR^3)\cap L^\infty(\bR^3)$ by
		\begin{equation}\label{eq:defhN} h_N = \big(N^3(Vw_\ell) (N.)\ast |\pn|^2\big)\pn. \end{equation}
The proof of the following proposition is a straightforward adaption of \cite[Prop. 6.2]{BSS} and will therefore be omitted. 
\begin{lemma}\label{lm:GN1} There exists a constant $C>0$ such that
		\[\begin{split}
		 \cG_{N}^{(1)} = &\, \big[ \sqrt{N} b(\gamma(h_N) ) + \sqrt{N} b^*(\sigma(h_N) )+\emph{h.c.}\big] + \cE_{N}^{(1)},
		 \end{split}\]
where the self-adjoint operator $\cE_{N}^{(1)}$ satisfies  $ \pm \cE_{N}^{(1)} \leq C  N^{-1/2}(\cN+1)^{3/2}  $. 
\end{lemma}

\subsection{Analysis of $\cG_{N}^{(2)}$}\label{sec:GN2}
In this section, we study $\cG_N^{(2)} = e^{-B} \cL_N^{(2)} e^B$. We split the operator $\cL_N^{(2)}$ introduced in (\ref{eq:cLNj}) acoording to \[ \cL_N^{(2)} =: \cK + \cV_\text{ext} + \cL_{N}^{(2,V)}.  \]
We start with the conjugation of $\cV_\text{ext}$.
\begin{lemma}\label{lm:Vext}
We have that
		\[\begin{split} e^{-B} \cV_\text{ext}e^{B} &= \int dx\,V_\text{ext}(x) \big[ b^*(\gamma_x)b(\gamma_x) + b^*(\sigma_x)b(\sigma_x) + b^*(\gamma_x)b^*(\sigma_x) + b(\gamma_x)b(\sigma_x)\big] \\
		&\hspace{0.5cm} + \tr (\sigma V_\text{ext}  \sigma) + \cE_{\cV_\text{ext}}
		\end{split}\]
for an error $\cE_{\cV_\text{ext}}$ that satisfies $\pm \cE_{\cV_\text{ext}} \leq C N^{-1} (\cV_\text{ext}+\cN+1)(\cN+1).$
\end{lemma}
\begin{proof} The proof is a simple consequence of Cauchy-Schwarz, Lemma \ref{lm:d-bds}, the bounds $ \|\eta^{(n)}_x\|\leq C\pn(x)$ for $n\geq 1$, the fact that $ \pn$ has exponential decay with arbitrary rate, by Eq. \eqref{eq:expdecaypn}, and that $V_\text{ext}$ grows at most exponentially, by Eq. \eqref{eq:asmptsVVext}.
\end{proof}
We continue with the conjugation of $\cL_{N}^{(2,V)}$ which reads by definition
		\begin{align*}
		\begin{split}
		\mathcal{L}_N^{(2,V)} &=  \int dxdy\;  N^3 V(N(x-y)) \pn^2(y)  \Big(b_x^* b_x - \frac{1}{N} a_x^* a_x \Big)   \\
		&\hspace{0.5cm}+ \int   dxdy\; N^3 V(N(x-y)) \pn(x) \pn(y) \Big( b_x^* b_y - \frac{1}{N} a_x^* a_y \Big)  \\
		&\hspace{0.5cm}+ \frac{1}{2}  \int  dxdy\;  N^3 V(N(x-y)) \pn(y) \pn(x) \Big(b_x^* b_y^*  +  b_xb_y \Big).
		\end{split}\end{align*}

\begin{lemma}\label{lm:GN2V}
Set $\cG_N^{(2,V)}=e^{-B} \mathcal{L}_N^{(2,V)} e^{B}$. Then we have that 
	\[\begin{split}
	\cG_N^{(2,V)}=&\,  \int dx \, \big( N^3 V(N.) \ast \pn^2\big)(x) \big[b^*(\gamma_x)b(\gamma_x) + b^*(\gamma_x) b^*(\sigma_x)  \\
	&\qquad+ b(\sigma_x) b(\gamma_x) + b^*(\sigma_x)b(\sigma_x) \big]  \\
	&+   \int dx dy \, N^3 V(N(x-y)) \pn(x) \pn(y) \big[ b^*(\gamma_y)b(\gamma_x) \\
	&\qquad+ b^*(\gamma_y) b^*(\sigma_x)+ b(\sigma_y) b(\gamma_x) + b^*(\sigma_x) b(\sigma_y)   \big]\\
	&+  \frac12 \int dx dy \, N^3 V(N(x-y)) \pn(x) \pn(y) \big[ b^*(\gamma_y)b^*(\gamma_x) \\
	&\qquad+ b^*(\gamma_y) b(\sigma_x)+ b^*(\gamma_x) b(\sigma_y) + b(\sigma_x) b(\sigma_y)  +\emph{h.c.} \big]\\
	&+ \frac{1}{2} \int dx dy \, N^3 V(N(x-y)) \pn(x) \pn(y) \\
	&\hspace{1cm}\times \big[ \big( b^*(\gamma_x) + b(\sigma_x) \big) d_{\eta,y}^*  + d_{\eta,x}^* \big( b^*(\gamma_y) + b(\sigma_y) \big) +\emph{h.c.}] \\
	&+\tr\big( \sigma \big[   N^3 V(N.) \ast \pn^2  + N^3 V(N(x-y)) \pn(x)\pn(y)\big] \sigma\big)\\
	& + \tr\big( \gamma   \big[N^3 V(N(x-y)) \pn(x)\pn(y)\big]  \sigma )+ \big(\widehat{V}(0)- 8\pi \mathfrak{a}_0\big)\|\pn\|_4^4\, \cN+\mathcal{E}_N^{(2,V)}
	\end{split}\]
for an error $\mathcal{E}_N^{(2,V)}$ that satisfies \[ \pm \mathcal{E}_N^{(2,V)} \leq CN^{-1} (\cV_N + \mathcal{N}+1)(\cN+1).\] 
Here, we view $  N^3 V(N.) \ast \pn^2 $ as multiplication operator in $L^2(\bR^3)$ and we identify, by slight abuse of notation, $ N^3 V(N(x-y)) \pn(x)\pn(y)$ with its associated Hilbert-Schmidt operator in $L^2(\bR^3)$.
\end{lemma}
\begin{proof} Let us outline the main steps. First of all, we have
		\[\begin{split}
		& \pm \bigg(  \int dxdy\;  N^2 V(N(x-y)) \pn^2(y)   a_x^* a_x +  \int   dxdy\; N^2 V(N(x-y)) \pn(x) \pn(y)   a_x^* a_y \bigg)\\
		 &\hspace{0.5cm}\leq CN^{-1}(\cN+1),  \\
		 \end{split}\]
so that these contributions can be neglected, by Lemma \ref{lm:Npow}. Similarly, the remaining diagonal terms are easily seen to be equal to 
		\[\begin{split}
		&  \int dxdy\;  N^3 V(N(x-y)) \big[  \pn^2(y)  e^{-B} b_x^* b_x e^{B}  +  \pn(x) \pn(y)  e^{-B} b_x^* b_y e^{B}\big]\\
		& =  \int dx \, \big( N^3 V(N.) \ast \pn^2\big)(x) \big[b^*(\gamma_x)b(\gamma_x) + b^*(\gamma_x) b^*(\sigma_x)  \\
	&\qquad+ b(\sigma_x) b(\gamma_x) + b^*(\sigma_x)b(\sigma_x) \big]  \\
	&\hspace{0.5cm}+   \int dx dy \, N^3 V(N(x-y)) \pn(x) \pn(y) \big[ b^*(\gamma_y)b(\gamma_x) \\
	&\hspace{1.2cm}+ b^*(\gamma_y) b^*(\sigma_x)+ b(\sigma_y) b(\gamma_x) + b^*(\sigma_x) b(\sigma_y)   \big]\\
	&\hspace{0.5cm}+\tr\big( \sigma \big[   N^3 V(N.) \ast \pn^2  + N^3 V(N(x-y)) \pn(x)\pn(y)\big] \sigma\big) + \cE_{N}^{(21,V)}
		\end{split}\]
for an error $\pm \mathcal{E}_N^{(21,V)} \leq CN^{-1} (\mathcal{N}+1)^2$. This follows from the decomposition \eqref{eq:defdx} and the bounds from Lemma \ref{lm:d-bds}. The pairing term can be treated similarly. We compute 
		\[\begin{split}
		e^{-B} b^*_x b^*_y e^{B} &=   b^*(\gamma_y)b^*(\gamma_x) + b^*(\gamma_y) b(\sigma_x)+ b^*(\gamma_x) b(\sigma_y) \\
		&\hspace{0.5cm} + b(\sigma_x) b(\sigma_y)  + \langle\sigma_x, \gamma_y\rangle -  N^{-1} \langle\sigma_x, \gamma_y\rangle \,\cN \\ 
		&\hspace{0.5cm} + \big( b^*(\gamma_x) + b(\sigma_x) \big) d_{\eta,y}^*  + d_{\eta,x}^* \big( b^*(\gamma_y) + b(\sigma_y) \big) \\ 
		&\hspace{0.5cm}  - N^{-1} a^*(\gamma_y) a(\sigma_x) + d^*_{\eta,x}d^*_{\eta,y}.
		\end{split}\]
Now, we observe that 
		\[ \pm \bigg( \int dxdy\, N^2V(N(x-y))\pn(x)\pn(y) \big(\langle\sigma_x, \gamma_y\rangle - k(x;y)\big) \cN\bigg)\leq CN^{-1} \cN, \]	
because $| \langle\sigma_x, \gamma_y\rangle - k(x;y)| \leq C\pn(x)\pn(y)$ (recall the definition (\ref{eq:defgammasigma}) and the bounds \eqref{eq:bndsetaptw}, (\ref{eq:bndsetan})). Since $ N^{3}(Vw_\ell)(N.)$ is an approximation of Dirac's measure of mass $\big(\wh V(0)- 8\pi \mathfrak{a}_0\big)+ \mathcal{O}(N^{-1})$, by Eq. \eqref{eq:Vfa0}, we then find
		\[\begin{split}  \int dxdy\, N^2V(N(x-y))\pn(x)\pn(y) (-k(x;y)) &= \int dxdy\, N^3(Vw_\ell)(N(x-y))\pn^2(x)\pn^2(y) \\
		&= \big(\wh V(0)- 8\pi \mathfrak{a}_0\big)\|\pn\|_4^4 + \mathcal{O}(N^{-1}),  \end{split}\]  
where we used that $\pn \in H^1(\bR^3)$. Similarly, we bound with Cauchy-Schwarz that
		\[ \pm \bigg( \int dxdy\, N^2V(N(x-y))\pn(x)\pn(y)a^*(\gamma_y) a(\sigma_x) +\text{h.c.}\bigg)\leq CN^{-1} (\cN+1)  \]
and from the bound \eqref{eq:bnddxdy}, it is simple to see that for all $\xi\in \cFpn$ we have
		\[\begin{split}
		&\int dxdy\, N^3V(N(x-y))\pn(x)\pn(y) |\langle\xi, d^*_{\eta,x}d^*_{\eta,y}\xi\rangle|  \\
		&\leq C \int dxdy\, NV(N(x-y))\pn(x)\pn(y) \|(\cN+1)\xi\|\Big[    \|\eta_x\| \|\eta_y\|  \Vert (\mathcal{N}+1)^{2} \xi \Vert  \\
	&\hspace{1cm} + \|\eta\|  |\eta(x;y)|   \Vert (\mathcal{N}+1) \xi \Vert +  \|\eta\| \|\eta_y\|  \vert   \Vert a_x (\mathcal{N}+1)^{3/2)} \xi \Vert    \\
	&\hspace{1cm}+ \|\eta\| \| \eta_x\| \Vert a_y(\mathcal{N}+1)^{ 3/2} \xi \Vert +  \|\eta\|^2\Vert a_x a_y (\mathcal{N}+1)  \xi \Vert \Big] \\
	&\leq CN^{-1}\langle\xi, (\cV_N+\cN+1)(\cN+1)\xi\rangle.
		\end{split}\]
Collecting all terms and using the previous bounds, we conclude the proof of the lemma.
\end{proof}

Finally, let us analyze the conjugation of the kinetic energy $\cK$ under $e^{B}$. We start with the following preparation. 
\begin{lemma}\label{lm:nablad-bds}
Assume \eqref{eq:asmptsVVext}. Then, there exists a constant $C>0$ such that for all $n,m\geq 1$, we have that
		\begin{align*} 
		 \int dx \ a^* (\nabla_x \eta_{x}^{(n)}) a(\nabla_x \eta_{x}^{(m)}) \leq C^{n+m}  (\mathcal{N}+1).
		\end{align*}
Moreover, recalling \eqref{eq:defd}, \eqref{eq:defdx} and setting $ \overline{d}_{\eta,x} = d_{\eta,x}+(\cN/N)b^*(\eta_{x})$, there exists for every $n\in\mathbb{N}$ a constant $C>0$ such that 
\begin{equation} \label{eq:nabladx1}
\begin{split}
& \int dx \ \Vert \sqrt{N} (\mathcal{N} + 1)^{(n-1)/2} \nabla_x d_{s\eta,x} \xi \Vert^2  
		\leq C    \Vert (\cK+\cN+1)^{1/2}(\cN+1)^{n/2} \xi \Vert ^2 , \\
		& \int dx \ \Vert \sqrt{N} (\mathcal{N} + 1)^{(n-1)/2} \nabla_x \overline{d}_{s\eta,x} \xi \Vert^2 
		\leq C    \Vert (\cK+\cN +1)^{1/2}(\cN+1)^{n/2} \xi \Vert ^2  
		\end{split}
		\end{equation}
and that
\begin{align*} 
		\begin{split}
		&  \int dx \ \Vert \sqrt{N} (\mathcal{N}+1)^{(n- 1)/2} d_{s\eta}(\nabla_x \eta_{x}) \xi \Vert^2   
		\leq C  \Vert (\cN+1)^{(n+1)/2} \xi \Vert^2, \\	
		&  \int dx \ \Vert \sqrt{N} (\mathcal{N}+1)^{(n- 1)/2} d_{s\eta}^*(\nabla_x \eta_{x}) \xi \Vert^2  
		\leq C  \Vert (\cN+1)^{(n+1)/2} \xi \Vert^2
		\end{split}
		\end{align*}
for all $\xi\in \cF^{\leq N}$ and all $s\in [0;1]$.
\end{lemma}
Lemma \ref{lm:nablad-bds} is a slight generalization of Lemma \cite[Lemma 4.7]{BSS} and can be proved in the same way, based on the kernel bounds from Lemma \ref{lm:bndseta}; we omit the details.
\begin{lemma}\label{lm:eBK}
	We have that
	\[\begin{split}
	&e^{-B} \mathcal{K} e^{B} \\
	& =  \int dx\, \big[ \nabla_x b^*(\gamma_x) \nabla_x b(\gamma_x) + \nabla_x b^* (\gamma_x) \nabla_x b^*(\sigma_x) + \nabla_x b(\sigma_x) \nabla_x b(\gamma_x) + \nabla_x b^*(\sigma_x) \nabla_x b(\sigma_x)\big] \\ 
	&\hspace{0.5cm}+\int dx\, \big[b(\nabla_x \eta_x) \nabla_x d_{\eta,x} + \emph{h.c.}\big]  +   \|\nabla_1\sigma\|^2(1-\cN/N)+  N^{-1}\Vert \nabla_1 \eta \Vert^2 \\
	&\hspace{0.5cm}+ \frac{\Vert \nabla_1 \eta \Vert^2}{N} \Big( \int dx \, \big[ b^*(\gamma_x) b(\gamma_x)+ b^*(\sigma_x) b(\sigma_x)  + b^*(\gamma_x) b^*(\sigma_x)+ b(\gamma_x) b(\sigma_x)\big] +  \Vert \sigma \Vert^2 \Big) \\
	&\hspace{0.5cm}+ \mathcal{E}_{\cK}
	\end{split}\]
where the error $\cE_{\cK}$ satisfies \[ \pm \mathcal{E}_{\cK} \leq CN^{-1/2} (\mathcal{K}+\cN^2+1)(\mathcal{N}+1).\]
\end{lemma}
\begin{proof} The proof is a straightforward adapation of \cite[Lemma 7.2]{BBCS4} to the setting of trapped particles, so let us focus on the main steps. We start with the identity
	\begin{equation} \label{rewrite K}
	\begin{split}
	\mathcal{K} 
	&= \int dx \nabla_x b_x^* \nabla_x b_x +  N^{-1} \int dx dy \nabla_x b_x^* b_y^* b_y \nabla_x b_x \\
	&\hspace{0.5cm}+  N^{-2}(\mathcal{N}-1)\int dx \nabla_x b_x^* \nabla_x b_x  + N^{-2}\mathcal{K}(\mathcal{N}+1)^2
	\end{split}	
	\end{equation}
and notice that, by the rough bound \eqref{bound conjugation K N}, the last two terms are errors bounded by 
		\[ \pm N^{-2}e^{-B}\bigg( (\mathcal{N}-1)\int dx \, \nabla_x b_x^* \nabla_x b_x  +  \mathcal{K}(\mathcal{N}+1)^2\bigg)e^{B}\leq CN^{-1} (\mathcal{K}+\cN^2+1)(\mathcal{N}+1). \] 
Hence, let us consider the remaining two terms on the right hand side in \eqref{rewrite K}, starting with the first term. We apply the decomposition \eqref{eq:defdx} and find that	
	\begin{align*}
	\int dx\, e^{-B} \nabla_x b_x^* \nabla_x b_x e^{B} 
	=E_1 +  E_2 + E_3 +  \int dx \big[b(\nabla_x \eta_x) \nabla_x d_{\eta,x} + \text{h.c.}\big],
	\end{align*}
where
	\begin{align*}
	E_1 &= \int dx \ (\nabla_x b^*(\gamma_x)+  b (\nabla_x \sigma_x))(\nabla_x b(\gamma_x)+ b^*(\nabla_x \sigma_x)),  \\
	E_2 &= \int dx \ (\nabla_x b^*(\gamma_x)+ b(\nabla_x  (\sigma_x - \eta_x))) \nabla_x d_{\eta,x} +\text{h.c.}, \\
	E_3 &= \int dx \ \nabla_x d_{\eta,x}^* \nabla_x d_{\eta,x}.
	\end{align*}
The contributions $E_2$ and $E_3$ are error terms. This follows from \eqref{eq:nabladx1}, which implies 
	\begin{align*}
	0\leq  E_2 \leq CN^{-1}  (\cN+\cK+1)(\cN+1)
	\end{align*}
and, using in addition the bounds from Lemma \ref{lm:bndseta}, that 
	\begin{align*}
	\vert \langle \xi, E_3 \xi \rangle \vert 
	&\leq CN^{-1/2} \Big[ \Vert \mathcal{K}^{1/2} \xi \Vert + \sum_{m\geq 2} \frac{1}{m!} \Vert \nabla_1 \eta^{(m)} \Vert \cdot \Vert (\mathcal{N}+1)^{1/2} \xi \Vert \Big] \\
	&\hspace{1.5cm}\times \Vert (\mathcal{K}+\cN+1)^{1/2} (\mathcal{N}+1)^{1/2}\xi \Vert   \\
	&\leq CN^{-1/2} \langle \xi, (\cN+\mathcal{K} +1)(\mathcal{N}+1) \xi\rangle.
	\end{align*}
We are left with $E_1$. Using the standard commutation relations, 
	\begin{align*}
	E_1 
	&=  \int dx \big[ \nabla_x b^*(\gamma_x) \nabla_x b(\gamma_x) + \nabla_x b^* (\gamma_x) b^*( \nabla_x \sigma_x) \\	&\hspace{1.5cm}+ b(\nabla_x \sigma_x) \nabla_x b(\gamma_x) +  b^*( \nabla_x \sigma_x)  b( \nabla_x \sigma_x)\big] \\
	&\hspace{0.5cm}+  \Vert \nabla_1 \sigma \Vert^2 (1-\cN/N) - \frac{1}{N} \int dx \, a^*(\nabla_x \sigma_x) a(\nabla_x \sigma_x).
	\end{align*}
By Lemma \ref{lm:nablad-bds}, the last term on the right hand side is an error term bounded by
	\begin{align*}
	0\leq  N^{-1} \int dx \, a^*(\nabla_x \sigma_x) a(\nabla_x \sigma_x) \leq CN^{-1} (\cN+1)
	\end{align*}
and therefore we conclude that
		\[\begin{split}
		 &\int dx\, e^{-B} \nabla_x b_x^* \nabla_x b_x e^{B} \\
		& =  \int dx\, \big[ \nabla_x b^*(\gamma_x) \nabla_x b(\gamma_x) + \nabla_x b^* (\gamma_x)  b^*( \nabla_x  \sigma_x) +  b( \nabla_x \sigma_x) \nabla_x b(\gamma_x) +  b^*( \nabla_x \sigma_x)  b( \nabla_x  \sigma_x)\big] \\ 
		&\hspace{0.5cm}+\int dx\, \big[b(\nabla_x \eta_x) \nabla_x d_{\eta,x} + \text{h.c.}\big]  +   \|\nabla_1\sigma\|^2(1-\cN/N) +\wt\cE_{\cK},
		\end{split}\]
for an error $\pm \wt\cE_{\cK}\leq CN^{-1/2}  (\cN+\cK+1)(\cN+1)$. It remains to analyze the second term on the right hand side of \eqref{rewrite K}. A straightforward adaption of the analysis in \cite[Eq. (7.18) to (7.20)]{BBCS4} shows that this term is given by 
		\[\begin{split}
		&N^{-1} \int dx dy \nabla_x b_x^* b_y^* b_y \nabla_x b_x \\
		& = N^{-1}\Vert \nabla_1 \eta \Vert^2 \Big( \int dx \, \big[ b^*(\gamma_x) b(\gamma_x)+ b^*(\sigma_x) b(\sigma_x)  + b^*(\gamma_x) b^*(\sigma_x)+ b(\gamma_x) b(\sigma_x)\big] +  \Vert \sigma \Vert^2 \Big) \\
		&\hspace{0.5cm}+ N^{-1}\Vert \nabla_1 \eta \Vert^2 + \wt{\mathcal{E}}'_{\cK}
		\end{split}\]
for an error $\pm \wt\cE_{\cK}'\leq CN^{-1/2}  ( \cK+\cN^2+1)(\cN+1)$. Combining the two last identities and collecting the error terms concludes the claim.
\end{proof}

\subsection{Analysis of $\cG_{N}^{(3)}$}\label{sec:GN3}
\begin{lemma}\label{lm:GN3}
	Let $h_N = \big(N^3(Vw_\ell) (N.)\ast |\pn|^2\big)\pn$, as in Eq. \eqref{eq:defhN}. Then, we have that
	\begin{align*}
	\begin{split}
	\mathcal{G}_N^{(3)} 	&= \int dx dy \, N^{5/2} V(N(x-y)) \pn(y) \Big( b_x^* b_y^* \big[ b(\gamma_x) + b^*(\sigma_x)\big] + \emph{h.c.} \Big) \\
	&\hspace{0.5cm}-\big[ \sqrt{N} b(\gamma(h_N) ) + \sqrt{N} b^*(\sigma(h_N) )+\emph{h.c.}\big] + \mathcal{E}_N^{(3)}
	\end{split}
	\end{align*}
	for an error $ \mathcal{E}_N^{(3)}$ that satisfies
	\begin{align*}
	\pm \mathcal{E}_N^{(3)} \leq  CN^{-1/2} ( \mathcal{V}_N +  \mathcal{N}+1)(\mathcal{N}+1).
	\end{align*}
\end{lemma}
\begin{proof}
	We recall \eqref{eq:cLNj} and use $a_y^* a_x = b_y^* b_x + N^{-1} a_y^* \mathcal{N} a_x$ to write $\cL_N^{(3)}$ as
	\begin{equation}\label{eq:GN3-1}
	\begin{split}  
	\mathcal{L}_N^{(3)}
	&=   \int dx dy \, N^{5/2} V(N(x-y)) \pn(y) \big[b_x^* b_y^* b_x + \text{h.c.}\big] \\
	&\quad+  \int dx dy \, N^{3/2} V(N(x-y)) \pn(y) \big[b_x^* a_y^*\, \mathcal{N} a_x + \text{h.c.}\big].
	\end{split}\end{equation} 
The contribution arising from the second term on the right hand side of the last equation is an error. Indeed, we infer from Cauchy-Schwarz and Lemma \ref{lm:roughVNK} that
	\begin{align*}
	&  \int dx dy \, N^{3/2} V(N(x-y)) \pn(y)\big| \langle e^{B}\xi, b_x^* a_y^* \mathcal{N} a_x e^{B}\xi 
	\rangle \big| \\
	&\hspace{0.5cm}\leq CN^{-1}\langle \xi, (\cV_N + N)(\cN+1)\xi\rangle^{1/2} \langle \xi, (\cN+1)^2\xi\rangle^{1/2}\\
	& \hspace{0.5cm}\leq CN^{-1/2} \langle\xi, (\cV_N+\cN+1)(\cN+1)\xi\rangle.
	\end{align*}
As for the first contribution from the r.h.s. of (\ref{eq:GN3-1}), we decompose it with \eqref{eq:defdx} into
the sum of the terms  	
	\begin{equation}\label{eq:M0toM3}
	\begin{split}
	M_0 &=  \int dx dy\, N^{5/2} V(N(x-y)) \pn(y) \big[b^*(\gamma_x)+ b(\sigma_x)\big]  \\
	&\hspace{0.7cm} \times\big [b^*(\gamma_y) + b(\sigma_y)\big] \big[b(\gamma_x)+ b^*(\sigma_x)\big] + \text{h.c.},\\
	M_1 &=  \int dx dy\, N^{5/2} V(N(x-y)) \pn(y) \\
	&\hspace{0.7cm}\times \Big(d_{\eta,x}^* \big[b^*(\gamma_y)+ b(\sigma_y)\big]\big[b(\gamma_x)+ b^*(\sigma_x)\big] \\
	&\hspace{1.5cm}+ \big[b^*(\gamma_x)+ b(\sigma_x)\big]d_{\eta,y}^*\big[b(\gamma_x)+ b^*(\sigma_x)\big] \\
	&\hspace{1.5cm}+ \big[b^*(\gamma_x)+ b(\sigma_x)\big] \big[b^*(\gamma_y)+ b(\sigma_y)\big]d_{\eta,x} + \text{h.c.}\Big),\\
	M_2 &=  \int dx dy\, N^{5/2} V(N(x-y)) \pn(y) \Big( d_{\eta,x}^* d_{\eta,y}^* \big[b(\gamma_x)+ b^*(\sigma_x)\big] \\
	&\hspace{0.7cm}+ d_{\eta,x}^*\big[ b^*(\gamma_y)+ b(\sigma_y)\big]d_{\eta,x} 
	+ \big[b^*(\gamma_x)+ b(\sigma_x)\big]d_{\eta,y}^* d_{\eta,x} + \text{h.c.}  \Big),\\
	M_3 &=   \int dx dy \, N^{5/2} V(N(x-y)) \pn(y) d_{\eta,x}^* d_{\eta,y}^* d_{\eta,x} + \text{h.c.}
	\end{split}
	\end{equation}
Notice that the index $i$ in $M_i$ counts the number of $d_{\eta}$-operators it contains. 

The operators $M_1, M_2$ and $M_3$ are error terms. Let us illustrate this for $M_1$, the contributions $M_2$ and $M_3$ can be handled analogously. If we apply Cauchy-Schwarz together with the bounds \eqref{eq:bnddx}, \eqref{eq:bndaydxbar} and Lemma \ref{lm:bndseta}, we find that 
		\[\begin{split}
		& \int dx dy\, N^{5/2} V(N(x-y)) \pn(y) |\langle\xi, \!d_{\eta,x}^* (b^*(\gamma_y)\!+\! b(\sigma_y) )(b(\gamma_x)\!+\! b^*\!(\sigma_x) ) \xi\rangle|\\
		&\,\leq C \int dx dy\, N^{5/2} V(N(x-y)) \pn(y)  \big[ \|\eta_{y}\|\|  d_{\eta,x}\xi\| + \| (\cN+1)^{-1/2}a_y d_{\eta,x}\xi\|\big]\\
		&\hspace{1cm}\times \big[ \|\eta_{x}\|\| (\cN+1)^{1/2}\xi\| + \| (\cN+1)^{1/2}a_x \xi\|\big]\\
		&\,\leq C \int dx dy\, N^{3/2} V(N(x-y)) \pn(y)  \big[  \|\eta_x\|\|\eta_y\| \| (\cN+1) \xi\| + CN\| (\cN+1)^{1/2}\xi\| \\
		&\hspace{4cm} +  \|a_x(\cN+1)^{1/2} \xi\| + \|\eta_x\|\|a_y(\cN+1)\xi\| + \|a_xa_y(\cN+1)^{1/2}\xi\| \big]\\
		&\hspace{1cm}\times \big[ \|\eta_{x}\|\| (\cN+1) \xi\| + \|a_x (\cN+1)^{1/2} \xi\|\big]\\
		&\leq C N^{-1/2}\langle\xi, (\cV_N + \cN+1)(\cN+1)\xi\rangle.
		\end{split}\] 	
Notice that we used that $\|\eta\|_\infty\leq CN$ and that $\|\eta_x\|\leq C\pn(x)$. Similarly, we find that
		\[\begin{split}
		& \int dx dy\, N^{5/2} V(N(x-y)) \pn(y)   \\
		&\hspace{0.5cm}\times \| (\cN+1)^{-1/2}d_{\eta,y}(b(\gamma_x)\!+ \!b^*(\sigma_x) ) \xi\|  \| (\cN+1)^{1/2} (b(\gamma_x)\!+\!b^*(\sigma_x) ) \xi \| \\
		&+\int dx dy\, N^{5/2} V(N(x-y)) \pn(y) \\
		&\hspace{1cm}\times \| (b(\gamma_y)\!+\!b^*(\sigma_y)) (b(\gamma_x)\!+\! b^*(\sigma_x) )\xi\| \|d_{\eta,x}  \xi \|\\
		&\leq C N^{-1/2}\langle\xi, (\cV_N + \cN+1)(\cN+1)\xi\rangle,
		\end{split}\] 	
so that altogether $\pm (M_1+\text{h.c.}) \leq C N^{-1/2}  (\cV_N + \cN+1)(\cN+1)$. The contributions $M_2$ and $M_3$ can be controlled similarly, using additionally the bound \eqref{eq:bnddxdy}. This shows 
		\[ \pm (M_2 + M_3+\text{h.c.})\leq C N^{-1/2} (\cV_N + \cN+1)(\cN+1). \]
Now, let us determine the main contributions to $M_0$, defined in \eqref{eq:M0toM3}. With
		\[\begin{split}
		 &\big[b^*(\gamma_x)+ b(\sigma_x)\big]  \big [b^*(\gamma_y) + b(\sigma_y)\big] \\
		 & = b^*_x b^*_y +  \langle\sigma_x, \gamma_y\rangle (1-\cN/N) - N^{-1}a^*(\gamma_y)a(\sigma_x)\\
		 &\hspace{0.5cm}+ b^*_x \big[b^*(\gamma_y)-b^*_y\big]  +  \big[b^*(\gamma_x)-b^*_x\big]b^*(\gamma_y)  \\
		 &\hspace{0.5cm}  + b^*(\gamma_x)b(\sigma_y) + b^*(\gamma_y)b(\sigma_x)+b(\sigma_x)b(\sigma_y) \\
		 \end{split}\] 
and with the bounds \eqref{eq:bndsetaptw}, \eqref{eq:bndsetan}, it is simple to verify that 
		\[\begin{split}\pm& \bigg( M_0- \int dx dy \, N^{5/2} V(N(x-y)) \pn(y) b_x^* b_y^*  \big[b(\gamma_x) + b^*(\sigma_x)  \big] \\
		 &\hspace{0.2cm}- \int dx dy \, N^{5/2} V(N(x-y)) \pn(y) \langle\sigma_x, \gamma_y\rangle \big[b(\gamma_x) + b^*(\sigma_x)  \big] +\text{h.c.}\bigg)\\
		 &\leq CN^{-1/2}(\cN+1)^{3/2}.
		\end{split} \]
Notice that Lemma \ref{lm:bndseta} also implies $| \langle\sigma_x, \gamma_y\rangle - k(x;y)| \leq C \pn(x)\pn(y)$, with $k$ defined in \eqref{eq:defk}, so that we can simplify this further to
		\[\begin{split}\pm& \bigg( M_0- \int dx dy \, N^{5/2} V(N(x-y)) \pn(y) b_x^* b_y^*  \big[b(\gamma_x) + b^*(\sigma_x)  \big] \\
		 &\hspace{0.2cm}+\sqrt N \int dx dy \, N^{3} (Vw_\ell )(N(x-y))\pn^2(y)  \pn(x)  \big[b(\gamma_x) + b^*(\sigma_x)  \big] +\text{h.c.}\bigg)\\
		 &\leq CN^{-1/2}(\cN+1)^{3/2}.
		 \end{split}\]
Finally, combining the estimates on $M_0$, $M_1$, $M_2$ and $M_3$ with the observation that
		\[\begin{split}& \sqrt N \int dx \big(N^3(Vw_\ell) (N.)\ast |\pn|^2\big)\pn(x)   \big[b(\gamma_x) + b^*(\sigma_x)  \big] \\
		&  = \sqrt{N} b(\gamma(h_N) ) + \sqrt{N} b^*(\sigma(h_N) ), \end{split} \]
where $h_N = \big(N^3(Vw_\ell) (N.)\ast |\pn|^2\big)\pn $, we conclude the proof of the lemma.
\end{proof}

\subsection{Analysis of $\cG_{N}^{(4)}$}\label{sec:GN4}

\begin{lemma}\label{lm:GN4}
	We have that
	\[\begin{split}
	\mathcal{G}_N^{(4)}&=\cV_N+ \frac{1}{2} \int dx dy \, N^2 V(N(x-y))   \vert \langle \sigma_x, \gamma_y \rangle \vert^2 \Big( 1 + \frac{1}{N} - \frac{2\mathcal{N}}{N} \Big) \\
	&\hspace{0.5cm} - \frac{1}{2} \int dx dy \, N^3(Vw_\ell)(N(x-y))  \pn(x)\pn(y)  \big[b^*(\gamma_x) b^*(\gamma_y)  \\
	&\hspace{1.5cm}+ b^*(\gamma_y) b(\sigma_x) + b^*(\gamma_x) b(\sigma_y) + b(\sigma_x) b(\sigma_y) + \emph{h.c.}\big]\\
	&\hspace{0.5cm}- \frac12 \int dx dy \, N^3 (Vw_\ell)(N(x-y)) \pn(y)\pn(x)\\
	&\hspace{1.5cm} \times \big[d_{\eta,x}^* (b^*(\gamma_y) + b(\sigma_y)) + (b^*(\gamma_x)+ b(\sigma_x))d_{\eta,y}^* +\emph{h.c.}\big] \\
	&\hspace{0.5cm}+  \int dx dy\, N^3 (Vw_\ell^2)(N(x-y))  \pn^2(x)\pn^2(y) \\
	&\hspace{1cm} \times\bigg(\int du\,  \big[b^*(\gamma_u) b^*(\sigma_u) + b^*(\gamma_u) b(\gamma_u)  \\
	&\hspace{2.2cm}+ b^*(\sigma_u) b(\sigma_u)+ b(\sigma_u) b(\gamma_u)\big] + \|\sigma_u\|^2 \bigg)  
	+ \mathcal{E}_N^{(4)}
	\end{split}\]
	for an error $\cE_N^{(4)}$ that satisfies $\pm \mathcal{E}_N^{(4)} \leq CN^{-1/2}(\cV_N+\cN+1)(\cN+1) $.
\end{lemma}
\begin{proof}
	We closely follow \cite[Lemma 7.4]{BBCS4} and outline the main steps. Using that
	\begin{align*}
	a_x^* a_y^* a_y a_x
	= b_x^* b_y^* b_y b_x \left( 1 - \frac{3}{N} + \frac{2\mathcal{N}}{N} \right) + a_x^* a_y^* a_y a_x \Theta (\cN)
	\end{align*}
	for some function $\Theta: \mathbb{N}\to\mathbb{R}$ that satisfies $\pm \Theta(\cN) \leq N^{-2}(\mathcal{N}+1)^2$, we split
	\begin{align*}
	\mathcal{G}_N^{(4)}
	&= \frac{1}{2} \int dx dy \; N^2 V(N(x-y))  e^{-B} b_x^* b_y^* b_x b_y \left(1 - \frac{3}{N} + \frac{2\,\mathcal{N}}{N}\right) e^{B}\\
	&\hspace{0.5cm} + \frac{1}{2} \int dx dy \; N^2 V(N(x-y))   e^{-B} a_x^* a_y^* a_y a_x \Theta(\cN) e^{B}.
	\end{align*}
Using the rough estimate \eqref{bound conjugation V_N} for the conjugation of $\mathcal{V}_N$, we immediately find that
	\begin{align*}
	&\left\vert \frac{1}{2} \int dx dy \, N^2 V(N(x-y))   \langle\xi, e^{-B} a_x^* a_y^* a_y a_x \Theta(\cN) e^{B}\xi  \rangle \right\vert\\
	&\hspace{0.5cm}\leq  CN^{-1} \langle \xi, (\mathcal{V}_N +\cN+1) (\mathcal{N} + 1) \xi \rangle.
	\end{align*}
We can therefore split $ \mathcal{G}_N^{(4)}$ into 
	\begin{equation} \label{reduction G4}
	\begin{split}
	\mathcal{G}_N^{(4)}
	&= \frac{N+1}{2N} \int dx dy \, N^2 V(N(x-y))  e^{-B} b_x^* b_y^* b_y b_x e^{B} \\
	&\qquad+ \frac{1}{N} \int dudx dy \, N^2 V(N(x-y))  e^{-B} b_x^* b_y^* b^*_ub_u b_y b_x e^{B} + \tilde{\mathcal{E}}_1
	\end{split}
	\end{equation}
for some error $\wt{\mathcal{E}}_1$ that satisfies $\pm \wt \cE_1 \leq CN^{-1}( \mathcal{V}_N+\cN+1) (\mathcal{N} + 1)$. 

Now, we analyze the remaining two contributions on the right hand side in \eqref{reduction G4}, starting with the one in the first line. Applying Eq. \eqref{eq:defdx}, we split this term into
	\begin{align*}
	&\frac{(N+1)}{2N} \int dx dy \, N^2 V(N(x-y))  e^{-B} b_x^* b_y^* b_y b_x e^{B} 
	= V_0 + V_1 + V_2 + V_3 + V_4,
	\end{align*}
where 
\begin{align*}\begin{split} 
	V_0 &=  \frac{(N+1)}{2N} \int dx dy \, N^2 V(N(x-y)) \\
	&\hspace{0.5cm} \times \big[b^*(\gamma_x)b^*(\gamma_y) + b^*(\gamma_x) b(\sigma_x) + b^*(\gamma_y) b(\sigma_y) + b(\sigma_x)b(\sigma_y) + \langle \sigma_x, \gamma_y \rangle  \\
	 &\hspace{1cm}- N^{-1}\langle \sigma_x, \gamma_y \rangle\, \cN- N^{-1}a^*(\gamma_y) a(\sigma_x)\big] \\
	&\hspace{0.5cm} \times \big[b(\gamma_y) b(\gamma_x) + b^*(\sigma_y) b(\gamma_x) + b^*(\sigma_x) b(\gamma_y) + b^*(\sigma_y) b^*(\sigma_x) + \langle \sigma_x, \gamma_y\rangle\\
	&\hspace{1cm}- N^{-1}\langle \sigma_x, \gamma_y \rangle\, \cN- N^{-1}a^*(\sigma_x) a(\gamma_y)\big], \\
	\end{split}
	\end{align*}
	\begin{align*} 
	\begin{split}
	V_1 &=  \frac{(N+1)}{2N} \int dx dy \, N^2 V(N(x-y)) \big[ d_{\eta,x}^* (b^*(\gamma_y)+ b(\sigma_y))+ (b^*(\gamma_x)+b(\sigma_x))d_{\eta,y}^* ] \\
	&\hspace{0.5cm} \times \big[b(\gamma_y) b(\gamma_x) + b^*(\sigma_y) b(\gamma_x) + b^*(\sigma_x) b(\gamma_y) + b^*(\sigma_y) b^*(\sigma_x) + \langle \sigma_x, \gamma_y\rangle\\
	&\hspace{1cm}- N^{-1}\langle \sigma_x, \gamma_y \rangle\, \cN- N^{-1}a^*(\sigma_x) a(\gamma_y)\big] +\text{h.c.},
	\end{split}
	\end{align*}
as well as
	\begin{align*} 
	\begin{split}
	V_2 &= \frac{(N+1)}{2N} \int dx dy \, N^2 V(N(x-y)) \big[(b^*(\gamma_x)+ b(\sigma_x)) d_{\eta,y}^* (b(\gamma_y)+ b^*(\sigma_y)) d_{\eta,x} \\
	&\hspace{1cm}+ (b^*(\gamma_x)+b(\sigma_x))d_{\eta,y}^* d_{\eta,y} (b(\gamma_x)+b^*(\sigma_x)) \\
	&\hspace{1cm}+ d_{\eta,x}^* (b^*(\gamma_y)+ b(\sigma_y))(b(\gamma_y)+b^*(\sigma_y)) d_{\eta,x} \\
	&\hspace{1cm}+ d_{\eta,x}^* (b^*(\gamma_y)+b(\sigma_y)) d_{\eta,y}(b(\gamma_x)+ b^*(\sigma_x))\big] \\
	&+  \frac{(N+1)}{2N} \int dx dy \, N^2 V(N(x-y)) \big[ d_{\eta,x}^* d_{\eta,y}^* (b(\gamma_y) + b^*(\sigma_y))  (b(\gamma_x) + b^*(\sigma_x))  +\text{ h.c.} \big], \\
	V_3 &= \frac{(N+1)}{2N} \int dx dy \, N^2 V(N(x-y))  \\ 
	&\hspace{0.5cm} \times \big[ (b^*(\gamma_x)+b(\sigma_x))d_{\eta,y}^* + d_{\eta,x}^*(b^*(\gamma_y)+ b(\sigma_y) \big]d_{\eta,y} d_{\eta,x} + \text{h.c.}  ,\\
	V_4 &=  \frac{(N+1)}{2N} \int dx dy \, N^2 V(N(x-y))  d_{\eta,x}^* d_{\eta,y}^* d_{\eta,y} d_{\eta,x}.
	\end{split}\end{align*}
The only relevant contributions to the energy are contained in $V_0$ and $V_1$ while $V_2$, $V_3$ and $V_4$ are negligible. To see this, let us start with $V_4$. Applying  \eqref{eq:bnddxdy}, we get
	\begin{align*}
	\vert \langle \xi, V_4 \xi \rangle \vert
	&\leq C \int dx dy \, N^2 V(N(x-y))  \Vert d_{\eta,x} d_{\eta,y} \xi \Vert^2 \\
	&\leq C  \int dx dy \, N^2 V(N(x-y))  N^{-4}\big[  \pn^2(x)\Vert a_y (\mathcal{N}+1)^{5/2}\xi \Vert^2   \\
	&\hspace{1.5cm}+ N^2 \pn^2(x)\pn^2(y) \Vert (\mathcal{N}+1)^2 \xi \Vert^2 + \Vert a_xa_y (\mathcal{N}+1)^2 \xi \Vert^2  \big] \\
	&\leq N^{-1} \langle \xi, (\mathcal{V}_N+\cN+1) (\mathcal{N}+1) \xi \rangle. 
	\end{align*}
Similarly, we can use the bounds from Lemma \ref{lm:d-bds} to show that
\[ \pm V_2 , V_3 \leq CN^{-1} (\mathcal{V}_N+\cN+1) (\mathcal{N}+1). \] 
It remains to extract the order one contributions to $V_0$ and $V_1$. Recalling that, by Lemma~\ref{lm:bndseta}, $ | \langle\sigma_x, \gamma_y\rangle - k(x;y)| \leq C \pn(x)\pn(y)$ we find that
		\[\begin{split} V_0 + V_1 &=\cV_N+ \frac{1}{2} \int dx dy \, N^2 V(N(x-y)) \pn(y) \vert \langle \sigma_x, \gamma_y \rangle \vert^2 \Big( 1 + \frac{1}{N} - \frac{2\mathcal{N}}{N} \Big) \\
	&\hspace{0.4cm} - \frac{1}{2} \int dx dy \, N^3(Vw_\ell)(N(x-y))  \pn(x)\pn(y)  \big[b^*(\gamma_x) b^*(\gamma_y)  \\
	&\hspace{0.9cm}+ b^*(\gamma_y) b(\sigma_x) + b^*(\gamma_x) b(\sigma_y) + b(\sigma_x) b(\sigma_y) + \text{h.c.}\big]\\
	&\hspace{0.4cm}- \frac12 \int dx dy \, N^3 (Vw_\ell)(N(x-y)) \pn(y)\pn(x)\\
	&\hspace{0.6cm} \times \big[d_{\eta,x}^* (b^*(\gamma_y) + b(\sigma_y)) + (b^*(\gamma_x)+ b(\sigma_x))d_{\eta,y}^* +\text{h.c.}\big] +\wt\cE_2,
	\end{split}\]
where the error $\wt\cE_2$ is such that $\pm \wt\cE_2\leq CN^{-1} (\mathcal{V}_N+\cN+1) (\mathcal{N}+1)$. 

This concludes the analysis of the first term on the right hand side in \eqref{reduction G4}. The second term is treated similarly and an adaption of the arguments from above yields
		\[\begin{split}
		& \frac{1}{N} \int dudx dy \, N^2 V(N(x-y))  e^{-B} b_x^* b_y^* b^*_ub_u b_y b_x e^{B} \\
		& =  \int dx dy\, N^3 (Vw_\ell^2)(N(x-y))  \pn^2(x)\pn^2(y) \\
	&\hspace{0.5cm} \times\bigg(\int du\,  \big[b^*(\gamma_u) b^*(\sigma_u) + b^*(\gamma_u) b(\gamma_u)  \\
	&\hspace{1.7cm}+ b^*(\sigma_u) b(\sigma_u)+ b(\sigma_u) b(\gamma_u)\big] + \|\sigma\|^2 \bigg)  
	+ \wt \cE_3 
		\end{split}\]
for some error $\pm \wt\cE_3\leq CN^{-1} (\mathcal{V}_N+\cN+1) (\mathcal{N}+1)$. The details are analogous to those in \cite[Lemma 7.4]{BBCS4}, taking into account the different setting; we omit the details.
\end{proof}

\subsection{Proof of Proposition \ref{prop:GN}}\label{sec:proofpropGN}

In this section, we collect the results about $\cG_N$ and prove Proposition \ref{prop:GN}. As mentioned at the beginning of this section, we will focus on the decomposition \eqref{eq:propGN} and the bounds \eqref{eq:GDelta}, \eqref{eq:cEcGNbnd}; the commutator bounds \eqref{eq:adkN} can be proved in the same way. 

Combining the results of Lemmas \ref{cor:GN0}, \ref{lm:GN1}, \ref{lm:Vext}, \ref{lm:GN2V}, \ref{lm:eBK}, \ref{lm:GN3} and \ref{lm:GN4}, we find that 
		\begin{equation}\label{eq:cGNid1}
		\cG_N = \widetilde{\kappa}_{\cG_N} + \int dxdy\, \Big[\wt \Phi(x;y) b^*_x b_y + \frac12 \wt \Gamma(x;y) \big(b^*_x b^*_y+b_xb_y\big)\Big] +\cD_{\cG_N}+ \cC_{\cG_N} +\cV_N+ \wt{\cE}_{\cG_N}
		\end{equation}
where
		\[\begin{split}
		\widetilde{\kappa}_{\cG_N}& = N\big\langle \pn, (-\Delta +V_\text{ext} + \wh V(0)\pn^2/2)\pn\big\rangle-\frac12\wh V(0)\|\pn\|_4^4  + N^{-1}\|\nabla_1\eta\|^2\\
		&\hspace{0.5cm}  + \tr\big(\sigma (-\Delta+V_\text{ext})\sigma\big) + \tr\big( \gamma   \big[N^3 V(N(x-y)) \pn(x)\pn(y)\big]  \sigma ) \\
		&\hspace{0.5cm}  +\tr\big( \sigma \big[   N^3 V(N.) \ast \pn^2  + N^3 V(N(x-y)) \pn(x)\pn(y)\big] \sigma\big)\\
		&\hspace{0.5cm} + \frac12\big(1+N^{-1}\big) \int dxdy\, N^2V(N(x-y))| \langle\sigma_x, \gamma_y\rangle|^2 \\
		&\hspace{0.5cm} + \int dxdy\,N^3(Vw_\ell^2)(N(x-y))\pn^2(x)\pn^2(y) \|\sigma\|^2\\
		&\hspace{0.5cm}+ N^{-1}\|\nabla_1\eta\|^2\|\sigma\|^2- \big( \langle \pn, (-\Delta  +V_\text{ext})\pn\rangle  + \wh V(0)\|\pn\|_4^4\big)\|\sigma\|^2 , 
		\end{split}\] 
where the operators $ \wt \Phi$ and $\wt \Gamma$ are defined by
		\[\begin{split}
		\wt \Phi & = \gamma \big(-\Delta+V_\text{ext}+ N^3V(N.)\ast\pn^2 + N^3V(N(x-y))\pn(x)\pn(y)\big)\gamma \\
		&\hspace{0.4cm}+ \sigma \big(-\Delta+V_\text{ext}+ N^3V(N.)\ast\pn^2 + N^3V(N(x-y))\pn(x)\pn(y)\big)\sigma\\
		&\hspace{0.4cm}+\frac12\Big( \gamma \big(  N^3(Vf_\ell)(N(x-y))\pn(x)\pn(y)\big)\sigma+\text{h.c.}\Big)+ \big(\widehat V(0)-8\pi \mathfrak{a}_0\big)\\
		&\hspace{0.4cm} - \big( \langle \pn, (-\Delta  +V_\text{ext})\pn\rangle  + \wh V(0)\|\pn\|_4^4\big) (\gamma^2 + \sigma^2),\\
		&\hspace{0.4cm} +2\Big( \int dxdy\, N^3(Vw_\ell^2)(N(x-y))\pn^2(x)\pn^2(y) + N^{-1}\|\nabla_1\eta\|^2\Big)  \sigma^2,\\
		\wt \Gamma & =  \gamma \big(  N^3(Vf_\ell)(N(x-y))\pn(x)\pn(y)\big)\gamma\\
		&\hspace{0.5cm}+\sigma \!\big( N^3(Vf_\ell)(N(x-y))\pn(x)\pn(y)\big)\!\sigma\\
		&\hspace{0.5cm} +\Big( \sigma \big(-\Delta+V_\text{ext}+ N^3V(N.)\ast\pn^2 + N^3V(N(x-y))\pn(x)\pn(y)\big)\gamma + \text{h.c.}\Big)\\
		&\hspace{0.5cm} + \Big( \int dxdy\, N^3(Vw_\ell^2)(N(x-y))\pn^2(x)\pn^2(y) \Big)  \big( \gamma \sigma +\text{h.c.}\big)  \\
		&\hspace{0.5cm} - \big( \langle \pn, (-\Delta  +V_\text{ext})\pn\rangle  + \wh V(0)\|\pn\|_4^4 -N^{-1}\|\nabla_1\eta\|^2\big) \big( \gamma \sigma +\text{h.c.}\big)  ,\\
		\end{split}\]
where $\cC_{\cG_N}$ is defined as in \eqref{eq:cCcGN}, where
	\[\begin{split}
		\cD_{\cG_N} &= \int dx\, \big[ b(\nabla_x \eta_x)\nabla_x d_{\eta,x} + \text{h.c.}\big] \\
		&\hspace{0.5cm}+  \frac{1}{2} \int dx dy \, N^3 (Vf_\ell)(N(x-y)) \pn(x) \pn(y) \\
	&\hspace{1cm}\times \big[ \big( b^*(\gamma_x) + b(\sigma_x) \big) d_{\eta,y}^*  + d_{\eta,x}^* \big( b^*(\gamma_y) + b(\sigma_y) \big) +\text{h.c.}\big] \\
		\end{split}\]
and where the error $\wt{\cE}_{\cG_N}$ satisfies the bound
		\[ \pm \wt{\cE}_{\cG_N} \leq CN^{-1/2} ( \cH_N + \cN^2+1)(\cN+1). \]
Notice that here we used $\pm( \cN - \int dx\, b_x^* b_x)\leq N^{-1}(\cN+1)^2 $ as well as
		\[\begin{split}  N^{-1} &  \int dxdy\, N^2V(N(x-y))| \langle\sigma_x, \gamma_y\rangle|^2  \\
		&= \int dxdy\, N^3(Vw_\ell^2)(N(x-y))\pn^2(x)\pn^2(y)  +\mathcal{O}(N^{-1}). \end{split}\]		

Let us now focus on $\cD_{\cG_N}$. We proceed here very similarly as in \cite[Section 7.5]{BBCS4}, so let us focus on the main steps. First of all, by the scattering equation \eqref{eq:scatlN}, we see that
		\[\begin{split}
		&\pm \Big(\int dx\, \big[ b(\nabla_x \eta_x)\nabla_x d_{\eta,x} + \text{h.c.}\big] \\
		&\hspace{1.5cm}  - \int dxdy\, N^3(-\Delta f_\ell)(N(x-y)) \pn(x)\pn(y) \big[b_y d_{\eta,x} +\text{h.c.}\big]\Big)\\
		&\hspace{0.5cm}\leq CN^{-1/2}( \cH_N+\cN^2+1)(\cN+1).
		\end{split}\]
We can therefore write 
			\begin{align*}
			\cD_{\cG_N} 
			&= \int dx dy \, N^3 \big[-(\Delta f_\ell)(N(x-y)) + \frac{1}{2} (Vf_\ell)(N(x-y))  \big] \pn(x) \pn(y) [b_y d_{\eta,x} + \text{h.c.}] \\
			& \hspace{0.5cm}+ \frac{1}{2} \int dx dy \, N^3 (Vf_\ell)(N(x-y))  \pn(x) \pn(y) \\
			&\hspace{1.5cm}\times \big[  b( \gamma_y) d_{\eta,x}  -b_yd_{\eta,x}   + b^*(\sigma_y)d_{\eta,x}   +\text{h.c.}\big] \\
			&+ \frac{1}{2} \int dx dy \, N^3 (Vf_\ell)(N(x-y))   \pn(x) \pn(y)\\
			&\hspace{1.5cm}\times\big[ d_{\eta,x}  \big(b(\gamma_y) + b^*(\sigma_y) \big) +\text{h.c.} \big] +\wt \cE_1 \\
			&=: \cD_1 + \cD_2 + \cD_3 + \wt\cE_1
			\end{align*}
for an error $\pm \wt \cE_1 \leq CN^{-1/2}( \cH_N+\cN^2+1)(\cN+1)$. The contributions $ \cD_1$ and $\cD_2$ are easily seen to be small, using once more the scattering equation \eqref{eq:scatlN} together with the operator bounds from Lemma \ref{lm:d-bds}. This yields that
		\[ \pm \cD_1 , \pm \cD_2 \leq CN^{-1} (\cN+1)^3,  \]
Similarly, the term $\cD_3$ can be reduced to
		\[\cD_3 = \frac{1}{2} \int dx dy \, N^3 (Vf_\ell)(N(x-y))   \pn(x) \pn(y) \big[ d_{\eta,x}  b^*(\eta_y)  +\text{h.c.} \big]  +\wt\cE_2,\]
up to another error that satisfies $\pm \wt \cE_2 \leq CN^{-1/2}( \cH_N+\cN^2+1)(\cN+1)$. To extract the main contribution to this term, we adapt the arguments from \cite{BBCS4}, using the identity 
		\[\begin{split}
		d_{\eta,x}	&= \int_0^1 ds \ d_{s\eta}^* (\eta_x) - \frac{1}{N} \int_0^1 ds \ e^{-s B} \mathcal{N} b^*(\eta_x) e^{s B} 	\\
		&\hspace{0.5cm}- \frac{1}{N} \int_0^1 ds \int du dy \ \eta(u;y) e^{-sB} b_u^* a_y^* a_x e^{sB}.
		\end{split}\]
Plugging this into $ \cD_3$, extracting the main contributions with the same tools as in the previous sections and collecting the previous bounds, one eventually finds that
		\[\begin{split} 
		\cD_{\cG_N}  =&\,  \frac12 \int dx dy \, N^3 (Vf_\ell w_\ell)(N(x-y))  \pn^2(x) \pn^2(y)   \\
		&\hspace{0.5cm} \times \int dv\, \big[b^*(\sigma_v) b^*(\gamma_v) + b(\sigma_v) b(\gamma_v) + 2 b^*(\sigma_v) b(\sigma_v)\big] \\
		&  +\frac{1}{2} \int dx dy \, N^3 ( Vw_\ell f_\ell)(N(x-y)) \pn^2(x) \pn^2(y)  \Vert \sigma \Vert^2 + \cE_{\cD }
		\end{split}\]
for an error $\pm \wt \cE_{\cD } \leq CN^{-1/2}( \cH_N+\cN^2+1)(\cN+1)$. Combining this term with the quadratic part in \eqref{eq:cGNid1}, we arrive at the identity \eqref{eq:propGN} and the error bound \eqref{eq:cEcGNbnd}. Notice that we use here Eq. \eqref{eq:scatlN} and the results from Lemma \ref{3.0.sceqlemma} which yields that
		\[\begin{split}
		 & \int dxdy\,N^3(Vw_\ell^2)(N(x-y))\pn^2(x)\pn^2(y)+ N^{-1}\|\nabla_1\eta\|^2 -   \langle \pn, (-\Delta  +V_\text{ext})\pn\rangle \\
		&   -\wh V(0)\|\pn\|_4^4  + \frac12\int dxdy\,N^3(Vf_\ell w_\ell)(N(x-y))\pn^2(x)\pn^2(y) \\
		& = -  \langle \pn, (-\Delta  +V_\text{ext}+ 8\pi\mathfrak{a}_0\pn^2 )\pn\rangle   +\mathcal{O}(N^{-1}) = -\eps_{GP} + \mathcal{O}(N^{-1}).
		\end{split}\]		
In particular, this proves part $b)$ of Proposition \ref{prop:GN}. 

To prove the statements of part $a)$, we first note that, by \cite[Theorem 1.1]{BSS}, we have the lower bound $ E_N \geq N\cE_{GP}(\pn)-C$, $E_N$ denoting the ground state energy of $H_N$. On the other hand, it is straightforward that the scattering equation \eqref{eq:scatlN} and Lemma \ref{3.0.sceqlemma} imply that
\begin{align*} 
\begin{split}
\kappa_{\cG_N} =&\, N\langle\pn , (-\Delta + V_\text{ext}+(N^3V(N.)\ast \pn) \pn)\pn\rangle + \| \nabla_1\eta\|^2\\
&- \frac12\int dxdy\, \big[ 2N^3(Vw_\ell)(N(x-y)) +  N^3(Vw^2_\ell)(N(x-y))\big]\pn^2(x)\pn^2(y)+ \mathcal{O}(1) \\
=&\, N\cE_{GP}(\pn) + \mathcal{O}(1).
\end{split}
\end{align*}
Hence, for $\Omega = (1,0,\dots,0)\in \cFpn$, we can use the trial state $ U_N^* e^{B} \Omega \in  L^2_s(\bR^{3N})$ to obtain the upper bound
		\[E_N \leq \langle \Omega, e^{-B} U_N H_N U_N^* e^{B}\rangle = \langle \Omega, \cG_N\Omega\rangle  \leq N\cE_{GP}(\pn)  +C.  \]
This proves that $ |E_N - N\cE_{GP}(\pn)|\leq C$ for some constant $C>0$, independent of $N$.

Finally, it is simple to show that for any $\delta>0$, we have that
		\[ \pm (\cG_N- \kappa_{\cG_N} -\cH_N)  \leq \delta \cH_N + C\delta^{-1}(\cN+1)  \]
for some $C>0$. This follows directly from the decomposition \eqref{eq:propGN}, Cauchy-Schwarz and by observing that 
		\[\begin{split} 
		\pm& \bigg( \int dxdy\, \Big[ \Gamma(x;y)  + \Delta\eta (x;y) - N^3(Vf_\ell)(N(x-y))\pn(x)\pn(y)\Big]b_xb_y+\text{h.c.}\bigg)\\
			&\leq C(\cK+\cN+1)
		\end{split}\]
and that, by the scattering equation \eqref{eq:scatlN}, we also have that
		\[\pm \bigg( \int dxdy\,  \Big[ -\Delta\eta (x;y) + N^3(Vf_\ell)(N(x-y))\pn(x)\pn(y)\Big]b_xb_y+\text{h.c.}\bigg) \leq C(\cK+\cN+1).\]		
The remaining contributions to $\cG_N$ are simple to control and we omit the details. Together with $ \kappa_{\cG_N} = E_N + \mathcal{O}(1)$, explained above, we thus conclude the proposition.


\section{Analysis of $\cJ_{N}$}  \label{sec:JN}
The goal of this section is to prove Proposition \ref{prop:JN}. To this end, we need to control the action of the operator $A$ defined in (\ref{eq:A-def}) on the different parts of the Hamiltonian $\cG_N$, which can be decomposed as in (\ref{eq:propGN}). 

A first observation is that conjugation with $e^A$ preserves, approximately, the number of excitations. The proof of the following Lemma is very similar to the proof of \cite[Lemma 2.6]{BSS} and will therefore be omitted.
\begin{lemma} \label{lm:growA}
	Let $0<6\exps \leq \exph$. Then, for every $j\in \NN$, there is $C=C_{j,\varepsilon }>0$ such that, for all $s\in [0;1]$, we have on  $\mathcal{F}_{\perp \pn}^{\leq N}$ the estimate 
	\begin{align*} 
	e^{-sA} (\mathcal{N}+1)^j e^{sA} \leq C (\mathcal{N}+1)^j.
	\end{align*}
\end{lemma}

Another important point for our analysis, is that we can often replace $A$ with the operator 
\begin{align*} 
\wt{A} = \frac{1}{\sqrt{N}} \int dx dy \, \wt{k}_H (x;y) b_x^* b_y^* \left[ b (\gamma_{L,x})  + b^* (\sigma_{L,x})  \right] - \text{h.c.} \end{align*}  
where the $\wt{b}$, $\wt{b}^*$ fields appearing in (\ref{eq:A-def}) have been replaced by $b, b^*$ operators, removing the projection $Q$ on the orthogonal complement of $\ph_0$. The difference between $A$ and 
$\wt{A}$ has the form 
\begin{equation}\label{eq:A-wtA}
A-\wt{A} = \frac{1}{\sqrt{N}} \int dx dy dz \ \left[ \rho_\sigma(x;y;z) b_x^* b_y^* b_z + \rho_\gamma(x;y;z) b_x^* b_y^* b_z^* - h.c. \right].
\end{equation} 
with $\rho_\sigma = \nu_\sigma - (Q \otimes Q \otimes Q) \nu_\sigma$, $\rho_\gamma =  \nu_\gamma - (Q \otimes Q \otimes Q) \nu_\gamma$ and $\nu_\sigma (x,y,z) = \wt{k}_H (x,y) \sigma_L (x,z)$, $\nu_\gamma (x,y,z) = \wt{k}_H (x,y) \gamma_L (x,z)$. With the bounds 
\begin{equation} \label{eq:rho}
	\begin{split}
&\| \rho_\gamma \|, \| \rho_\sigma \|, \sup_{x \in \bR^3} \| \rho_\gamma (x, . , .) \|, \sup_{x \in \bR^3} \| \rho_\gamma (., x , .) \|, \sup_{x \in \bR^3} \| \rho_\gamma (., . , x) \| \leq C, \\
&\sup_{x,z\in \mathbb{R}^3} \Vert \rho_\sigma(x,., z) \Vert, \sup_{y,z \in \mathbb{R}^3} \Vert \rho_\sigma(.,y, z) \Vert \leq C. \end{split} \end{equation}  
it is possible to control terms arising from the difference (\ref{eq:A-wtA}) (estimating these terms on $\cF^{\leq N}_{\perp \pn}$ is simple, because we consider states where all particles are orthogonal to $\pn$). We will omit most details for these contributions.


\subsection{Analysis of $e^{-A} [ \mathcal{Q}_{\cG_N} - \cK - \cV_\text{ext}] e^A$}

We start by conjugating the quadratic part of $\cG_N$, defined in (\ref{eq:cQcGN}), after subtracting the kinetic energy and the external potential (which will be considered in the next subsections). We will make use of the following lemma.  

\begin{lemma} \label{lemma:conjquadA}
Let $F : \bR^3 \times \bR^3 \to \mathbb{C}$, with $F(y,x) = \bar{F} (x,y)$. Then there exists $C>0$ such that
\begin{equation} \label{eq:F1-bd} \pm \int F(x,y) \left[ b_x^* b_y , A \right] dx dy \leq \frac{C}{\sqrt{N}}    \min \left\{ \| F \|_2 , \sup_x \int dy |F (x,y)| \right\}  (\cN + 1)^2. \end{equation} 
Moreover, 
\begin{equation} \label{eq:F2-bd} \pm \int F (x,y) \left[ b_x^* b_y^* + b_y b_x , A \right] dx dy \leq \frac{C \| F \|_2}{\sqrt{N}}  (\cN+1)^2. \end{equation}  
\end{lemma} 
\begin{proof}
A tedious but straightforward computation shows that, on $\cF_{\perp \ph_0}^{\leq N}$,  
\begin{equation}\label{eq:upsi-def}  \begin{split} 
\int dx dy \, &F (x,y) [ b_x^*  b_y , A ] \\ = \; &\frac{1}{\sqrt{N}} \int dx \big( b^* (\gamma_{L,x}) + b (\sigma_{L,x}) \big) b (\wt{k}_{H,x}) (1 - \cN / N) b (\overline{F}_x)  \\ &+  \frac{1}{\sqrt{N}} \int dx \big( b^* (\gamma_{L,x}) + b (\sigma_{L,x}) \big) b_x  (1 - \cN / N) b ( \overline{(F\wt{k}_H)}_x) 
\\ &-\frac{1}{\sqrt{N}} \int dx \, b_x^* b^* (\wt{k}_{H,x}) (1- \cN / N) b ( \overline{(\gamma_L F)}_x) \\ &+ \frac{1}{\sqrt{N}} \int dx \, (1- \cN/N) b_x b (\wt{k}_{H,x}) b (\overline{(\sigma_L F)}_x) 
\\ &- \frac{1}{N^{3/2}} \int dx dv  \big( b^* (\gamma_{L,x}) + b (\sigma_{L,x}) \big) \big( a^* (F_v) a_x b (\overline{\wt{k}_{H,x}}) + b_x a^* (F_v) a (\overline{\wt{k}_{H,x}}) \big) b_v \\ &+ \frac{1}{N^{3/2}} \int dx dv b_x^* b^* (\wt{k}_{H,x}) a^* (F_v) a (\overline{\gamma_{L,x}}) b_v \\ &- \frac{1}{N^{3/2}} \int dx dv \, a^* (F_v) a( \sigma_{L,x}) b_x b (\overline{\wt{k}_{H,x}}) b_v \\
&+\frac{1}{\sqrt{N}} \int dx dy \ F(x,y) [b_x^* b_y, A-\wt{A}] \\
=: \; & \sum_{j=1}^7 \Upsilon_j +\frac{1}{\sqrt{N}} \int dx dy \ F(x,y) [b_x^* b_y, A-\wt{A}]. \end{split} \end{equation} 
Here we interpret $F (x,y)$ as the kernel of a symmetric operator, denoted again with $F$, and we use for example the notation $F \wt{k}_H$ to indicate the product of the two operators $F$ and $\wt{k}_H$. We can bound 
\begin{equation}\label{eq:upsi1}  \begin{split} 
|\langle \xi , \Upsilon_1 \xi \rangle | \leq \; &\frac{1}{\sqrt{N}} \int dx \, \| b (\gamma_{L,x}) \xi \| \| b (\wt{k}_{H,x})  b (\overline{F}_x) \xi \| \\ &+ \frac{1}{\sqrt{N}} \int dx \| (\cN+1)^{3/4} \xi \| \| b (\sigma_{L,x}) b (\wt{k}_{H,x}) b (\overline{F}_x) (\cN+1)^{-3/4} \xi \| \\ \leq \; &\frac{1}{\sqrt{N}} \left[ \sup_x \| \wt{k}_{H,x} \| \right] \left[ \| \gamma_L \|_\text{op} + \| \sigma_L \|_\text{HS} \right] \| F \|_\text{op} \| \| (\cN+1)^{3/4} \xi \|^2. \end{split} \end{equation} 
By Lemma \ref{lm:bds-cubic} and because $\| F \|_\text{op} \leq \sup_x \int dy |F(x,y)|$ we conclude that $\Upsilon_1 \leq C N^{-1/2} (\cN+1)^{3/2}$. In (\ref{eq:upsi1}), we used the fact that  
\[ \int dx \, \| b (F_x) \xi \|^2 \leq \langle \xi , d\Gamma (F^2) \xi \rangle \leq \| F \|_\text{op}^2 \| \cN^{1/2} \xi \|^2 \]
(and similarly with $F_x$ replaced by $\gamma_{L,x}$). As for the second term on the r.h.s. of (\ref{eq:upsi-def}), we have 
\begin{align*}
\begin{split} 
|\langle \xi , \Upsilon_2 \, \xi \rangle | \leq \; &\frac{1}{\sqrt{N}} \int dx \, \left[ \| b (\gamma_{L,x}) \xi \| + \| \sigma_{L,x} \| \| (\cN+1)^{3/4} \xi \| \right] \| (F \wt{k}_H)_x \| \| b_x (\cN+1)^{1/4} \xi \|. \end{split} 
\end{align*} 
We can bound 
\[ \begin{split} \| (F \wt{k}_H)_x \|_2^2 = &\; \int \left| \int dy F(z,y) \wt{k}_H (y,x) dy \right|^2 dz \\ = &\; \int F (z,y_1) \wt{k}_H (y_1, x) F (y_2, z)  \wt{k}_H (x,y_2) dy_1 dy_2 dz  \\ 
\leq \; &\int |F (z,y_1)| |F (y_2, z)| |\wt{k}_H (y_1,x)|^2 dy_1 dy_2 dz  
\leq  \left[ \sup_z \int |F(y, z)| dy \right]^2 \| \wt{k}_{H, x} \|_2^2.   \end{split} \] 
Alternatively, we can estimate $\| (F\wt{k}_H )_x \| \leq \| F \|_2 \| \wt{k}_{H,x} \|$. Thus, with Lemma \ref{lm:bds-cubic}, 
\[   |\langle \xi , \Upsilon_2 \xi \rangle | \leq \frac{C}{\sqrt{N}}  \min \left\{ \| F \|_2 , \sup_x \int dy |F (x,y)| \right\}  
\| (\cN+1)^{3/4} \xi \|^2.  \]

The contribution of $\Upsilon_3$ and $\Upsilon_4$ can be bounded similarly as the one of $\Upsilon_1$ (using the fact that $\| \gamma_L F \|_\text{op} \leq \| \gamma_L \|_\text{op} \| F \|_\text{op} \leq C \sup_x \int dy \, |F(x,y)|$ and, analogously, $\| \sigma_L F \|_\text{op} \leq C \sup_x \int dy \, |F(x,y)|$). Let us consider $\Upsilon_5$. We first bring the term into normal order, then we have, with Lemma \ref{lm:bds-cubic},  
\[ \begin{split} 
|\langle \xi , \Upsilon_5 \xi \rangle | &\leq \frac{2}{N^{3/2}} \int \| a (F_v) (b (\gamma_{L,x}) + b^*(\sigma_{L,x})  \xi \| \| a_x b(\overline{\wt{k}_{H,x}}) b_v \xi \| dx dv \\&\quad+ \frac{1}{N^{3/2}} \int dx \ \Vert (b (\gamma_{L,x}) + b^*(\sigma_{L,x}) \xi \Vert \Vert b(\overline{\wt{k}_{H,x}}) b(F_x) \xi \Vert \\ &\leq \frac{3}{N^{3/2}} \| F \|_\text{op} \| \gamma_L \|_\text{op} \left[ \sup_x \| \wt{k}_{H,x} \| \right]  \| \cN \xi \| \| \cN^{3/2} \xi \| \\ &\leq \frac{C}{\sqrt{N}}  \left[ \sup_x \int dy \, |F(x,y)| \right] \, \| (\cN +1)^{3/4} \xi \|^2. \end{split} \]
The terms $\Upsilon_6, \Upsilon_7$ can be bounded similarly. As for $\Upsilon_8$, we proceed as explained after (\ref{eq:A-wtA}), using the bounds from (\ref{eq:rho}). This completes the proof of (\ref{eq:F1-bd}). 

To show (\ref{eq:F2-bd}), we compute the commutator 
\begin{equation}\label{eq:F2-1}  \int dx dy \, F (x,y) \left[ b_x^* b_y^*  + b_x b_y , A \right]. \end{equation} 
We find contributions similar to those appearing on the r.h.s. of (\ref{eq:upsi-def}), with some creation operators replaced by annihilation operators or vice versa. After normal ordering, these contributions can be bounded as 
we did above. Through normal ordering, however, commutators produce new terms. Let us consider an example. Similarly to $\Upsilon_1$ in (\ref{eq:upsi-def}), the operator (\ref{eq:F2-1}) contains the term  
\[ \frac{1}{\sqrt{N}} \int dx \, \big( b^* (\gamma_{L,x}) + b (\sigma_{L,x}) \big) b (\overline{\wt{k}_{H,x}}) (1 - \cN / N) b^* (F_x). \]
Normal ordering produces a new term having the form (we focus here on the largest contribution to the commutator $[ b(\overline{\wt{k}_{H,x}}) , b^* (F_x) ]$) 
\[ \wt{\Upsilon} = \frac{1}{\sqrt{N}} \int dx \, \big( b^* (\gamma_{L,x}) + b (\sigma_{L,x}) \big) \langle \overline{\wt{k}_{H,x}} , F_x \rangle = \frac{1}{\sqrt{N}} \int dx \, \big( b^* (\gamma_{L,x}) + b (\sigma_{L,x}) \big) (F \wt{k}_H) (x;x). \]
This contribution can be bounded, using Lemma \ref{lm:bds-cubic}, by 
\[ \begin{split}  |\langle \xi , \wt{\Upsilon} \xi \rangle | &\leq \frac{1}{\sqrt{N}} \int dx \| b (\gamma_{L,x}) \xi \| \| \xi \| |(F \wt{k}_H) (x;x)| \leq \frac{C \| F \|_2}{\sqrt{N}} \| \cN^{1/2} \xi \| \| \xi \|. \end{split} \]
(This term and similar contributions are the reason why (\ref{eq:F2-bd}) requires $F$ to be an Hilbert-Schmidt operator). Other contributions emerging from (\ref{eq:F2-1}) can be bounded analogously (or as we did above for the terms $\Upsilon_j$, $j=1,\dots ,7$). Again, the contribution arising from $A-\wt{A}$ can be controlled with the bounds in (\ref{eq:rho}). 
\end{proof} 

\begin{lemma} \label{lm:QA-conj} 
Let $\cQ_{\cG_N}$ denote the quadratic part of the excitation Hamiltonian $\cG_N$, defined in (\ref{eq:cQcGN}). 
Then, we have 
\[ e^{-A} \big[ \cQ_{\cG_N} - \cK - \cV_\text{ext} \big] e^A = \big[ \cQ_{\cG_N} - \cK - \cV_\text{ext} \big] + \cE_1 \]
where
\[ \pm \cE_1 \leq \frac{C}{\sqrt{N}} (\cN+1)^{3/2}.  \]
\end{lemma} 

\begin{proof} 
We have 
\begin{equation}\label{eq:eQe} \begin{split} 
e^{-A} \cQ_{\cG_N} e^A = \; &\cQ_{\cG_N} + \int_0^1 ds \, e^{-sA} \big[ \cQ_{\cG_N} , A \big] e^{sA} \\ = 
\; &\cQ_{\cG_N} + \int_0^1 ds \, \int dx dy  \, \Phi' (x;y) \, e^{-sA} \big[ b_x^* b_y , A \big] e^{sA} \\ &+ \frac{1}{2} \int_0^1 ds \, \int dx dy \,  \Gamma (x;y) \, e^{sA} \big[ b_x^* b_y^* + b_x b_y , A \big] e^{sA} \end{split} \end{equation} 
with $\Phi' = \Phi - \left[ -\Delta + V_\text{ext} \right]$ and $\Phi$ and $\Gamma$ as defined in (\ref{eq:Phi-def}), (\ref{eq:Gamma-def}). Let us consider first the diagonal term, proportional to $\Phi'$. The contribution from terms in (\ref{eq:Phi-def}) that do not contain $-\Delta$ or $V_\text{ext}$ can be handled with (\ref{eq:F1-bd}) (using, on the r.h.s., the quantity $\sup_x \int dy \, |F(x,y)|$). To deal with the remaining terms (after subtracting $-\Delta$ and $V_\text{ext}$) we observe that, by Lemma \ref{lm:bds-cubic}, the operators
\[ (\gamma - 1) (-\Delta + V_\text{ext}) \gamma , \quad (-\Delta + V_\text{ext}) (\gamma - 1) , \quad \sigma (-\Delta + V_\text{ext}) \sigma \]
all have a uniformly bounded Hilbert-Schmidt norm; thus we can apply again (\ref{eq:F1-bd}), this time using $\| F \|_2$ on the right hand side. We conclude that 
\[\begin{split}  \left|  \int dx dy  \, \Phi' (x;y) \, \langle \xi, e^{-sA}  \big[ b_x^* b_y , A \big] e^{sA} \xi \rangle \right| &\leq \frac{C}{\sqrt{N}} \| (\cN + 1)^{3/4} e^{-sA} \xi \|^2 \\ & \leq \frac{C}{\sqrt{N}} \| (\cN + 1)^{3/4} \xi \|. \end{split} \]
 where in the last step we used Lemma \ref{lm:growA}. 
 
Let us now consider the off-diagonal term on the r.h.s. of (\ref{eq:eQe}), proportional to $\Gamma$. Here, we use (\ref{eq:F2-bd}), combined with the bound \eqref{eq:HStGamma} and the definitions \eqref{eq:Gamma-def}, \eqref{eq:deftildeGamma} showing that 
\begin{align*}
\Vert \Gamma \Vert_\text{HS}
\leq \Vert \wt{\Gamma}\Vert_\text{HS} +  \Vert \sigma [N^3 V(N.)*\pn^2- 8\pi\frak{a}_0 \pn^2] \gamma \Vert_\text{HS}
\leq \Vert \wt{\Gamma}\Vert_\text{HS} + C \Vert \sigma \Vert_\text{HS}\leq C \,.
\end{align*}
We obtain
\[\begin{split}  \left|  \int dx dy  \, \Gamma (x;y) \, \langle \xi, e^{-sA}  \big[ b_x^* b^*_y + b_x b_y , A \big] e^{sA} \xi \rangle \right|  \leq \frac{C}{\sqrt{N}} \| (\cN + 1)^{3/4} \xi \|^2. \end{split} \]
Inserting in (\ref{eq:eQe}) and integrating over $s \in [0;1]$, we obtain the claim. 
\end{proof} 

\subsection{Analysis $e^{-A} \mathcal{V}_{ext} e^A$}

\begin{lemma}\label{lm:comm-extA}
Let $0<6\exps \leq \exph$, then we have
\begin{equation}\label{eq:comm-extA} \pm \big[ \cV_\text{ext} , A \big] \leq \frac{C}{\sqrt{N}} \left[ \cV_\text{ext} (\cN+1)^{1/2} + (\cN+1)^{3/2}  \right]. \end{equation} 
\end{lemma}

\begin{proof}
The same tedious computation leading to (\ref{eq:upsi-def}) gives \[  \begin{split} 
\big[ \cV_\text{ext} , A \big]  = \; &\frac{1}{\sqrt{N}} \int dx dy \, \big( V_\text{ext} (x) + V_\text{ext} (y) \big) \wt{k}_H(x;y) \big( b^* (\gamma_{L,x}) + b (\sigma_{L,x}) \big) b_y (1 - \cN / N) b_x  \\ &- \frac{1}{\sqrt{N}} \int dx dz V_\text{ext} (z) \gamma_L (x,z) b_x^* b^*  (\wt{k}_{H,x}) (1-\cN / N) b_z  \\
&+ \frac{1}{\sqrt{N}} \int dx dz V_\text{ext} (z) \sigma_L (x,z) b_x^* b^*  (\wt{k}_{H,x}) (1-\cN / N) b^*_z   \\
&- \frac{1}{N^{3/2}}  \int dx du V_\text{ext} (u)   \big( b^* (\gamma_{L,x}) + b (\sigma_{L,x}) \big) a^*_u a_x b (\wt{k}_{H,x}) a_u  \\
&- \frac{1}{N^{3/2}} \int dx du V_\text{ext} (u)   \big( b^* (\gamma_{L,x}) + b (\sigma_{L,x}) \big) b_x a^*_u  a (\wt{k}_{H,x}) a_u  \\
&+\frac{1}{N^{3/2}} \int dx du V_\text{ext} (u) b_x^* b^* (\wt{k}_{H,x}) a^*_u a (\gamma_{L,x}) b_u  \\
&-\frac{1}{N^{3/2}} \int dx du V_\text{ext} (u) a_u^* a (\sigma_{L,x}) a_x b (\wt{k}_{H,x}) b_u  + \text{h.c.} \\
&+ [\cV_\text{ext}, A-\wt{A}] \\
=: &\sum_{j=1}^7 \Delta_j + \text{h.c.} + [\cV_\text{ext}, A-\wt{A}] . \end{split} \]
To bound the contribution proportional to $V_\text{ext} (y)$ in $\Delta_1$, we use $\sup_y V_\text{ext} (y) \| \wt{k}_H (.,y) \| \leq C \sup_y V_\text{ext} (y) \pn (y) \leq C$ and we proceed similarly as in the bound (\ref{eq:upsi1}). As for the part of $\Delta_1$ proportional to $V_\text{ext} (x)$, we decompose $\gamma_{L,x} (y) = \check{g}_L (x-y) + (\gamma-1) *_2 \check{g}_L (x;y)$. Terms arising from $(\gamma-1) *_2 \check{g}_L$ and also from $\sigma_{L,x}$ can be handled as before, using the factor $\pn (x)$ to bound $V_\text{ext} (x)$. Thus, we only need to estimate 
\begin{align*}
		&\left\vert \int dx dy dz \ V_\text{ext}(x) \wt{k}_H(x,y) \check{g}_L(x-z) \langle \xi, b_z^* b_y b_x \xi\rangle  \right\vert \\
		&\leq C\int dx dz \ V_\text{ext}(x) \check{g}_L(x-z) \Vert a_x (\cN +1)^{1/4} \xi \Vert \Vert a_z (\cN +1)^{1/4} \xi \Vert \\
		&\leq C \left[ \int dx \, V_\text{ext} (x) \| a_x (\cN+1)^{1/4} \xi \|^2 \right]^{1/2} \left[ \int dz \, (V_\text{ext} *\check{g}_L) (z) \| a_z (\cN+1)^{1/4} \xi \|^2 \right]^{1/2} \\
		&\leq C \langle \xi, \big[ \cV_\text{ext} (\cN+1)^{1/2} + (\cN+1)^{3/2} \big] \xi \rangle 
		\end{align*}
using $(V_\text{ext}*\check{g}_L) \leq C(V_\text{ext}+1)$. We conclude that 
\[ \pm  (\Delta_1 + \text{h.c.}) \leq \frac{C}{\sqrt{N}} \left[ \cV_\text{ext} (\cN+1)^{1/2} + (\cN+1)^{3/2} \right]. \]   

To control $\Delta_2$, we set $\gamma'_{L,x} = (\gamma - 1)_x * \check{g}_L$. Since $\| V_\text{ext} \gamma'_L \|_\text{HS} \leq C$ by Lemma \ref{lm:bds-cubic}, we only need to estimate 
\[ \begin{split} \Big|  \frac{1}{\sqrt{N}}  \int V_\text{ext} (z) \check{g}_L &(x-z) \langle \xi, b_x^* b^* (\wt{k}_{H,x}) b_z \xi \rangle \Big| \\ &\leq \frac{C}{\sqrt{N}} \int V_\text{ext} (z) \check{g}_L (x-z) \| a_x (\cN+1)^{1/4}  \xi \| \| a_z (\cN+1)^{1/4} \xi \| \\ &\leq \frac{C}{\sqrt{N}} \| (\cV_\text{ext}+1)^{1/2} (\cN+1)^{1/4} \xi \|^2 \end{split} \]  
where, in the last step, we used the bound (\ref{eq:Vextconvg}). We conclude that 
\[ \pm \Delta_2  \leq \frac{C}{\sqrt{N}} \left[  \cV_\text{ext} (\cN+1)^{1/2} + (\cN+1)^{3/2} \right].  \]
The term $\Delta_3$ can be handled similarly (here, there is no need to subtract the identity as we did with $\gamma_{L,x}'$). As for $\Delta_4$, we use the inequalities $\sup_x \| \wt{k}_{H,x} \|_2 \leq C$ and $\sup_x \| \sigma_{L,x} \|_2 \leq C$ from Lemma \ref{lm:bds-cubic} to get 
\[ \begin{split}  |\langle \xi, \Delta_4 \xi \rangle | \leq \; &\frac{1}{N^{3/2}} \int dx du \, V_\text{ext} (u) \| b (\gamma_{L,x}) a_u \xi \| \| a_x a_u (\cN + 1)^{1/2} \xi \|  \\ &+ \frac{1}{N^{3/2}} \int dx du \, V_\text{ext} (u) \| a_u (\cN+1)^{1/2} \xi \| \| a_x a_u (\cN+1)^{1/2} \xi \|  \\ \leq \; &\frac{C}{N^{3/2}} \| \cV_\text{ext}^{1/2} (\cN+1)^{1/2} \xi \| \| \cV_\text{ext}^{1/2} (\cN+1) \xi \| \leq \frac{C}{\sqrt{N}} \| \cV_\text{ext}^{1/2} (\cN+1)^{1/4} \xi \|^2.  \end{split} \]
The terms $\Delta_5, \Delta_6, \Delta_7$ can be bounded similarly. The contribution arising from the difference $A- \wt{A}$ can be controlled proceeding as explained after (\ref{eq:A-wtA}), with (\ref{eq:rho}). 
\end{proof}

In order to use the estimate in Lemma \ref{lm:comm-extA}, we need to show that the r.h.s. of (\ref{eq:comm-extA}) remains small, when conjugated with $e^A$. 
\begin{lemma}\label{lm:groext-A} 
Let $0<6\exps \leq \exph$, then we have 
\begin{equation}\label{eq:groext-cub1} e^{-A} \cV_\text{ext} (\cN+1)^2 e^A \leq C \left[ \cV_\text{ext} (\cN+1)^2 + (\cN+1)^3  \right].
\end{equation} 
\end{lemma} 
\begin{proof} 
To show (\ref{eq:groext-cub1}) we compute, for an arbitrary $\xi \in \cF_{\perp \pn}^{\leq N}$, 
\begin{equation}\label{eq:ddsVextA} \begin{split} 
\frac{d}{ds} \langle \xi &, e^{-sA}  (\cN+1) \cV_\text{ext} (\cN+1) e^{sA} \xi \rangle \\ = \; &\langle \xi , e^{-sA}  (\cN+1) \big[ \cV_\text{ext} , A \big]  (\cN+1) e^{sA} \xi \rangle + 2 \text{Re } \langle \xi, e^{-sA} \big[\cN , A \big] \cV_\text{ext} (\cN+1) e^{sA} \xi \rangle.  \end{split} \end{equation} 
The first term on the r.h.s. can be bounded with (\ref{eq:comm-extA}). We find
\[  \langle \xi , e^{-sA}  (\cN+1) \big[ \cV_\text{ext} , A \big]  (\cN+1) e^{sA} \xi \rangle \leq C  \langle \xi , e^{-sA}  (\cN+1) \big( \cV_\text{ext} + \cN+1 \big)  (\cN+1) e^{sA} \xi \rangle. \]
As for the other terms, we compute 
\begin{equation}\label{eq:gron-ext1} \begin{split} 
\langle e^{-sA} \xi , &\big[ \cN , A \big]  \cV_\text{ext} (\cN+1) e^{sA} \xi \rangle  \\= \;& \frac{1}{\sqrt{N}} \int dx du \, V_\text{ext} (u) \langle \xi, b_x^* b^* (\wt{k}_{H,x}) \big(b (\gamma_{L,x}) + b^* (\sigma_{L,x}) \big) a_u^* a_u (\cN+1) \xi \rangle \\ &+  \frac{1}{\sqrt{N}} \int dx du \, V_\text{ext} (u) \langle \xi, \big(b^* (\gamma_{L,x}) + b (\sigma_{L,x}) \big)  b_x  b (\wt{k}_{H,x}) a_u^* a_u (\cN+1) \xi \rangle. \end{split} 
\end{equation} 
With $\sup_x \| \wt{k}_{H,x} \| < C$ and $\| \sigma_L \|_\text{HS} \leq C$, we can control 
\[ \begin{split}  \Big| \frac{1}{\sqrt{N}} \int dx du \, &V_\text{ext} (u) \langle \xi, b_x^* b^* (\wt{k}_{H,x}) b^* (\sigma_{L,x}) \big) a_u^* a_u (\cN+1) \xi \rangle \Big| \leq  C \| \cV_\text{ext}^{1/2} (\cN+1) \xi \|^2. \end{split} \] 
The contribution from the first term on the r.h.s. of (\ref{eq:gron-ext1}) proportional to $b (\gamma_{L,x})$, when arranged in normal order, produces a term of the form (we focus here on the largest term): 
\begin{equation}\label{eq:comm0-gro}
\frac{1}{\sqrt{N}} \int dx du \, V_\text{ext} (u) \,  \gamma_{L} (x;u) \langle \xi, b_x^* b^* (\wt{k}_{H,x})  a_u (\cN+1) \xi \rangle.
\end{equation} 
Since $\gamma'_L =  (\gamma - 1)*\check{g}_L$ is such that
\[\sup_x \int du \, (V_\text{ext} (x) + V_\text{ext} (u)) |\gamma'_L (x;u)| \leq C,
\]
it is enough to bound 
\[ \begin{split} 
\Big| \frac{1}{\sqrt{N}} \int dx du \, &V_\text{ext} (u) \,  \check{g}_L (x-u) \langle \xi, b_x^* b^* (\wt{k}_{H,x})  a_u (\cN+1) \xi \rangle \Big| \\  \leq \; &\frac{C}{\sqrt{N}} \int dx du \, V_\text{ext} (u) \,\check{g}_L (x-u)  \| b_x  (\cN+1)^{1/2}  \xi \| \| a_u (\cN+1) \xi \| \\ \leq \; &\frac{C}{\sqrt{N}} \| (\cV_\text{ext}+1)^{1/2} (\cN+1) \xi \|^2 \end{split} \]
where we used (\ref{eq:Vextconvg}). As for the second term  on the r.h.s. of (\ref{eq:gron-ext1}), we have to move $a_u^*$ to the left. This produces commutator terms, which can be handled similarly as we did with (\ref{eq:comm0-gro}). The claim now follows inserting these bounds in (\ref{eq:ddsVextA}), using Lemma \ref{lm:growA} and Gronwall's lemma.  
\end{proof} 
 
 Combining the results of the last two lemmas, we can control the action of $A$ on the external potential. 
 \begin{lemma}\label{lm:Vext-conj} 
Let $0<6\exps \leq \exph$, then we have 
 \[ e^{-A} \cV_\text{ext} e^A = \cV_\text{ext} + \cE_2 \]
 where
 \[ \pm \cE_2 \leq  \frac{C}{\sqrt{N}} \big[ \cV_\text{ext}  (\cN+1)^2 + (\cN + 1)^3 \big]. \]
 \end{lemma} 
 
 \begin{proof} 
 We can write 
 \[ \cE_2 = \int_0^1 ds e^{-sA} \left[ \cV_\text{ext} , A \right] e^{sA}. \]
 From Lemma \ref{lm:comm-extA}, we find
 \[ \pm \cE_2 \leq \frac{C}{\sqrt{N}} \int_0^1 ds e^{-sA} \left[ \cV_\text{ext} (\cN+1)^{1/2} + (\cN+1)^{3/2} \right] e^{sA}. \]
Estimating $\cV_\text{ext} (\cN+1)^{1/2} \leq \cV_\text{ext} (\cN+1)^2$, applying Lemma \ref{lm:growA}, Lemma \ref{lm:groext-A} and  integrating over $s \in [0;1]$, we obtain the claim. 
  \end{proof}

\subsection{Analysis of $e^{-A}(\mathcal{K}+ \mathcal{V}_N)e^A$}

\begin{lemma}\label{lm:KV-A}
Let $0<6\exps \leq \exph\leq 1/2$, then we have 
\begin{equation}\label{eq:KV-A1}  \pm\big[ \cK + \cV_N , A \big] \leq C \left( \cK + \cV_N + (\cN+1)^2 \right). \end{equation} 
Moreover, we can decompose 
\[ \big[ \cK + \cV_N , A \big] = \Theta_0 + \Theta_0^* + \cE \]
where 
\begin{equation}\label{eq:theta0-def} \Theta_0 = -\sqrt{N} \int dx dy N^2 V (N (x-y)) \ph_0 (y) b_x^* b_y^* \big[ b (\gamma_{L,x}) + b^* (\sigma_{L,x}) \big] \end{equation} 
and where
\begin{equation}\label{eq:defTheta0} \pm \cE \leq C N^{-1/2+\eps} ( \cV_N + \cK + \cN +1) (\cN+1).  \end{equation} 
\end{lemma} 
 
\begin{proof} 
With a long but straightforward computation, we find 
\begin{equation}\label{eq:KA-comm} \begin{split} 
\big[ \cK , A \big] = \; &\frac{1}{\sqrt{N}} \int dx dy \,  \nabla_x \wt{k}_H (x,y) \nabla_x b_x^* b_y^* \big[ b (\gamma_{L,x}) + b^* (\sigma_{L,x}) \big]   \\
&+ \frac{1}{\sqrt{N}} \int dx dy \; \wt{k}_H (x;y) \nabla_x b_x^* b_y^* \big[ b (\nabla_x \gamma_{L,x}) + b^* (\nabla_x \sigma_{L,x}) \big]  \\
&+ \frac{1}{\sqrt{N}} \int dx dy \, \nabla_y \wt{k}_H (x,y)  \, b_x^* \nabla_y b_y^* \big[ b (\gamma_{L,x}) + b^* (\sigma_{L,x}) \big]  \\
&+ \frac{1}{\sqrt{N}} \int dx dy du \;  \wt{k}_H (x;y) \nabla \ph_0 (u) \ph_0 (x) \nabla_u a_u^* b_y^* \big[ b(\gamma_{L,x}) + b^* (\sigma_{L,x}) \big]   \\
&+ \frac{1}{\sqrt{N}} \int dx dy du \;  \wt{k}_H (x;y) \nabla \ph_0 (u) \ph_0 (y) a_x^* \nabla_u b_u^* \big[ b(\gamma_{L,x}) + b^* (\sigma_{L,x}) \big]   \\
&+ \frac{1}{\sqrt{N}} \int dx dy dz \;  \wt{k}_H (x;y) b_x^* b_y^* \big[ b(\Delta \gamma_{L,x}) + b^* (\Delta \sigma_{L,x}) \big] + \text{h.c.}  \\
&+ \big[ \cK, A-\wt{A} \big].
\end{split} \end{equation} 
Let us consider first the contribution proportional to $b (\nabla_x \gamma_{L,x})$ in the second term on the r.h.s.
We write $\gamma_{L,x} (z) = \check{g}_L (x-z) + \gamma'_{L,x} (z)$. Using that $\sup_x \| \wt{k}_{H,x} \|  \leq C$, we find  
\begin{equation}\label{eq:sec-KA} \Big| \frac{1}{\sqrt{N}} \int dx dy \, \wt{k}_H (x;y) \check{g}_L (x-z) \langle \xi, \nabla_x b_x^* b_y^* \nabla_z b_z \xi \rangle \Big| \leq \frac{C}{\sqrt{N}} \| \cK^{1/2} \cN^{1/2} \xi \| \| \cK^{1/2} \xi \|. \end{equation} 
The contribution proportional to $\gamma'_L$ can be estimated by 
\[ \Big| \frac{1}{\sqrt{N}} \int dx dy \, \wt{k}_H (x;y) \langle \xi, \nabla_x b_x^* b_y^* b (\nabla_x \gamma'_{L,x}) \xi \rangle \Big| \leq \frac{C}{\sqrt{N}} \| \cK^{1/2} \cN^{1/2} \xi \| \| \cN^{1/2} \xi \|, \]
because $ \| \nabla_1 \gamma'_{L} \| \leq C$ by \eqref{eq:bds-cubic4}. As for the part of the second term on the r.h.s. of (\ref{eq:KA-comm}) proportional to $\nabla_x \sigma_{L,x}$ we use that, by Lemma \ref{lm:bds-cubic}, $\sup_x \Vert \nabla_x \sigma_{L,x} \Vert \leq CN^{\exps/2}$. Hence  
\[ \begin{split}  \Big| \frac{1}{\sqrt{N}} \int dx dy \; \wt{k}_H (x;y)  \langle \xi , &\nabla_x b_x^* b_y^*  b^* (\nabla_x \sigma_{L,x}) \xi \rangle \Big| \\ &\leq  CN^{-1/2} \| \cK^{1/2} (\cN+1)^{1/2} \xi \| \| (\cN+1)^{1/2} \xi \|. \end{split} \]
To control the fourth and fifth term on the r.h.s. of (\ref{eq:KA-comm}), we proceed similarly. 

As for the sixth term on the  r.h.s. of (\ref{eq:KA-comm}), we write again $\gamma_{L,x} (z) = \check{g}_L (x-z) + \gamma'_{L,x} (z)$. The contribution proportional to $\check{g}_L$ can be decomposed after integration by parts into a term of the form (\ref{eq:sec-KA}) and the term 
\begin{equation} \label{eq:last-KA} \begin{split} \Big| \frac{1}{\sqrt{N}}& \int  dx dy dz  \, \nabla_x \wt{k}_H (x;y)  \check{g}_L (x-z) \langle \xi, b_x^* b_y^* \nabla_z b_z \xi \rangle \Big| \\ 
\leq \; &\frac{1}{\sqrt{N}} \int dx dy dz \, |\wt{k}_H (x;y)| \check{g}_L (x-z) \| b_x \nabla_y b_y \xi \| \Vert \nabla_z b_z \xi \Vert \\
&+ \frac{1}{\sqrt{N}} \int dx dy dz \, |(Nw_\ell(N.)*\widecheck{\chi}_H)(x-y) \nabla \pn(y)| \check{g}_L (x-z) \\&\qquad \qquad \times \| b_x  b_y (\cN+1)^{-1/2} \xi \| \Vert \nabla_z b_z (\cN+1)^{1/2} \xi \Vert \\
\leq \; &\frac{C}{\sqrt{N}} \| (\cK+\cN)^{1/2} \xi \| \| \cK^{1/2} (\cN+1)^{1/2} \xi \| \end{split} \end{equation} 
where we used $\nabla_x \wt{k}_H(x,y)= - \nabla_y \wt{k}_H(x,y) - (Nw_\ell(N.)*\widecheck{\chi}_H)(x-y) \nabla \pn(y)$ and integration by parts. The contribution arising from the difference $A - \wt{A}$ can be controlled with the strategy outlined after (\ref{eq:A-wtA}), with the bounds (\ref{eq:rho}). We obtain 
\begin{align*}
\pm [\cK, A -\wt{A}] \leq \frac{C}{\sqrt{N}} \cN^{3/2}.
\end{align*}
We are left with the first and the third terms on the r.h.s. of (\ref{eq:KA-comm}). Here, we replace $\wt{k}_H$ with 
$\wt{k} (x;y) = -N w_\ell (N (x-y)) \ph_0 (y)$, using that, by , by (\ref{eq:diffk}), $\sup_x \| \nabla_x (\wt{k} - \wt{k}_{H})_x \| \leq C N^{\exph}$. Integrating by parts, we conclude that
\begin{equation} \label{eq:KA-fin} \begin{split} &\big[ \cK ,  A \big]  \\&= \sqrt{N} \int dx dy \big[ (\Delta_x + \Delta_y)  w_\ell (N(x-y)) \big]  \ph_0 (y) \,  b_x^* b_y^* \big[ b (\gamma_{H,x}) + b^* (\sigma_{H,}) \big]+ \text{h.c.} + \cE_1  \\ 
&=  \frac{2}{\sqrt{N}}  \int dx dy \, N^3 (\Delta w_\ell) (N(x-y)) \, \ph_0 (y) \,  b_x^* b_y^* \big[ b (\gamma_{H,x}) + b^* (\sigma_{H,}) \big]+ \text{h.c.} + \cE_1
\end{split} \end{equation} 
where
\begin{align*}
| \langle \xi, \cE_1 \xi \rangle | &\leq CN^{\exph-1/2} \| \cK^{1/2}  \xi \| \| (\cN+1) \xi \| + CN^{-1/2} \Vert \cK^{1/2} (\cN +1)^{1/2} \xi \| \Vert (\cK+ \cN)^{1/2} \xi \Vert \\
&\quad + CN^{-1/2} \Vert (\cN+1)^{3/4} \xi \Vert^2. 
\end{align*}
Here, we used the fact that contributions where one (or two) derivative hits $\ph_0$ produce error terms that can be bounded similarly as in (\ref{eq:last-KA}).

We consider next the commutator with the potential energy operator $\cV_N$. We find 
\begin{equation} \label{eq:VA-deco} \begin{split} 
\big[ \cV_N , A \big] = \; &\frac{1}{\sqrt{N}} \int dx dy \, N^2 V(N(x-y)) \wt{k}_H (x;y)  b_x^* b_y^* \big[ b(\gamma_{L,x}) + b^* (\sigma_{L,x}) \big]  \\  
&+ \frac{1}{\sqrt{N}} \int dx dy dz du \, N^2 \big[ V (N(x-u)) + V (N(y-u)) - V (N (z-u)) \big] \\ 
&\hspace{7cm} \times \wt{k}_H (x;y)  \gamma_L (x;z)  b_x^* b_y^* a_u^* b_z a_u  \\ 
&+  \frac{1}{\sqrt{N}} \int dx dy dz du \, N^2 \big[ V (N(x-u)) + V (N(y-u)) - V (N (z-u)) \big] \\ 
&\hspace{7cm} \times \wt{k}_H (x;y)  \sigma_L (x;z)  b_x^* b_y^* a_u^* b^*_z a_u  \\
&+\frac{1}{\sqrt{N}} \int dx dy dz \, N^2 \big[ V(N(x-z)) + V (N (y-z) \big] \wt{k}_H (x;y) \sigma_L (x;z) b_x^* b_y^* b_z^*  \\ 
&+ \text{h.c.} + [\cV_N, A-\tilde{A}].
\end{split} \end{equation} 
To control the  second term on the r.h.s. of (\ref{eq:VA-deco}), we write $\gamma_L (x;z) = \check{g} (x-z) + \gamma'_L (x;z)$. Focusing for example on the contribution proportional to $V(N(x-u))$ and to $\check{g}_L$, we can proceed as follows:
\[ \begin{split} \Big| \frac{1}{\sqrt{N}} &\int dx dy dz du \, N^2 V (N(x-u)) \wt{k}_H (x;y) \check{g}_L (x-z) \langle \xi, b_x^* b_y^* a_u^* b_z a_u \xi \rangle \Big|  \\ &\leq \frac{1}{\sqrt{N}} \int dx dy du dz \, N^2 V (N(x-u)) |\wt{k}_H (x;y)| \, \check{g}_L (x-z) \,  \| b_x b_y a_u  \xi \|  \| b_z a_u \xi \| \\ 
&\leq \frac{1}{\sqrt{N}} \Big[ \int dx dy du dz \, N^2 V (N(x-u)) \, \check{g}^2_L (x-z) \,  
\| b_x b_y a_u  \xi \|^2 \Big]^{1/2} \\ &\hspace{3cm} \times \Big[ \int dx dy du dz \, N^2 V (N(x-u)) |\wt{k}_H (x;y)|^2 \,  \| b_z a_u  \xi \|^2 \Big]^{1/2} \\
&\leq  CN^{(3\exps-\exph)/2-1} \| \cV_N^{1/2} \cN^{1/2} \xi \|  \| \cN \xi \| 
\leq CN^{-1} \| \cV_N^{1/2} \cN^{1/2} \xi \|  \| \cN \xi \| 
\end{split} \]
where we used the fact that $\sup_x \| \wt{k}_{H,x} \| \leq CN^{-\exph/2}$, $\| \check{g}_L \| \leq C N^{3\exps/2}$ shown in Lemma \ref{lm:bds-cubic}. The other contributions to the second term on the r.h.s. of  (\ref{eq:VA-deco}) can be treated similarly. Also the third term on the r.h.s. of (\ref{eq:VA-deco}) can be treated analogously. As for the fourth term on the r.h.s. of (\ref{eq:VA-deco}), we can proceed as follows (focussing for example on the contribution proportional to $V(N(x-z))$): 
\[ \begin{split} \Big| \frac{1}{\sqrt{N}} \int dx dy dz \, &N^2 V(N(x-z)) \wt{k}_H (x;y) \sigma_L (x;z) \langle \xi, b_x^* b_y^* b_z^* \xi \rangle \Big| \\ \leq \; &\frac{1}{\sqrt{N}} \int dx dy dz \, N^2 V(N(x-z)) |\wt{k}_H (x;y)| |\sigma_L (x;z)| \\ &\hspace{5cm} \times \| b_x b_y b_z (\cN+1)^{-1/2} \xi \| \| (\cN + 1)^{1/2} \xi \|  \\ \leq \; & N^{(3\exps-\exph)/2-1}  \| \cV_N^{1/2} \xi \| \| (\cN + 1)^{1/2} \xi \|
\leq CN^{-1} \| \cV_N^{1/2} \xi \| \| (\cN + 1)^{1/2} \xi \|. \end{split} \]
The contribution arising from the difference $A- \wt{A}$ can again be controlled using (\ref{eq:rho}). We find 
\begin{align*}
\vert \langle \xi, [\cV_N, A -\wt{A}] \xi \rangle \vert \leq CN^{-1} \Vert \cV_N (\cN+1)^{1/2} \xi \Vert \Vert (\cN +1) \xi \Vert. 
\end{align*}
We are left with the first term on the r.h.s. of (\ref{eq:VA-deco}). Here we replace $\wt{k}_H$ with 
$\wt{k} (x;y) = -N w_\ell (N (x-y)) \ph_0 (y)$, using that, by (\ref{eq:diffk}), $| (\wt{k} - \wt{k}_H) (x;y)| \leq C N^{\exph} \ph_0 (y)$. We conclude that 
\[ \begin{split} &\big[ \cV_N , A \big] \\ &= -\sqrt{N} \int dx dy \, N^2 V(N(x-y))  w_\ell (N (x-y))  \ph_0 (y) \, b_x^* b_y^* \big[ b(\gamma_{L,x}) + b^* (\sigma_{L,x}) \big]+ \text{h.c.} + \cE_2  \end{split} \]
where
\begin{equation}\label{eq:VA-fin} |\langle \xi , \cE_2 \xi \rangle | \leq  C N^{-1}  \| \cV_N^{1/2} (\cN+1)^{1/2} \xi \| \| (\cN + 1) \xi \|  + CN^{\exph- 1} \| \cV_N^{1/2} \xi \| \| (\cN+1)^{1/2} \xi \|. \end{equation} 
Combining (\ref{eq:KA-fin}) with (\ref{eq:VA-fin}) and using the scattering equation (\ref{eq:scatl}), we obtain 
\[ \big[ \cK + \cV_N , A \big] =  -\sqrt{N} \int dx dy \, N^2 V(N(x-y))  \ph_0 (y) \, b_x^* b_y^* \big[ b(\gamma_{L,x}) + b^* (\sigma_{L,x}) \big]+ \text{h.c.}  + \cE_3  \]
where the error term $\cE_3$ takes into account the error terms $\cE_1, \cE_2$ and also the contributions emerging from the r.h.s. of (\ref{eq:scatl}). From (\ref{eq:KA-fin}), (\ref{eq:VA-fin}), we obtain that
\begin{align*}
\vert \langle \xi, \cE_3 \xi \rangle \vert &\leq C N^{\exph -1/2}\Vert \cK^{1/2}\xi \Vert \Vert (\cN +1)\xi \Vert  + CN^{-1/2} \Vert \cK^{1/2} (\cN+1)^{1/2} \xi \Vert \Vert (\cK+\cN)^{1/2}  \xi \Vert \\
&\quad + CN^{-1} \Vert \cV_N^{1/2} (\cN +1)^{1/2} \xi \Vert \Vert (\cN+1) \xi \Vert+ CN^{\exph-1} \Vert\cV_N^{1/2} \xi \Vert \Vert (\cN+1)^{1/2} \xi \Vert \\
&\quad + CN^{-1/2} \Vert (\cN+1)^{3/4} \xi \Vert^2. 
\end{align*} 
Hence, using $\exph\leq 1/2$, we obtain as well
\[ \pm \cE_3 \leq C \left[ \cV_N + \cK + (\cN+1)^2 \right]. \]
To conclude the proof of the lemma, we just have to observe that 
\[ \pm \sqrt{N} \left( \int dx dy \, N^2 V(N(x-y)) \ph_0 (y) \, b_x^* b_y^* \big[ b(\gamma_{L,x}) + b^* (\sigma_{L,x}) \big] +h.c.\right) \leq C (\cV_N + \cN+1). \] 
\end{proof}

As an application of (\ref{eq:KV-A1}), we establish a bound on the growth of $\cK+\cV_N$ with respect to the action of $A$.
\begin{lemma} \label{lm:gronA-KV} 
Let $0<6\exps\leq \exph\leq 1/2$. For $k \in \{ 0 ,2 \}$, there is a constant $C > 0$ such that, for all $s\in [0;1]$ 
\begin{equation}\label{eq:gronA-KV}  e^{-sA} (\cK + \cV_N) (\cN+1)^k e^{sA} \leq C  (\cK + \cV_N) (\cN+1)^k + (\cN+1)^{k+2}. \end{equation} 
\end{lemma}

\begin{proof} 
For a fixed $\xi \in \cF_{\perp \ph_0}^{\leq N}$ and for $s \in [0;1]$, let 
\[ \ph_\xi (s) = \langle \xi, e^{-sA} (\cK + \cV_N) (\cN+1)^k e^{sA} \xi \rangle = \langle \xi, e^{-sA} (\cN+1)^{k/2}  (\cK + \cV_N) (\cN+1)^{k/2} e^{sA} \xi \rangle. \]
We compute 
\begin{equation}\label{eq:der-phKV} \begin{split} \frac{d}{ds} \ph_\xi (s) = \; &\langle \xi,  e^{-sA} (\cN+1)^{k/2} \big[ \cK + \cV_N , A \big] (\cN+1)^{k/2} e^{sA} \xi \rangle \\ &+ \Big\{ \langle \xi , e^{-sA}\big[ (\cN+1)^{k/2} , A \big]  (\cK + \cV_N) (\cN+1)^{k/2} e^{sA} \xi \rangle  \xi \rangle  + \text{hc.} \Big\}. \end{split} \end{equation} 
With Lemma \ref{lm:KV-A}, the term on the first line is bounded by 
\[  \begin{split}  \langle \xi,  e^{-sA} (\cN+1)^{k/2} \big[ \cK + \cV_N , A \big] (\cN+1)^{k/2} e^{sA} \xi \rangle &\leq  \ph_\xi (s) + \langle \xi, e^{-sA} (\cN+1)^{k+2} e^{sA} \xi \rangle \\ &\leq  \ph_\xi (s) + C \langle \xi, (\cN+1)^{k+2} \xi \rangle \end{split} \]
where in the last step we applied Lemma \ref{lm:growA}. If $k=0$, (\ref{eq:gronA-KV}) follows by Gronwall's Lemma. If $k=2$, we still have to handle the terms on the second line of (\ref{eq:der-phKV}). To this end, we observe that, writing $\wt{A} = A_\gamma + A_\sigma - A^*_\gamma - A^*_\sigma$, with 
\[ A_\gamma = \frac{1}{\sqrt{N}} \int dx dy \, \wt{k}_H (x;y) b_x^* b_y^* b (\gamma_{L,x}), \qquad A_\sigma =   \frac{1}{\sqrt{N}} \int dx dy \, \wt{k}_H (x;y) b_x^* b_y^* b^* (\sigma_{L,x}), \]
we find, as $\cN$ preserves $\cF_{\perp \pn}^{\leq N}$,
\[ \big[ \cN , A \big] = \big[ \cN , \wt{A} \big] = A_\gamma + 3A_\sigma + 3 A_\sigma^* + A_\gamma^* \]
and thus 
\[ \big[ \cN , A \big] (\cK+\cV_N) = \left\{ A_\gamma (\cK+\cV_N) + 3A_\sigma (\cK + \cV_N) + \text{hc.} \right\} + 3 \big[ A_\sigma^* , (\cK+ \cV_N) \big] + \big[ A_\gamma^* , (\cK+ \cV_N)]. \]
To control these terms, we can proceed similarly as in the proof of Lemma \ref{lm:KV-A}. We skip the details.  
\end{proof} 

Since the term $\Theta_0$ defined in (\ref{eq:theta0-def}), emerging from the commutator of $\cK+\cV_N$ with $A$, is not small, we need to conjugate it again with $e^{A}$. To this end, we will use the following lemma. 
\begin{lemma} \label{lm:theta0}
We assume $0<3\exps\leq \exph$. Let $\Theta_0$ be defined as in (\ref{eq:theta0-def}). Let
\begin{equation}\label{eq:Pi0-def} \begin{split} \Pi_0 = \; &- \int dx dy N^2 V (N(x-y)) \ph_0 (y) \wt{k}_H (x;y) \big( b^* (\gamma_{L,x}) + b (\sigma_{L,x}) \big) \big( b (\gamma_{L,x}) + b^* (\sigma_{L,x}) \big) \\  &- \int dx dy N^2 V (N(x-y)) \ph_0 (x) \wt{k}_H (x;y) \big( b^* (\gamma_{L,x}) + b (\sigma_{L,x}) \big) \big( b (\gamma_{L,y}) + b^* (\sigma_{L,y}) \big). \end{split} \end{equation} 
Then we have
\begin{equation}\label{eq:ThA} 
\big[ \Theta_0 + \Theta_0^* , A \big] = \Pi_0 + \Pi_0^* + \cE \end{equation} 
where 
\[ \pm \cE \leq \frac{C}{\sqrt{N}} \big[ \cV_N + (\cN+1)^2 \big]. \]
Moreover, 
\begin{equation}\label{eq:PiA}
 \pm \big[ \Pi_0 + \Pi_0^* , A \big] \leq C N^{-1/2} (\cN+1)^2. \end{equation}
\end{lemma} 

\begin{proof} 
We compute 
\begin{equation} \label{eq:CA-00} 
\begin{split} 
\big[ \Theta_0 &+ \Theta_0^* , A \big]  \\ = \; &\frac{1}{\sqrt{N}} \int du dv N^3 V (N(u-v)) \ph_0 (v) \big[ b_u^* , A \big] b_v^* \big(b(\gamma_{L ,u}) + b^* (\sigma_{L ,u}) \big) + \text{h.c.} \\  &+\frac{1}{\sqrt{N}} \int du dv N^3 V (N(u-v)) \ph_0 (v) b_u^* \big[ b_v^* , A \big] \big(b(\gamma_{L ,u}) + b^* (\sigma_{L ,u}) \big) + \text{h.c.} \\  &+\frac{1}{\sqrt{N}} \int du dv N^3 V (N(u-v)) \ph_0 (v)  b_u^*  b_v^*  \big[ \big(b(\gamma_{L ,u}) + b^* (\sigma_{L ,u}) \big) , A \big] + \text{h.c.} \\ 
&+ \big[ \Theta_0 + \Theta_0^* , A-\wt{A} \big] \\
= \; & \text{I} + \text{II} + \text{III}  + \big[ \Theta_0 + \Theta_0^* , A-\wt{A} \big]. 
\end{split} 
\end{equation} 
Let us first consider the operator $\text{I}$. With the expression (\ref{eq:A-def}) for $A$, the commutator $[b_u^*, A]$ produces several terms.
From (\ref{eq:comm-b}), we can also predict that corrections to the canonical commutation relations carry an additional $N^{-1}$ factor and will therefore lead to small errors. An example of such a term contributing to $\text{I}$ is
\[ \begin{split} \Big| \frac{1}{N^2} \int &du dv dx dy N^3 V(N(u-v)) \ph_0 (v) \wt{k}_H (x;y) \langle \xi , b^* (\gamma_{L,x}) a_u^* b_v^* a_x b_y b(\gamma_{L,u}) \xi \rangle \Big| \\ &\leq 
 \frac{1}{N^2} \int du dv dx dy N^3 V(N(u-v)) \ph_0 (v) |\wt{k}_H (x;y)| \| a_u a_v a (\gamma_{L,x}) \xi \| \| a_x a_y a (\gamma_{L,u}) \xi \|  \\ &\leq 
 \frac{C}{N^{3/2}} \| \cV_N^{1/2} (\cN+1)^{1/2} \xi \|  \| (\cN+1)^{3/2} \xi \|.  \end{split} \]
 Other contributions of order 6 in creation and annihilation operators can be bounded similarly. Let us now focus on terms of order 4, which appear from the term $\text{I}$ if we use canonical commutation relations for the $b$-fields. The main contribution arises from the part of $A$ containing 2 or 3 annihilation operators. It has the form (from $[b_u^*, b_x] \simeq \delta (u-x)$) 
 \begin{equation}\label{eq:CA-0} \begin{split} \int &dv dx dy N^2 V (N(x-v)) \ph_0 (v) \wt{k}_H (x;y) \big( b^* (\gamma_{L,x}) + b (\sigma_{L,x}) \big) b_y b_v^* \big(b(\gamma_{L,x}) + b^* (\sigma_{L,x}) \big) \\ \simeq \; & \int dv dx dy N^2 V (N(x-v)) \ph_0 (v) \wt{k}_H (x;y) \\ &\hspace{3cm} \times \big( b^* (\gamma_{L,x}) + b (\sigma_{L,x}) \big) b_v^* b_y  \big(b(\gamma_{L,x}) + b^* (\sigma_{L,x})\big) \\ &+  \int dx dy N^2 V (N(x-y)) \ph_0 (y) \wt{k}_H (x;y) \big( b^* (\gamma_{L,x}) + b (\sigma_{L,x}) \big)  \big(b(\gamma_{L,x}) + b^* (\sigma_{L,x})\big) \end{split}  \end{equation} 
 where again we used the fact that corrections to the canonical commutation relations produce only small errors. 
 The first term on the r.h.s. of the last equation is small, on states with few excitations. In fact, we can bound (considering for example the contribution proportional to $b^* (\gamma_{L,x})$ and to $b(\gamma_{L,x})$)  \begin{equation}\label{eq:CA-1} 
 \begin{split} \Big|  \int dv & \,dx dy \, N^2 V (N(x-v)) \ph_0 (v) \wt{k}_H (x;y)  \langle \xi, b^* (\gamma_{L,x})  b_v^* b_y  b(\gamma_{L,x}) \xi \rangle \Big| \\  \leq \; &\int dv dx dy dz \, N^2 V (N(x-v)) \ph_0 (v) |\wt{k}_H (x;y)| \check{g}_L (x-z) \| a_z a_v \xi \|  \|  b_y  b (\gamma_{L,x}) \xi \| \\ &+   \int dv dx dy dz \, N^2 V (N(x-v)) \ph_0 (v) |\wt{k}_H (x;y)|  \| a (\gamma'_{L,x}) a_v \xi \|  \|  b_y  b (\gamma_{L,x}) \xi \| \end{split} \end{equation} 
 where we used the notation $\gamma'_{L} = (\gamma_{L} -1) *_2 \hat{g}_L$. The first term on the r.h.s. of (\ref{eq:CA-1}) can be estimated (using that $\sup_x \| \wt{k}_{H,x} \| \leq C <\infty$) by 
 \[ \begin{split} 
 \int dv dx dy & dz \, N^2 V (N(x-v)) \ph_0 (v) |\wt{k}_H (x;y)| \check{g}_L (x-z) \| a_z a_v \xi \|  \|  b_y  b (\gamma_{L,x}) \xi \| \\ &\leq \left[ \int dv dx dy dz \, N^2 V (N(x-v)) |\wt{k}_H (x;y)|^2 \| a_z a_v \xi \|^2 \right]^{1/2} \\ &\hspace{1.5cm} \times \left[  \int dv dx dy dz \, N^2 V (N(x-v))  \check{g}^2_L (x-z)  \|  a_y  a (\gamma_{L,x}) \xi \|^2 \right]^{1/2}\\ & \leq 
 C N^{-1+(3\tau -\exph)/2} \| (\cN+1) \xi \|^2. \end{split} \] 
 The second term on the r.h.s. of (\ref{eq:CA-1}) can be bounded using that $\sup_x \| \gamma'_{L,x} \| < \infty$. There are many more contributions to $\text{I}$ that are quartic in creation and annihilation operators (some of them appear in the first term on the r.h.s. of (\ref{eq:CA-0})); they can all be controlled analogously (in fact, the term (\ref{eq:CA-0}) is the only one producing a non-negligible quadratic contribution; the reason is that the two commutators leading to the second term on the r.h.s. of (\ref{eq:CA-0}) contract the variables in the kernel $\wt{k}_H$ with the variables in the interaction potential). We conclude that 
 \[ \text{I} = \int dx dy N^2 V (N(x-v)) \ph_0 (y) \wt{k}_H (x;y) \big( b^* (\gamma_{L,x}) + b (\sigma_{L,x}) \big)  (b(\gamma_{L,x}) + b^* (\sigma_{L,x})) + \cE_1 \]
 where
 \[ \pm \cE_1 \leq \frac{C}{\sqrt{N}} \big[ \cV_N + (\cN+1)^2 \big]. \]
 The term $\text{II}$ on the r.h.s. of (\ref{eq:CA-00}) can be handled similarly. We find 
 \[ \text{II} = \int dx dy N^2 V (N(x-v)) \ph_0 (x) \wt{k}_H (x;y) \big( b^* (\gamma_{L,x}) + b (\sigma_{L,x}) \big)  (b(\gamma_{L,y}) + b^* (\sigma_{L,y})) + \cE_2 \]
with an error $\cE_2$ satisfying the same estimate as $\cE_1$. To bound $\text{III}$ we can proceed similarly. We obtain 
\[ \pm \text{III} \leq \frac{C}{\sqrt{N}}  \big[ \cV_N + (\cN+1)^2 \big]. \]
The contribution arising from $A-\wt{A}$ can be handled as indicated in (\ref{eq:A-wtA}), with the bounds in (\ref{eq:rho}). We obtain 
\begin{align*}
\vert \langle \xi, [\Theta_0+ \Theta_0^*, A-\wt{A}] \xi \rangle \vert \leq \frac{C}{N^{3/2}} \Vert \cV_N^{1/2} \xi \Vert \Vert \cN^2 \xi \Vert.
\end{align*}
This proves (\ref{eq:ThA}). 

To show (\ref{eq:PiA}), we can use Lemma \ref{lemma:conjquadA}. To this end, we just need to observe that the kernel $s (x;y) := N^2 V (N(x-y)) (\ph_0 (y) + \ph_0 (x)) |\wt{k}_H (x;y)|$ is such that $\sup_x \int dy |s (x;y)| < \infty$, uniformly in $N$. This follows readily from $\vert \wt{k}_H(x,y) \vert \leq CN\pn(y).$ 
\end{proof}

Combining Lemma \ref{lm:KV-A} with Lemma \ref{lm:gronA-KV} and Lemma \ref{lm:theta0}, we can now compute precisely the action of $A$ on the operator $\cK+ \cV_N$. 
\begin{cor} \label{cor:HA-conj}
Let $0<6\exps\leq \exph\leq 1/2$, then we have 
\[ e^{-A} (\cK + \cV_N) e^A = \cK + \cV_N + \Theta_0 + \Theta_0^* + \frac{1}{2} \left[ \Pi_0 + \Pi_0^* \right] + \cE \]
where
\[ \pm \cE \leq C N^{\exph -1/2} \left[ (\cV_N + \cK) (\cN+1)^2 + (\cN+1)^4 \right]. \]
\end{cor}

\subsection{Analysis of $e^{-A} \cC_{\cG_N} e^A$}

\begin{lemma}\label{lm:Xi} 
Let 
\begin{equation} \label{eq:Xi0-def} \begin{split} \Xi_0 = \; & \int dx dy \, N^2 V(N(x-y)) \pn(y) \, \wt{k}_H (x;y) \left( b^* (\gamma_{L,x}) + b (\sigma_{L,x}) \right) \left( b(\gamma_x) + b^*(\sigma_x) \right) \\
\; &+ \int dx dy \, N^2 V(N(x-y)) \pn(x) \, \wt{k}_H (x;y) \left( b^* (\gamma_{L,x}) + b (\sigma_{L,x}) \right) \left( b(\gamma_y) + b^*(\sigma_y) \right). \end{split} \end{equation} 
Then we have 
\begin{align*}
\big[ \cC_{\cG_N} , A \big] = \Xi_0 + \Xi_0^* + \cE \end{align*} 
where 
\[ \pm \cE \leq \frac{C}{\sqrt{N}} \big[ \cV_N + ( \cN+1)^2 \big]. \]
Moreover, 
\begin{align*}
\pm \big[ \Xi_0 + \Xi_0^*, A \big]  \leq C N^{- 1/2} (\cN+1)^2. \end{align*} 
\end{lemma} 

Since the operator $\cC_{\cG_N}$ is very similar to $\Theta_0 + \Theta_0^*$, with $\Theta_0$ as defined in (\ref{eq:defTheta0}), Lemma \ref{lm:Xi} can be shown very similarly to Lemma \ref{lm:theta0}. We skip the details.  
With Lemma \ref{lm:Xi} (and using Lemma \ref{lm:gronA-KV}) to bound the growth of $\cV_N$), we can control 
the action of $A$ on $\cC_{\cG_N}$.   
\begin{lemma} \label{lm:conj-CN} 
We have
\[ e^{-A} \cC_{\cG_N} e^A = \cC_{\cG_N} + \Xi_0 + \Xi_0^* + \cE \]
where 
\[ \pm \cE \leq \frac{C}{\sqrt{N}}  \big[ \cV_N + (\cN+1)^2 \big]. \]
\end{lemma}

\subsection{Proof of Prop. \ref{prop:JN}} 

Combining (\ref{eq:propGN}) with Lemma \ref{lm:QA-conj} , Lemma \ref{lm:Vext-conj}, Lemma \ref{lm:KV-A}  and Lemma \ref{lm:conj-CN}, we conclude that 
\begin{equation}\label{eq:JN-dec}
\cJ_N = e^{-A} \cG_N e^A = \kappa_{\cG_N} + \cQ_{\cG_N} + \cC_{\cG_{N}} + \cV_N + \Theta_0 + \Theta_0^* + \Xi_0 + \Xi_0^* + \frac{1}{2} \big[ \Pi_0 + \Pi_0^* \big] +  \cE 
\end{equation} 
where
\[ \pm \cE \leq C N^{\exph-1/2} \big[ (\cK + \cV_\text{ext} + 1) (\cN+1)^2 + \cV_N \big]. \]

From (\ref{eq:cCcGN}) and (\ref{eq:theta0-def}), we find (with $\delta$ denoting Dirac's delta-distribution) 
\begin{equation} \label{eq:CTheta} \begin{split}  \cC_{\cG_N} &+ \Theta_0 + \Theta_0^* \\ = \; &\int dx dy \, N^{5/2} V (N(x-y)) \ph_0 (y) b_x^* b_y^*  \big[ b (\gamma_x * (\delta- \check{g}_L)) + b^* (\sigma_x * (\delta- \check{g}_L)) \big]  \\ &+ \text{h.c.} \end{split} \end{equation} 
 Thus, we find 
 \begin{equation}\label{eq:CXi} \begin{split} \Big| \langle \xi, &\big( \cC_{\cG_N} + \Theta_0 + \Theta_0^* \big) \xi \rangle \Big| \\ &\leq C \| \cV_N^{1/2} \xi \|  \, \Big\{ \int  dx dy N^3 V (N (x-y))  \\ &\hspace{1.5cm} \times \big[  \| b ( \gamma_x * (\delta- \check{g}_L) ) \xi \|^2 + \| b^* (\sigma_x * (\delta- \check{g}_L) )\xi \|^2 \big] \Big\}^{1/2}. \end{split} \end{equation} 
 Decomposing $\gamma = 1 + p_\eta$, we find 
 \[  \| b ( \gamma_x * (\delta- \check{g}_L) ) \xi \|^2 \leq \| b ((\delta - \check{g}_L)_x) \xi \|^2 + \| b (p_{\eta,x} * (\delta- \check{g}_L)) \xi \|^2. \]
 Switching to Fourier space, we estimate 
 \[  \int dx \| b ((\delta - \check{g}_L)_x) \xi \|^2 =\int dp \  (1-g_L (p))^2 \| \hat{b}_p \xi \|^2 \leq CN^{-2\exps} \| \cK^{1/2} \xi \|^2. \]
 On the other hand,
 \begin{equation}\label{eq:peta} \| b (p_{\eta,x} * (\delta- \check{g}_L)) \xi \|^2 \leq \| p_{\eta,x} *  (\delta- \check{g}_L) \|^2 \|  \cN^{1/2} \xi \|^2 \end{equation} 
 and 
 \[ \begin{split} \int dx \, \| p_{\eta,x} *  (\delta- \check{g}_L) \|^2 &= \int dx dy \, \Big| \int dz \, \big[ p_{\eta} (x;y) - p (x;z) \big] \check{g}_L (y-z) \Big|^2 \\ &\leq \int dx dy dz \big[ p_{\eta} (x;y) - p (x;z) \big]^2 \check{g}_L (y-z) \\ &\leq   \int_0^1 ds \int dx dy dz \big|(\nabla_2 p)_{\eta} (x; sy + (1-s)z) \big|^2  (y-z)^2 \check{g}_L (y-z) \\ &\leq C N^{-2\gamma} \| \nabla_2 p_\eta \|
 \end{split} \]
where, in the last step, we shifted the integration variable $y\to y+z$. 
To bound the second term in the parenthesis on the r.h.s. of (\ref{eq:CXi}), we decompose $\sigma_x = k_x + (\sigma -k)_x$. The contribution proportional to $(\sigma-k)_x$ can be controlled similarly as in (\ref{eq:peta}). To control the contribution proportional to $k_x$, on the other hand, we switch to Fourier space. Using Lemma \ref{3.0.sceqlemma} and \eqref{eq:expdecaypn} we get  $\vert \hat{k}_x (p)\vert \leq C\vert p \vert^{-2} \pn(x)$
and therefore
\begin{align*}
\int dx \ \Vert k_x *(\delta- \check{g}_L) \Vert^2 &= \int dx dp \ \vert \hat{k}_x (p) \vert^2 (1-g_L(p))^2 \\
&\leq C\int_{\vert p \vert < N^{-\exps}} dp \  \frac{1}{\vert p \vert^4}  N^{-2\exps} \vert p \vert^2 + C\int_{\vert p \vert \geq N^{-\exps}} dp \ \frac{1}{\vert p \vert^4} \\
&\leq C N^{-\exps}.
\end{align*}
From (\ref{eq:CXi}), Lemma \ref{lm:growA} and Lemma \ref{lm:gronA-KV}, we conclude that
\[ \pm \big[ \cC_{\cG_N} + \Theta_0 + \Theta_0^* \big]  \leq C N^{-\tau/2} (\cV_N + \cK+ \cN + 1). \]
Finally, we compare the terms $\Xi_0, \Pi_0$ on the r.h.s. of (\ref{eq:JN-dec}) with the operator 
\[ \begin{split} Z_0 = \; & -\int dx dy N^3 V(N(x-y)) \ph^2_0 (y) w_\ell (N(x-y)) \big( b^* (\gamma_x) + b (\sigma_x) \big) \big( b (\gamma_x) + b^* (\sigma_x) \big) \\ &-  \int dx dy N^3 V(N(x-y)) \ph_0 (x) \ph_0 (y)  w_\ell (N(x-y)) \big( b^* (\gamma_x) + b (\sigma_x) \big) \big( b (\gamma_y) + b^* (\sigma_y) \big) \end{split} \]
where we removed all cutoffs (recall the definition (\ref{eq:Pi0-def}) of $\Pi_0$ and (\ref{eq:Xi0-def}) of $\Xi_0$, with the kernels $\wt{k}_H$ and $\gamma_L, \sigma_L$ introduced in (\ref{eq:def-wtk}), (\ref{eq:def-gsL})). The restriction to small momenta inserted in the kernels $\gamma_L, \sigma_L$ can be removed 
similarly as in (\ref{eq:CTheta}) (here, the comparison is a bit easier, because we are dealing with quadratic, rather than cubic, expressions in creation and annihilation operators). To remove the restriction to high momenta in the kernel $\wt{k}_H$, we notice that, 
we have $\vert (\wt{k}-\wt{k}_H)(x,y)\vert \leq CN^{\exph}\pn(y)$ by \eqref{eq:diffk}.
We conclude that 
\[ \pm \big[ (\Xi_0 + \Xi_0^*) - (Z_0 + Z_0^*) \big] \leq  C N^{-\exps/2} (\cN+1) + C N^{-\exps}  \cK  \]
and that 
\[ \pm \big[ (\Pi_0 + \Pi_0^*) + (Z_0 + Z_0^*) \big] \leq  C N^{-\exps/2} (\cN+1) + C N^{-\exps} \cK. \]
Thus
\[ \cJ_N = \kappa_{\cG_N} + \cQ_{\cG_N} + \frac{1}{2} \big[ Z_0 + Z_0^* \big]  + \cV_N + \cE  \]
where 
\[ \pm \cE \leq C N^{\exph-1/2} \big[ (\cK+\cV_\text{ext} +1) ( \cN+1)^2 \big] + C N^{-\exps/2} \big[ \cK + \cV_N + \cN+ 1 \big]. \]
The proposition follows by noticing that 
\[ \kappa_{\cG_N} + \cQ_{\cG_N} +  \frac{1}{2} \big[ Z_0 + Z_0^* \big]  = \kappa_{\cJ_N} + \cQ_{\cJ_N} + \wt{\mathcal{E}} \]
with $\kappa_{\cJ_N}$ and $\cQ_{\cJ_N}$ defined as in (\ref{eq:kJN}) and (\ref{eq:cQcJN}) and $\pm \wt{\mathcal{E}}\leq CN^{-1}\cN$ being an error term arising from normal ordering.


\appendix
\section{Properties of the Gross-Pitaevskii Functional}\label{apx:gpfunctional}

In this appendix we collect several well-known results about the Gross-Pitaevskii functional $\cE_{GP}$, defined in equation \eqref{eq:defGPfunctional}. Let us recall that $\cE_{GP}:\cD_{GP}\to \mathbb{R}$ is given by
		\begin{align*} 
		 \mathcal{E}_{GP}(\varphi) = \int_{\bR^3} \left( \vert \nabla \pn(x) \vert^2 + V_\text{ext}(x) \vert \pn(x)\vert^2 + 4\pi\frak{a}_0 \vert \pn(x)\vert^4 \right) dx 
		\end{align*}
with domain
		 \[ \mathcal{D}_{GP} =\big\{ \varphi \in H^1(\bR^3)\cap L^4(\bR^3): \ V_\text{ext}\vert \varphi\vert^2\in L^1(\bR^3)\big\}. \]
Recall, moreover, assumption $(2)$ in Eq. \eqref{eq:asmptsVVext} on the external potential $V_\text{ext}$. The following was proved in \cite[Theorems 2.1, 2.5 \& Lemma A.6]{LSY}.

\begin{lemma} \label{thm:gpmin1}
	There exists a minimizer $\pn\in \cD_{GP}$ with $ \|\pn\|_2=1$ such that
			\[ \inf_{\psi\in \mathcal{D}_{GP} \ : \|\psi\|_2 =1} \mathcal{E}_{GP}(\psi) = \mathcal{E}_{GP}(\pn).  \]
	The minimizer $\pn$ is unique up to a complex phase, which can be chosen so that $\pn$ is strictly positive. Furthermore, the minimizer $\pn$ solves the Gross-Pitaevskii equation
	\begin{align*} 
	-\Delta \pn + V_\text{ext} \pn + 8\pi\frak{a}_0 \vert \pn \vert^2 \pn = \eps_{GP}\pn,
	\end{align*}
with $\mu$ given by
	$$ \eps_{GP}= \mathcal{E}_{GP}(\pn) + 4\pi\frak{a}_0 \Vert \pn \Vert_4^4.$$
Moreover, $\pn \in L^\infty(\bR^3)\cap C^1(\bR^3)$ and for every $\nu>0$ there exists $C_\nu$ (which only depends on $\nu$ and $\mathfrak{a}_0$) such that for all $x\in \bR^3$ it holds true that
	\begin{equation} \label{eq:expdecaypn}
	\vert \pn(x)\vert \leq C_\nu e^{-\nu \vert x \vert}.
	\end{equation} 
\end{lemma}
As in the main text, $ \pn$ denotes the unique, strictly positive minimizer of $\cE_{GP}$, subject to the contraint $ \|\pn\|_2=1$. In addition to Lemma \ref{thm:gpmin1}, we collect some additional facts about the regularity of $\pn$. The following was shown in \cite[Theorem A.2]{BSS}. 
\begin{lemma}\label{thm:gpmin2}
Let $ V_\text{ext}$ satisfy the assumptions in \eqref{eq:asmptsVVext}. Then $\pn\in H^2(\mathbb{R}^3)\cap C^2(\mathbb{R}^3)$ and for every $\nu>0$ there exists $C_\nu>0$ such that for every $x\in \mathbb{R}^3$ we have
	\begin{equation} \label{exponential decay}
	\vert \nabla \pn (x) \vert, \ \vert \Delta \pn(x) \vert, | \nabla \Delta \pn(x)| \leq C_\nu e^{-\nu \vert x \vert}.
	\end{equation}
Furthermore, if $\whpn$ denotes the Fourier transform of $\pn$, we have for all $ p\in\bR^3$ that
		\begin{equation} \label{eq:decphhat}
		| \whpn (p)|,  | \widehat{\pn^2} (p)|    \leq  \frac{C}{(1+ | p|)^3}.
		\end{equation}
\end{lemma}



\end{document}